\pdfoutput=1
\documentclass[english,11pt]{article}

\usepackage[utf8]{inputenc}
\usepackage{amsmath}
\usepackage{amssymb}
\usepackage{amsthm}
\usepackage{fullpage}
\usepackage{xcolor,mathtools}

\usepackage[lined,boxed,ruled,norelsize,algo2e]{algorithm2e}

\usepackage{thmtools}
\usepackage{thm-restate}

\usepackage{hyperref}
\hypersetup{
	colorlinks,
	allcolors=blue
}
\usepackage{cleveref}

\newtheorem{theorem}{Theorem}[section]
\newtheorem{lemma}[theorem]{Lemma}

\newtheorem{definition}[theorem]{Definition}
\newtheorem{claim}[theorem]{Claim}

\newcommand{\Tr}{\mathrm{tr}}

\newcommand{\mlap}{\mathcal{L}}
\newcommand{\otilde}{\tilde{O}}
\newcommand{\vol}{\mathrm{Vol}}
\newcommand{\defeq}{\coloneqq} 

\newcommand{\R}{\mathbb{R}}
\newcommand{\Z}{\mathbb{Z}}

\newcommand{\ma}{\textbf{A}}
\newcommand{\mb}{\textbf{B}}
\newcommand{\mm}{\textbf{M}}
\newcommand{\mvar}[1]{\textbf{#1}}

\newcommand{\eye}{\textbf{I}}
\newcommand{\loglog}{\mathrm{log log}}
\newcommand{\norm}[1]{\left\lVert#1\right\rVert}
\newcommand{\normInline}[1]{\lVert#1\rVert}

\newcommand{\vones}{\vec{1}}
\newcommand{\vzero}{\vec{0}}

\newcommand{\Reff}{\mathcal{R}^{eff}}

\newcommand{\poly}{\mathrm{poly}}
\newcommand{\dmin}{d_{\min}}
\newcommand{\dmax}{d_{\max}}

\newcommand{\cutset}{\partial}
\newcommand{\condvert}{\phi_{\mathrm{vert}}}
\newcommand{\condedge}{\phi_{\mathrm{edge}}}
\newcommand{\Vol}{\mathrm{Vol}}

\newcommand{\indicdiff}{\delta}
\newcommand{\indicvec}{1}
\newcommand{\effres}{\mathcal{R}^{\mathrm{eff}}}

\newcommand{\davg}{d_{\mathrm{avg}}}
\newcommand{\dratio}{d_{\mathrm{ratio}}}

\newcommand{\wmin}{w_{\min}}
\newcommand{\wmax}{w_{\max}}
\newcommand{\wratio}{w_{\mathrm{ratio}}}

\renewcommand{\ln}{\log}

\newcommand{\pseudo}{\dagger}
\newcommand{\E}{\mathbb{E}}

\newcommand{\ApproxAGD}{\mathsf{PreconditionedNoisyAGD}}

\newcommand{\im}{\mathrm{im}}

\begin{document}
	
\pagenumbering{gobble}
\title{Ultrasparse Ultrasparsifiers and Faster Laplacian System Solvers}
\author{Arun Jambulapati\\
	Stanford University\\
	\texttt{jmblpati@stanford.edu}\and 
	Aaron Sidford\\
	Stanford University \\
	\texttt{sidford@stanford.edu}
}
\maketitle

\begin{abstract}
In this paper we provide an $O(m \loglog^{O(1)} n \log(1/\epsilon))$-expected time algorithm for solving Laplacian systems on $n$-node $m$-edge graphs, improving improving upon the previous best expected runtime of $O(m \sqrt{\log n} \loglog^{O(1)} n \log(1/\epsilon))$ achieved by (Cohen, Kyng, Miller, Pachocki, Peng, Rao, Xu 2014). To obtain this result we provide efficient constructions of $\ell_p$-stretch graph approximations with improved stretch and sparsity bounds. Additionally, as motivation for this work, we show that for every set of vectors in $\mathbb{R}^d$ (not just those induced by graphs) and all $k > 1$ there exist ultrasparsifiers with $d-1 + O(d/\sqrt{k})$ re-weighted vectors of relative condition number at most $k$. For small $k$, this improves upon the previous best known relative condition number of $\tilde{O}(\sqrt{k \log d})$, which is only known for the graph case.

\end{abstract}
\pagebreak{}

\pagenumbering{arabic}
\setcounter{page}{1}

\maketitle

\tableofcontents
\newpage

\section{Introduction}

From the first proof of a nearly linear time Laplacian system solver \cite{SpielmanT04}, to the current state-of-the-art running time for Laplacian system solving \cite{CohenKMPPRX14}, to advances in almost linear time approximate maximum flow \cite{Sherman13, KelnerLOS14, Peng16} ultrasparsifiers have played a key role in the design of efficient algorithms. In \cite{SpielmanT04} ultrasparsifiers were used in the computation of sequences of graph preconditioners that enabled nearly linear time Laplacian system solvers. This initiated a long line of work on faster \cite{KMP11,KMP14,CohenKMPPRX14}, simpler \cite{LS13, KOSZ, KS16}, and more parallel \cite{BGKMPT14, PS14,KyngLPSS16} Laplacian system solvers, many of which leverage ultrasparsifiers or related sparse graph approximation, e.g. low stretch spanning trees. These results in turn fueled advances of graph decompositions for a range of problems including approximate maximum flow \cite{Sherman13, KelnerLOS14, Peng16}, directed Laplacian solving \cite{CKPPRSV17, CKPPSV16}, and transshipment \cite{Li20}.

The current fastest Laplacian solver \cite{CohenKMPPRX14} computes expected $\epsilon$-approximate solutions to Laplacians on $n$-node, $m$-edge graphs in time\footnote{Here and throughout this paper, the notation $\otilde(\cdot)$ hides $\loglog$ factors.} $\otilde(m \sqrt{\log n} \log(1/\epsilon))$. These solvers start from the graphs associated with Laplacians and compute randomized tree-based approximations with bounded expected $\ell_p$-stretch: they use these to construct a sequence of preconditioners that efficiently decrease the error in expectation. The ultimate runtime achieved by this approach, $\otilde(m \sqrt{\log n} \log(1/\epsilon))$, matches that of the runtime one would achieve if the best known ultrasparsfiers for graphs, due to \cite{KMST}, could be constructed in linear time and then preconditioning approaches related to \cite{KMP11} were applied. Though it is known that preconditioners exist that would enable an $\otilde(m) \log(1/\epsilon)$ time solver, the current best construction of such preconditioners takes $O(m \poly(\log(n)))$ time due to the need to compute linear sized sparsifiers \cite{BSS} of Schur complements. 

Consequently, the best known bounds of ultrasparsifiers for graphs due to \cite{KMST} constitute a fundamental barrier towards designing faster Laplacian system solvers. \cite{KMST} showed that arbitrary $n$-vertex graphs possess $\otilde(k \log n)$-spectral approximations with $n + \frac{n}{k}$ edges, and their proof is based on the existence of \emph{low-stretch spanning trees} with average stretch $\otilde(\log n)$. In turn \cite{CohenKMPPRX14} achieves their running times by leveraging that $\ell_p^p$ stretch variants of these trees can be computed in $\otilde(m)$ time. These methods all pay this $\log n$ factor due to the best known bounds combinatorial techniques for ball-growing and graph decomposition. It is known that the stretch bound of  $\otilde(\log n)$ is optimal up to an iterated logarithmic factor in a wide variety of graphs, such as the $n$-vertex grid or hypercube. In other words, if using just distance based combinatorial graph decomposition, one must pay a factor of $\Theta(\log n)$ in the worst case. This factor then appears in the best known ultrasparsifier bounds and as a $\sqrt{\log n}$ factor in the current best Laplacian solver runtimes, due to the nature of iterative methods for solving Laplacian systems.

The main conceptual contribution of this paper is that this barrier can be broken and this $\Theta(\log n)$ factor can, perhaps surprisingly, be avoided. As a quick, broad proof-of-concept, in  \Cref{sec:ultrasparsifiers} we show that \cite{KMST} is not optimal in all parameter regimes. For arbitrary matrices and small target distortion we show that there exist sparser ultrasparsifiers that do not pay this $\Theta(\log n)$ factor. Interestingly, we give a simple \cite{BSS} based argument that arbitrary matrices have low-stretch subgraphs and then we apply the arguments of  \cite{KMST}  to get our bounds.

Inspired by this proof of concept, the main technical contribution of this paper is to show that for the specific goal of constructing low-distortion spectral subgraphs, i.e. those that would suffice for Laplacian system solving, better bounds can be achieved and the $\Theta(\log n)$-factor can be avoided. We carefully combine both spectral methods and combinatorial decomposition techniques for this purpose. In particular, we show that traditional ball growing techniques can be augmented or patched by careful use of spectral sparsifiers to achieve lower distortion graph approximations. 

Interestingly, our procedure for augmenting a low-diameter decomposition requires us to efficiently compute stronger notions of spanners in sufficiently dense graphs. We provide an efficient procedure for computing a type of graph approximation related to fault-tolerant spanners, which we call \emph{path spanners}. In turn, to compute path sparsifiers we show that there are many short vertex disjoint paths in a dense near-regular expander. This proofs builds upon the seminal work of \cite{KleinbergR96}  which showed that every dense expander has many short edge-disjoint paths. We leverage this fact in an algorithm which combines a near-linear time expander decomposition procedure of \cite{SW19} with a new routine for approximately finding regular dense subgraphs. The resulting path sparsification algorithm serves as a type of vertex-based sparsification in our low-distortion subgraph computation algorithm. It is an interesting open problem if an alternative sparsification procedure can be used instead, however we think the tools developed for obtaining path sparsifiers may be of intrinsic interest.

Ultimately, we show that careful application of this routine for constructing low distortion spectral subgraphs yields an $\otilde(m \log(1/\epsilon))$-time algorithm for computing expected $\epsilon$-approximate solutions to Laplacian systems. This solver leverages heavily recursive preconditioning machinery of \cite{CKPPRSV17} and our efficiently computable spectral subgraphs. To simplify and clarify the derivation and analysis of this recursive solver we provide an analysis of a stochastic preconditioned variant of accelerated gradient descent (AGD) \cite{Nesterov1983}.

We hope that this work may serve as the basis for further improvements in ultrasparsification and graph decomposition. Given the myriad of applications of these techniques and the simplicity and generality of our approach for overcoming the $\Theta(\log n)$-factor in previous combinatorial approaches we hope this work may find further applications.

\paragraph{Paper Organization} In the remainder of this introduction we provide  preliminaries and notation we use throughout the paper (\Cref{sec:prelim}), our main results (\Cref{sec:results}), our approach for achieving them (\Cref{sec:approach}), and provide a brief discussion of previous work (\Cref{sec:previous_work}). In \Cref{sec:spectral_distort_and_solve} we then give our main graph decomposition and use it to obtain low-distortion spectral subgraphs. The results of this section hinge on the efficient construction of a new combinatorial object known as \emph{path sparsifiers}, which we compute efficiently in \Cref{sec:path_sparsify}. In \Cref{sec:solver} we leverage our low-distortion spectral subgraph construction to obtain our Laplacian system solver results. Our existence proof for ultrasparsifiers is given briefly in \Cref{sec:ultrasparsifiers}, our AGD analysis is given in \Cref{sec:chebyshev}, and additional proofs are given in the other appendix sections.


\subsection{Preliminaries}
\label{sec:prelim}

Here we provide notation and basic mathematical facts we use throughout the paper.

\textbf{Graphs}: Throughout this paper we let $G = (V, E, w)$ denote an undirected graph on vertices $V$, with edges $E \subseteq V \times V$, with integer positive edge weights $w \in \Z^E_{> 0}$.  Though many graphs in this paper are undirected, we typically use $(a,b) \in E$ notation to refer to an edge where we suppose without loss of generality that a canonical orientation of each edge has been chosen. Often in this paper we consider unweighted graphs $G = (V, E)$ where implicitly $w \defeq \vones$. Unless stated otherwise (e.g. much of \Cref{sec:path_sparse}) we make no assumption about whether graphs are simple in this paper and often consider graphs with multi-edges and self-loops.

\textbf{Degrees and Weights}: For graph $G = (V, E, w)$ and $a \in V$ we let $\deg_G(a) \defeq \sum_{e \in E | a \in e} w_e$. Further, we let $\dmin(G) \defeq \min_{a \in V} \deg_G(a)$, $\dmax(G) \defeq \max_{a \in V} \deg_G(a)$, $\dratio(G) \defeq \dmax(G) / \dmin(G)$ and $\davg(G) \defeq \sum_{a \in V} \deg_G(a) / |V|$. For weighted graph $G = (V, E, w)$ we let $\wmin(G) \defeq \min_{e \in E} w_e$, $\wmax(G) \defeq \max_{e \in E}$, and $\wratio(G) \defeq \wmax(G) / \wmin(G)$.

\textbf{Volumes}:  For graph $G = (V,E,w)$ and $S \subseteq V$ we let $\vol_G(S) \defeq \sum_{a \in S} \deg_G(a)$
 where $w \in \R^E$ are the edge weights of the graph. Note that when we allow self-loops, if a vertex has self-loops of total weight $w$ this contributes $w$ to the degree of that vertex and the volume of any set it is in. We further define the boundary volume $\vol_G(\partial S) \defeq \sum_{(u,v) \in E(G), u \in S, v \notin S} w_{(u,v)}$. 

\textbf{Distances and Shortest Path Balls}: For graph $G = (V, E, w)$ and path $P \subseteq E$ between vertices $a$ and $b$ we let $\ell(P) \defeq \sum_{e \in P} w_e$ denote the length of the path and we let $d_G(a,b)$ denote the length of the shortest path between $a$ and $b$. Further, we let $B_G(v,r) \defeq \{ a \in V ~ | ~ d_G(v,a) \leq r \}$ denote the (shortest-path) ball of distance $r$ from $v$.

\textbf{Neighbors}: For graph $G$ and vertex set $S \subseteq V$ we let $N(S) \defeq \{a \in V ~ | ~ (b,a) \in E \text{ for some } b \in S  \}$ denote the \emph{neighbors of $S$}. Overloading notation we let $N(a) \defeq N(\{a\})$ for all $a \in V$. 

\textbf{Subgraphs and Contractions}: For graph $G = (V, E, w)$ and $S \subseteq V$ we let $G[S]$ denote the subgraph induced by $S$, i.e. the graph with vertices $S$, edges  $E \cap (S \times S)$, and edge weights the same as in $G$. We overload notation and, similarly, for $F \subseteq E$ we let $G[F]$ denote the subgraph of $G$ induced by edge set $F$, i.e. the graph with vertices $V$, edges $F$, and edge weights the same as in $G$. For $S \subseteq V$, we let $G \backslash S$ denote the graph by contracting all vertices in $S$ into a single supernode, while preserving multi-edges and possibly inducing self-loops. 

\textbf{Graph Matrices}: For weighted graph $G = (V, E, w)$ we let $\mlap_G \in \R^{V \times V}$ denote its Laplacian matrix where for all $a,b \in V$ we have $\mlap_{a,b} = - w_{\{a,b\}}$ if $\{a,b\} \in E$ and $\mlap_{a,a} = \deg_G(a)$. 

\textbf{Effective Resistances}: For graph $G = (V,E,w)$ and nodes $u,v$, we say the effective resistance between $u$ and $v$ is $\mathcal{R}_{G}^{eff}(u,v) = (\mathbf{e}_u - \mathbf{e}_v)^\top \mlap_G^\dagger (\mathbf{e}_u - \mathbf{e}_v)$ where $\mlap_G^\dagger$ denotes the Moore-Penrose pseudoinverse of $\mlap_G$. Throughout the paper, we make use of several standard facts about effective resistances stated in \Cref{sec:effres_facts}.  

\textbf{Solver}: We use the following notation for linear system solvers:

\begin{definition}[Approximate Solver]
	We call a randomized procedure an \emph{$\epsilon$-solver for PSD $\ma \in \R^{n \times n}$} for $\epsilon \in [0,1)$ if given arbitrary $b \in \R^n$ it outputs random $x \in \R^n$ with 
	\begin{equation}
	\label{eq:solve_requirement}
	\E \norm{x - \ma^\pseudo b}_{\ma}^2 \leq \epsilon \norm{b}_{\ma^\pseudo}^2
	\end{equation}
\end{definition}

\begin{definition}[Laplacian Solver] 
	We call a randomized procedure an \emph{$\epsilon$-Laplacian (system) solver} for graph  $G = (V,E,w)$ if it is an $\epsilon$-solver for $\mlap_G$, i.e. given arbitrary $b \in \R^n$ it outputs random $x \in \R^n$ with 
	$\E [\normInline{x - \mlap_G^\dagger}_{\mlap_G}^2] \leq \epsilon \normInline{\mlap_G^\dagger b}_{\mlap_G}^2$. 
	\label{def:solver}
\end{definition}

Error guarantees in the $\ma$ and $\mlap_G$ norm are standard to the literature; they corresponds to an $\epsilon$-multiplicative decrease in the function error on the objective $f(x) = (1/2) x^\top \ma x - b^\top x$ from initial point $\vzero$. However, that our solver error is defined with respect to the expected square norm of the matrices is less standard. By concavity of $\sqrt{\cdot}$ for appropriate choice of $\epsilon$ this guarantee is stronger than defining error in terms of just the norm: for all PSD $\ma \in \R^{n \times n}$ and vectors $x \in \R^n$ we have $[\E \norm{x}_{\ma}]^2 \leq \E \norm{x}_{\ma}^2$.

\textbf{Asymptotics and Runtimes}: Throughout we use $\otilde(\cdot)$ to hide $\poly(\log \log)$ factors in $n$, the number of vertices in the largest graph considered.

\textbf{Misc.} All logarithms in this paper are in base $e$ unless a base is explicitly specified.

\subsection{Our Results}
\label{sec:results}

Here we present the main results of our paper. First, as discussed in the introduction we provide new bounds on existence of ultrasparsifiers for arbitrary matrices. Our construction is based on the spectral-sparsification results of \cite{BSS}. We prove this existence result  briefly in \Cref{sec:ultrasparsifiers}:

\begin{restatable}[Ultrasparsifier Existence]{theorem}{bss}
\label{thm:ultrasparsifer_existence}
Let $v_1, \dots v_m \in \R^n$ and $\ma \defeq \sum_{i \in [m]} v_i v_i^\top$. For any integer $k \geq 2$,  there exists $S \subseteq [m]$ with $|S| = n + O\left(\frac{n}{\sqrt{k}}\right)$ and $w \in \R^m_{\geq 0}$ where
\[
\ma \preceq \sum_{i \in S} w_i v_i v_i^\top \preceq k \ma.
\]
\end{restatable}

This result when specialized to graphs immediately yields $ n + O\left(\frac{n}{\sqrt{k}}\right)$-edge subgraphs with relative condition number $k$, and is a proof of concept towards the main results of this paper. We obtain it by a two-stage construction: we first find an ultrasparse subset of vectors satisfying a certain ``on average" notion of spectral approximation, and then we correct this to a true ultrasparsifier with a procedure based on the spectral sparsification algorithm of \cite{BSS}. In the case of graphs, we give an improved guarantee for the first phase of our construction. We call the objects we compute low distortion spectral subgraphs, defined as follows.
\begin{definition}[$\kappa$-Distortion Spectral Subgraph]
	Given a weighted graph $G = (V, E, w)$ we call $H = (V, E_H, w_H)$ a \emph{$\kappa$-Distortion Spectral Subgraph} if $\mlap_H \preceq \mlap_G$ and 
	\[
	\sum_{e \in E} (w_e \effres_H(e)) = 
	\sum_{e \in E} \left(w_e \delta_e^\top \mlap_H^\dagger \delta_e\right) \leq \kappa ~.
	\]	
    If $H$ is a subgraph of $G$ in addition to the above, we call it a $\kappa$-distortion subgraph.
\end{definition}
To give context for this definition, observe that $H=G$ is an $n$-distortion spectral subgraph of $G$, and any strict subgraph of $G$ has spectral distortion strictly larger than $n$. In fact, we show something slightly stronger than what this definition encompasses. We show that there exist subgraphs with this guarantee that can be computed efficiently. Our main theorem for this construction (specialized for its application for Laplacian solvers) is as follows.

\begin{restatable}[Efficient Construction of Ultrasparse $\kappa$-Distortion Subgraphs]{theorem}{fastlpsubgraphs}
\label{thm:fast_lp_subgraphs}
Let $G = (V,E,w)$ be a polynomially-bounded weighted graph, and let $c \geq 1$ be any fixed constant. Algorithm~\ref{alg:akpw} equipped with Theorem~\ref{thm:path-sparse} runs in $O(m)$ time and returns a $\kappa$-distortion subgraph $H$ with $n + O\left(\frac{m}{\left(\log \log n\right)^c}\right)$ edges, for 
\[
\kappa = O\left(m \left(\log \log n\right)^{\sqrt{8c} + 1 + o(1)} \right)~.
\]
It also returns a vector $\mathbf{\tau} \in \R^E_{\geq 0}$ with $\norm{\tau}_1 \leq \kappa$ where for any $e \in G$, $\tau_e  \geq w_e \delta_e^\top \mlap_H^\dagger \delta_e$ is an overestimate of the leverage score of $e$ measured through $H$. 
\end{restatable}

We use this result to obtain preconditioners: employing a framework based on \cite{CohenKMPPRX14}, we thus obtain our main result on solving Laplacian linear systems. 
\begin{restatable}[$\tilde{O}(m)$-Laplacian System Solver]{theorem}{solver}
\label{thm:lap}
There is a randomized algorithm which is an $\epsilon$-approximate Laplacian system solver for any input $n$-vertex $m$-edge graph with polynomially-bounded edge weights (see \Cref{def:solver}) and $\epsilon \in (0,1)$ and has the following runtime for any $\chi > 0$ 
\[
O(m (\log \log n)^{6 + 2 \sqrt{10}+
\chi} \log(1/\epsilon) ) ~.
\]
\end{restatable}
We assume polynomially-bounded weights primarily for simplicity of presentation: it can be removed via standard techniques (see for instance \cite{CKPPRSV17} for details) Further, in Appendix~\ref{sec:ultrasparsifiers} show that techniques for constructing $\kappa$-distortion subgraphs yield ultrasparsifiers with the following improved guarantees over Theorem~\ref{thm:ultrasparsifer_existence}. 

\begin{restatable}[Improved Ultrasparsifiers]{theorem}{ultraexport}
\label{thm:ultra_export}
There exists a polynomial time algorithm which given an input graph $G$ with polynomially-bounded edge weights can compute a reweighted subgraph $H$ with either of the following guarantees:
\begin{itemize}
    \item For any constant $c$, $H$ has $n + \frac{n}{(\log \log n)^c}$ edges and satisfies 
    \[
    \mlap_G \preceq \mlap_H \preceq O((\log \log n)^{c+ \sqrt{8c} + 1 + o(1)}) \mlap_G.
    \]
    \item For any constant $\delta > 0$ and $\alpha = \omega(\log^\delta n)$, $H$ has $n + \frac{n}{\alpha}$ edges and  satisfies
    \[
    \mlap_G \preceq \mlap_H \preceq \alpha^{1+o(1)} \mlap_G.
    \]
\end{itemize}
\end{restatable}

When compared to the previous state-of-the-art ultrasparsifier algorithm \cite{KMST}, our construction provides improved spectral approximation qualities for sparsities up to $n + \frac{n}{\alpha}$ for $\alpha = O\left(\exp\left( \log^{1/2 - \delta} n \right) \right)$ for any $\delta > 0$, and improves upon Theorem~\ref{thm:bss-ultra} for $\alpha = \omega\left( \poly(\log \log n)\right)$. In particular, our method improves upon the previous-best methods in the important regime of $\alpha = \poly(\log n)$ by a factor of $O(\log^{1-o(1)} n)$. Ultrasparsifiers of this quality form a critical part of the current best-known approximate max flow algorithms (\cite{Sherman13, Peng16, Sherman17}), and thus we believe our techniques may be used to improve the running times of these methods.

To compute $\kappa$-distortion subgraphs efficiently we introduce a new combinatorial object we call a \emph{$(\alpha, \beta)$-Path Sparsifier}, defined as follows. 

\begin{definition}[$(\alpha, \beta)$-Path Sparsifiers] Given unweighted graph $G = (V, E)$ and $F \subseteq E$ we call subgraph $H = G[F]$ an \emph{$(\alpha,\beta)$-path sparsifer} if for all edges $(u,v) \in E$, either $(u,v) \in F$ or there are $\alpha$ vertex-disjoint paths of length at most $\beta$ from $u$ to $v$ in $H$, where we do not count $u, v$ as part of the paths. 
\end{definition}

Path sparsifiers provide a type of approximation for distance in unweighted graphs that is even stronger than that of fault-tolerant spanners \cite{DK11, BP19} and crucial for obtaining our linear system solving runtimes. We prove the following theorem regarding path sparsifiers.

\newcommand{\PathSparsify}{\mathsf{PathSparsify}}
\begin{restatable}[Efficient Path Sparsification]{theorem}{pathsparse} 
\label{thm:path-sparse}
Given any $n$-node, $m$-edge graph and parameter $k \geq 1$ the procedure, $\PathSparsify(G,k)$ (\Cref{alg:pathsparsifier}) outputs w.h.p. $F \subseteq E$ with $|F| = O(n k \log^3(n))$ such that $G[F]$ is a $(k, O(\log^5 n))$-path-sparsifier of $G$ in expected time $O(m + n k \log^{13}(n))$.
\end{restatable}

When used to prove Theorem~\ref{thm:lap}, we avoid the (large) polylogarithmic dependence of Theorem~\ref{thm:path-sparse} by applying it to graphs with $O(m)$ edges and $O(\frac{m}{\poly \log n})$ vertices. Thus the cost associated with computing path sparsifiers is $O(m)$: it does not affect our final $\otilde(m)$ runtime claim for Laplacian solvers. 

%
%
%
%
%
%

\subsection{Overview of Approach}
\label{sec:approach}

Here we provide a brief overview of our approach towards obtaining the results of \Cref{sec:results}. 

\subsubsection{Ultrasparse Low Distortion Subgraphs}
\label{sec:ultrasparse_low_distort}
Our techniques for computing ultrasparse low-distortion subgraphs are based on existing algorithms for computing low-stretch spanning trees: in particular we base our construction on a simple recursive procedure from \cite{AKPW} (with analysis insights from \cite{CohenKMPPRX14}). Given a graph $G$, our algorithm begins by partitioning its vertices into $V_1, V_2....$ such that the partition cuts few edges and each $G[V_i]$ has low diameter. With this, it then computes an ultrasparse subgraph inside each $G[V_i]$ such that each edge inside a $G[V_i]$ receives a small effective resistance overestimate when measured though the subgraph. It then recurses on a graph formed by appropriately contracting parts of each $G[V_i]$ and deleting all edges which lie inside a $V_i$. In the case where the graph computed inside each $G[V_i]$ is a shortest-path tree, the algorithm described is exactly the low-stretch spanning tree procedure of \cite{AKPW}. Our algorithm extends this classic result by adding a small number of extra edges within each $G[V_i]$ to improve the effective resistance overestimate: we use fast algorithms for path sparsifiers developed in Section~\ref{sec:path_sparsify} for this. Combining this with a more careful graph decomposition gives us our result. 

\subsubsection{Path Sparsifiers}
\label{sec:path_sparse}

Here we briefly outline our approach for proving \Cref{thm:path-sparse}, i.e. efficiently computing path sparsifiers. Recall that subgraph $H$ is an $(\alpha,\beta)$-path sparsifier of $G = (V, E)$ if every edge in $G$ is either present in $H$ or connected by $\alpha$ vertex-disjoint paths of length at most $\beta$. Consequently, constructing a path-sparsifier essentially involves replacing dense components of a graph with sparse subgraphs containing many short vertex disjoint paths. 

One natural starting point for construct sparse subgraphs with vertex disjoint paths is to consider expanders, i.e. (informally) graphs where the number of edges leaving every small enough vertex subset is some bounded fraction of the number of edges contained inside the subset. It is known that by seminal work of \cite{KleinbergR96} that in every sufficiently dense expander every pair of vertices is connected by many short \emph{edge-disjoint paths}. Consequently, our first step in constructing path sparsifiers is to show that in fact this result generalizes to graph that are good \emph{vertex expanders}: graphs where the number of nodes neighboring every small enough vertex subset is some bounded fraction of the number of nodes contained inside the subset. By a standard connection between vertex-expansion and edge-expansion, this immediately yields that every dense enough graph with nearly uniform degrees has many short paths between every pair of vertices. 

Given this primitive, our task of computing a path sparsifiers reduces to the problem of decomposing an arbitrary dense graph into subsets where we can find sparser expanders of nearly uniform degree. To achieve this we provide a procedure for decomposing an arbitrary dense graph into nearly degree uniform dense subgraphs, sample these subgraphs uniformly, and apply known expander decompositions to the result. We show that these expanders with high probability contain enough of the volume of the original graph that by repeating on the edges not contained in these expanders we ultimately obtain a path sparsifier. Further, by careful sampling and use of known nearly linear time expander decompositions, i.e. \cite{SW19}, our algorithm is time efficient as well and yields the desired \Cref{thm:path-sparse}.

\subsubsection{Laplacian System Solvers}
\label{sec:lap_solve}

Finally, we leverage the contributions of Sections~\ref{sec:spectral_distort_and_solve} and~\ref{sec:path_sparsify} to obtain our improved Laplacian solving algorithms. Our approach is a modification of the ``preconditioning in expectation'' framework used in \cite{CohenKMPPRX14} to obtain $\otilde(m \sqrt{\log n} \log(1/\epsilon))$ time Laplacian solvers. Given a graph $G$ we first compute a low-distortion subgraph $H$, and then aim to solve linear systems in $G' = G + (\eta - 1) H$ for some appropriately chosen $\eta$. We use a modified version of accelerated gradient descent \cite{Nesterov1983} (given in Appendix~\ref{sec:chebyshev} for completeness) to show that solving linear systems in $\mlap_G$ can be reduced to solving $O(\sqrt{\eta})$ linear systems in $\mlap_{G'}$. To solve linear systems in $G'$, we form a series of preconditioners $G'_i$ by sampling each edge $e$ in $G'$ with probability proportional to $w_e r_H(e)$, where $r_H(e)$ is the effective resistance overestimate of $e$ given by the copy of $\eta H$ contained in $G'$. While asking for these $G'_i$ to be true sparsifiers of $G'$ would require paying a logarithmic oversampling factor (and hence appear as an $O(\sqrt{\log n})$ in the runtime guarantee), an insight of \cite{CohenKMPPRX14} shows that sampling without this logarithmic factor still suffices to ensure that solving linear systems in a few randomly sampled $\mlap_{G'_i}$ enables one to solve linear systems in $G'$. Finally, via some parameter tradeoffs we can ensure that the $G'_i$ consist of a tree plus a small number of edges: we then apply a combinatorial contraction procedure to eliminate the vertices and edges of the tree and recursively apply our solver to the remaining graphs. 

We remark that our analysis of our recursion more closely resembles the original analysis of \cite{KMP11} which yields slower runtimes, and not the more sophisticated one given in \cite{CohenKMPPRX14}. Although this tighter analysis was necessary in the previous work to reduce the logarithmic dependence, applying the same techniques here would only reduce our algorithm's $\poly(\log \log n)$ dependence. Further, doing this introduces several technical issues which complicate the presentation of our algorithm. For the sake of clarity, we give the analysis which loses $\poly(\loglog n)$ factors in this paper and make only limited attempt to control the polynomial dependence on $\loglog n$ throughout the paper.

\subsection{Previous Work}
\label{sec:previous_work}

\paragraph{Sparsifiers:} Spectral graph sparsification has been heavily studied since its invention by \cite{SpielmanT04} in the process of constructing the first near-linear time Laplacian solver. While the original procedure was somewhat involved, a dramatic simplification by \cite{SS08} shows that sampling edges with probability proportional to their statistical leverage score gives a $(1+\epsilon)$-approximate sparsifier with $O(n \log n \epsilon^{-2})$ edges with high probability. \cite{BSS} showed that there is a more computationally expensive procedure that obtains sparsifiers with the same approximation quality and without the $\log n$ dependence in size. \cite{KMP11} extend this idea by showing that sampling edges with probability proportional to leverage score overestimates induced by a sparse subgraph also gives sparsifiers with high probability. They use this insight to construct \emph{ultrasparsifiers}: given an input graph $G$ they find $H$ with $n + \otilde(\frac{n \log^2 n}{k})$ edges where$ \mlap_H \preceq \mlap_G \preceq k \mlap_H$. Extending this result, \cite{KMST} improves the sparsity to  $n + \otilde(\frac{n \log n}{k})$ by replacing the random sampling of \cite{KMP11} with a procedure based on a sparsity-optimal algorithm by \cite{BSS}. However, progress on removing this final $\log n$ factor has stalled, partially due to a intrinsic barrier posed by the use of low-stretch spanning trees as a primitive. We bypass this barrier in two ways: we give a purely spectral argument which improves on the sparsity of \cite{KMST} in certain parameter regimes, and we give a procedure which constructs ultrasparse subgraphs with better leverage score overestimates than trees can provide. While our results stop just short of obtaining truly comparable ultrasparsifiers (due to our graph decomposition approach), we provide the first methods to bypass the low-stretch tree barrier present in the previous work. 

\paragraph{Graph Decomposition:} While our specific notion of a $\kappa$-distortion spectral subgraph has (to our knowlege) not been studied before specifically, many related spectral primitives have been considered in the literature. We remark that standard sparsification routines \cite{BSS} trivially give $O(n)$-distortion spectral subgraphs with $O(n)$ edges, and that a low-stretch spanning tree of a sparsifier gives a tree which is an $\otilde(n \log n)$-distortion spectral subgraph: these are the facts which inspire our definition. The previous fastest Laplacian solver \cite{CohenKMPPRX14} modifies this latter guarantee by providing a tree which is  a spectral subgraph satisfying a certain $\ell_p$ notion of distortion. They provide a construction of a tree where the sum of the $p^{th}$ powers of $w_e \effres_H(e)$ is bounded by $O\left(\frac{1}{(1-p)^2} n \log^p n\right)$, for any $0 < p < 1$.\footnote{We remark that  \cite{CohenKMPPRX14} does not phrase their guarantee in this way. The tree they compute has steiner vertices, and their graph actually has $\kappa = O(\frac{1}{(1-p)^2} m \log^p n)$. These issues can be eliminated by standard vertex elimination techniques in trees and initially sparsifying the input, respectively. Further, the dependence on $m$ in their claim has no impact on their final Laplacian solver algorithm.} Although this guarantee does not recover the standard low-stretch tree guarantee (which yields a bound for $p=1$), \cite{CohenKMPPRX14} gives an algorithm which computes this tree in $\otilde(m)$ time: this has not yet been achieved by more standard low-stretch spanning tree algorithms. Our work extends this result by giving $\otilde(m)$ time algorithms which trade off the sparsity of the output subgraph with $\kappa$. 

\paragraph{Laplacian System Solvers:} Our result on Laplacian system solvers draws on a long series of work on time-optimal Laplacian system solvers. Our specific approach draws heavily from ultrasparsifier-based algorithms as pioneered by \cite{SpielmanT04} and refined by \cite{KMP11,KMP14,CohenKMPPRX14}. These papers solve the Laplacian linear system $\mlap x = b$ by first recursively solving linear systems in $\mlap' x = b$, where $\mlap'$ is spectrally close to $\mlap$ but sparse. They then use this ability to solve linear systems in $\mlap'$ to precondition a conjugate gradient-type method. Although the specific way of constructing the ultrasparse $\mlap'$ has changed significantly over the previous line of work, all base their construction on \emph{low-stretch trees} (or dominating trees) trees which ``on average'' contain a short path across the endpoints of a randomly chosen edge from their base graph. Our departure from this line of work is to base our ultrasparsifiers on graphs which are merely ultrasparse: we instead start with a subgraph consisting of a tree with $o(m)$ edges. While the presence of this small number of extra edges may seem inconsequential, we demonstrate that this small overhead allows us to bypass a $O(m \sqrt{\log n})$ barrier present in all of the previous work following the ultrasparsifier archetype. 

As we remarked earlier, there is an alternative approach for solving Laplacian systems based on sparsification alone, e.g. \cite{PS14,KyngLPSS16,KS16}, which is known to yield preconditioners that (once computed) yield $\tilde{O}(m)$ Laplacian system solvers. However, constructing such preconditioners currently requires $\Omega(m \log^c n)$ time, where $c \geq 1$ derives from the need to compute $\tilde{O}(n)$-edge sparsifiers. In our solver, we too need efficient strong sparsification-like results, however we show that it is possible to use the path sparsifiers we provide for this purpose. 

Another approach proposed by \cite{KLP12} gives Laplacian solvers running in $O(m + n \log^{O(1)} n)$ time: this is $O(m)$ for any slightly dense graph. Their approach is based on computing coarse $O(\log^{O(1)} n)$-quality $n + o(m)$-edge sparsifiers in $O(m)$ time. They then compute effective resistance overestimates in their computed sparsifier in $O(m + n \log^{O(1)} n)$ time by leveraging low-stretch spanning tree algorithms, and they finally leverage these estimates to obtain $o(m)$-edge $O(1)$-sparsifiers to their original graph. By finally using existing near-linear time Laplacian solvers to solve in this approximation to the input graph, they are able to use a standard preconditioning approach to obtain their claimed runtime. While their approach does not yield  linear-time algorithms for graphs with $m < n \log^{O(1)} n$, we find it an interesting question to see if their techniques can be combined with ours to obtain $O(m + n (\log \log n)^{O(1)})$-time Laplacian solvers. 

\paragraph{Fault Tolerant Spanners} To computing our $\kappa$-distortion subgraphs, we construct efficient algorithms for $(\alpha,\beta)$-path sparsifiers. These are related to multipath spanners \cite{GGV11} and vertex fault-tolerant spanners \cite{DK11, BP19} studied by the combinatorial graph algorithm community. A $k$-fault tolerant spanner of input $G$ is a subgraph $H$ such that for any ``fault set" $F$ with $|F| \leq k$, $H - F$ is a spanner of $G - F$. Intuitively, such subgraphs must contain many short disjoint paths across the endpoints of any edge not retained from the parent graph: if there was a small set of ``bottleneck" nodes present in any short path between $u$ and $v$ for $(u,v) \in E(G)$, deleting these would mean $H - F$ no longer spanned $G-F$. Our definition extends this notion by additionally requiring these short paths to be vertex-disjoint. Our algorithm also departs from the previous work by using the spectral notion of expander decomposition to construct path sparsifiers, in a similar spirit to the independently developed ideas in \cite{BBGNSS} and in contrast to the random sampling approach of \cite{DK11} and the greedy approach of \cite{BP19}. While the use of expander decomposition comes with some significant drawbacks (most notably algorithm complexity), it enables us to obtain a linear dependence on the number of edges in our algorithm's runtime. 

\section{Low Distortion Spectral Subgraphs}
\label{sec:spectral_distort_and_solve}

In this section we prove \Cref{thm:fast_lp_subgraphs} showing that we can efficiently compute $(\kappa,p)$-distortion in $O(m \cdot \poly(\log \log n))$ time.  First, we provide a single-level graph decomposition result in \Cref{sec:decomposition}. Then, leveraging our efficient path-sparsification procedure of \Cref{sec:path_sparsify} we recursively apply our decomposition to provide our efficient construction of $(\kappa,p)$-distortion subgraphs in \Cref{sec:spectral_stretch}. Later, in \Cref{sec:solver} we use these in a solver framework related to that of \cite{CohenKMPPRX14} to obtain our claimed $O(m (\log \log n)^{O(1)} \log(1/\epsilon))$ time Laplacian solver.

\subsection{Graph Decomposition}
\label{sec:decomposition}

In this section we give our core combinatorial graph-decomposition technique, which we use to compute low distortion subgraphs. This single-level graph decomposition can be interpreted as a significant modification of low-diameter decomposition as originally conceived by \cite{awerbuch}. Broadly, our algorithm chooses an arbitrary vertex and grows a shortest-path ball out from it. Whenever the cut defined by the ball is sufficiently small relative to its volume, we cut the edges defined by the cut, mark the vertices of the ball as a partition piece, and repeat on the remaining vertices in the graph. This basic procedure, known as \emph{low-diameter decomposition},  has seen many applications in graph algorithms \cite{CohenKMPPRX14, MillerPX13,LSY19,Bartal98}. We state the guarantee here:
\begin{theorem}
Let $G = (V,E)$ be an unweighted graph and let $\beta > 0$ be a parameter. There is an algorithm which runs in $O(m)$ time and computes a partition of the vertices into $V_1, V_2 ... $ such that
\begin{itemize}
    \item Each $G[V_i]$ has diameter $O(\beta \log n)$
    \item At most $\frac{m}{\beta}$ edges  $G$ cross different partition pieces.
\end{itemize}
\end{theorem}
Unfortunately a significant limitation of the above procedure is the $O(\log n)$ factor in the diameter of the partition pieces. This is neccessary: for example a constant-degree unweighted expander graph has diameter $O(\log n)$ but any partition of $G$ into balls of diameter $o(\log n)$ must necessarily cut at least a constant fraction of its edges. 

We avoid this logarithmic factor by settling for a weaker guarantee that still suffices for Theorem~\ref{thm:fast_lp_subgraphs}. Our modification is this: after we have grown a ball $B_G(v,r)$ and made a cut, we ``retract'' the ball by distance $\delta$ and consider $B_G(v,r-\delta)$. The key insight is that although the vertices of $B_G(v,r)$ formed a low-conductance cut in $G$, the cut defined by $B_G(v,r')$ was sufficiently high conductance for any $r' < r$. With this we upper bound the size of $B_G(v,r-\delta)$, and consequently ensure that most vertices in $B_G(v,r)$ are close to a small number of nodes in $B_G(v,r-\delta)$. 

Unfortunately, the presence of weights in a graph somewhat complicates this picture, due to the inherently unweighted nature of our expansion-based low-diameter decomposition algorithm. We circumvent this complication with a technique borrowed from \cite{AKPW}: we bucket the weights of the edges into a few classes $E_1, E_2, ...$, and decide to make a cut when $B_G(v,r)$ forms a low-conductance cut in the graph restricted to each $E_i$. Formally, we prove the following lemma.

\begin{algorithm2e}[t]
\LinesNumbered
\KwIn{Graph $G = (V,E)$, partition of $E$ into $E_1, E_2, \cdots E_\ell$, and parameters $\beta,r\geq 0$}
\KwOut{$\{V_1, ... V_\alpha$\} partition of $V$, $U_i^1, U_i^2,...U_i$ partition of $V_i$}
$t \gets 1$ \\
\While{$G \neq \emptyset$}{
$v \gets$ arbitrary vertex in $G$ \\
$R \gets 0$\\
$V_t \gets B_G(v,R)$\\
\While{
	$e^{-\frac{r \beta}{ \ell}} \Vol_G(\partial V_t) +  \Vol_{G[E_j]}(\partial V_t) \geq 3 \beta \left(e^{-\frac{r \beta}{ \ell}} \Vol_G(V_t) +  \Vol_{G[E_j]}(V_t)\right)$ for any $E_j$\label{line:while2}
}{
    $R \gets R+1$ \label{line:expand} \\
    $V_t \gets B_G(v,R)$ \\
}
$T \gets$ shortest path tree from $v$ in $B_G(v,R)$ \\
\eIf{$R \geq r$}{
$E_t \gets $ edges of $T$ contained in $B_G(v,R- r)$ \\
$U_t^1, U_t^2, ... U_t^{j_t} \gets$ connected components of $T - E_t$ \\
}{
$U_t^1 \gets V_t$ 
\\
}
$G \gets G - B_G(v,R)$ and 
$t \gets t+1$ \\
}
\Return{$\{V_i, U_i^j\}_{i,j \geq 1}$}
\caption{$\{V_i, U_i^j\}_{i,j \geq 1} = \mathsf{Decompose}(G, \{ E_1, E_2, ... \}, \beta,r)$} 
\label{alg:decompose}
\end{algorithm2e}

\begin{lemma}
\label{lemma:decomposition}
    Let $G = (V, E)$ be an unweighted $m$-edge (multi)graph, let $E_1, E_2, ... E_\ell$ be a partition of the edges, Let $r\geq 0$, and let $\beta \in [0, 1/6]$. Algorithm~\ref{alg:decompose} computes in $O(m)$ time a partition of $G's$ vertices, $V_1, V_2, \cdots V_\alpha$, and trees, $U_i^1, U_i^2, \cdots , U_i^{j_i}$, whose vertices partition $V_i$ such that 
    \begin{itemize}
		\item For all $i \in [\ell]$, at most $6 \beta|E_i| + 6\beta m e^{-\frac{r \beta}{ \ell}}$ edges of $E_i$ cross the $V_1,V_2, \cdots , V_\alpha$ partition.
		\item Each $U_i^j$ is a tree of radius $r$.
		\item The total number of $U_i^j$, i.e. $\sum_{i \in [\alpha]} j_i$, is at most $\alpha + 4m e^{-\frac{r \beta}{ \ell}}$.
	\end{itemize}
\end{lemma}
\begin{proof}
 We first bound the running time of Algorithm~\ref{alg:decompose}. Observe that any time an edge $(u,v)$ is traversed during the ball growing phase, one of its endpoints is deleted from $G$: thus we encounter each edge at most once during our traversal. Building the shortest path trees, the $V_i$, and the $U_i^j$ can be done in $O(m)$ total work given this. Checking the condition in the while loop on Line~\ref{line:while2} can be done in $O(m)$ total time by updating the relevant volumes whenever a new vertex is introduced to $V_t$. 

We now prove the correctness of Algorithm~\ref{alg:decompose}'s output. First, observe that we only cut a cluster $V_t$ of $G$ when for every $E_j$
\[
e^{-\frac{r \beta}{ \ell}} \Vol_G(\partial V_t) +  \Vol_{G[E_j]}(\partial V_t) < 3 \beta \left(e^{-\frac{r \beta}{ \ell}} \Vol_G(V_t) +  \Vol_{G[E_j]}(V_t)\right).
\]
Further, by construction $V_t = B_G(s,R)$ for some choice $s \in V[G], R \geq 0$. We cut $\Vol_{G[E_j]}(\partial V_t)$ from $E_j$ when we partition off $V_t$: this is therefore at most 
 $3 \beta e^{-\frac{r \beta}{ \ell}} \Vol_G(V_t) +  3 \beta \Vol_{G[E_j]}(V_t)$ from $E_j$, for any $j$. As the total of all the $\Vol_G(V_t)$ terms for this bound on edges removed from $E_j$ is $m$ and the total of the  $\vol_{G[E_j]}(V_t)$ terms is $2 |E_j|$ at the end of the partitioning procedure we cut at most $6 \beta |E_j| + 6 \beta me^{-\frac{r \beta}{ \ell}}$ edges from $E_j$. 
 
We now show each generated partition piece $V_t$ contains a forest $U_t^1, ... , U_t^{j_t}$ with the desired properties. Let $R_t$ be the parameter such that $V_t = B_G(s,R_t)$ when cutting it out from the rest of the graph, and let $T_t$ be the spanning tree rooted at $s$ in $V_t$. If $R_t \leq r$, then only one  $U_t^i$ is created and it is equal to $T_t$. If instead $R_t > r$, the algorithm constructs a forest $U_t^i$ where each subtree has diameter $r$: each $U_t^i$ is the subtree of $T_t$ from a node at distance $R_t - r$ from $s$ while no node is further than $R_t$ from $s$. Further, the forest is constructed by deleting at most  $\Vol(B_G(s,R_t -r))$ edges from $T_t$. By the pigeonhole principle, at least one edge partition piece $E_j$ must have passed the expansion condition on \Cref{line:expand} of the algorithm at least $r/\ell$ times. Now as 
\[
\Vol_G(B_H(s,\alpha + 1)) \geq \Vol_G(B_H(s,\alpha)) + \Vol_G(\partial B_H(s,\alpha))
\]
for any $H$ a subgraph of $G$ and since volume of balls is monotone increasing, for $E_j$ we observe
\begin{align*}
&e^{-\frac{r \beta}{ \ell}} \Vol_G(B_G(s,R_t)) +  \Vol_{G[E_j]}(B_G(s,R_t)) \\ &\geq (1+3\beta)^{r/\ell} \left(e^{-\frac{r \beta}{ \ell}} \Vol_G(B_G(s,R_t - r)) +  \Vol_{G[E_j]}(B_G(s,R_t - r)) \right)
\end{align*}

Now since $0 \leq \beta \leq 1/6$ implies $(1+3\beta) \geq e^{2 \beta}$, $(1+3\beta)^{r/\ell} \geq \exp(2 \frac{r \beta}{\ell})$. Further, 
\[
e^{-\frac{r \beta}{ \ell}} \Vol_G(B_G(s,R_t)) +  \Vol_{G[E_j]}(B_G(s,R_t)) \leq 2 \Vol_G(B_G(s,R_t)) =  2 \Vol_G(V_t).
\]
Substituting in and rearranging we observe 
\[
2 \exp\left(-\frac{r \beta}{\ell}\right) \Vol_G(V_t) \geq  \Vol_G(B_G(s,R_t-r)).
\]
Thus $U_t^1,...,U_t^{j_t}$ consists of at most $2 \exp(-\frac{r \beta}{\ell}) \vol(V_t)$ trees. Summing over all $V_t$ gives the result.
\end{proof}

\subsection{Obtaining Ultrasparse $\kappa$-Distortion Subgraphs}
\label{sec:spectral_stretch}
\newcommand{\APS}{\mathsf{AbstractPathSparsify}}
\newcommand{\SPS}{\mathcal{S}_{PS}}
\newcommand{\TPS}{\mathcal{T}_{PS}}
Here we leverage the graph decomposition primitive from the previous section to obtain $\kappa$-distortion subgraphs consisting of a tree plus a small number of edges. Our algorithm is a modification of a classic algorithm for construction low-stretch spanning trees due to \cite{AKPW}. Briefly, \cite{AKPW}'s algorithm on a graph $G$ performs the following steps: it first computes a low-diameter decomposition of $G$ into $V_1, V_2, \dots$, it then forms a shortest-path tree within each $G[V_i]$, and finally it contracts each $V_i$ to a single node and recurses on the remaining graph. Unfortunately due to aforementioned $\Omega(\log n)$ loss intrinsic to low-diameter decomposition mentioned, any algorithm based on standard low-diameter decomposition is insufficient for our purposes. 

To improve, we leverage that the combinatorial stretch bounds given by low-stretch spanning trees are stronger than what is needed for our purposes. We show that adding a tree plus a small number of edges inside each partition piece enables us to obtain effective resistance overestimates that do not lose an $O(\log n)$ factor. For an explicit demonstration, consider the case of a constant-degree expander $G$ seen in the previous subsection. By growing a shortest-path tree from an arbitrary root and removing the edges in all but the last $O(\log k)$ levels (as is done in Algorithm~\ref{alg:decompose}), we see that $G$ contains a forest consisting of $O(n/k)$ trees each of depth at most $O(\log k)$. We then show that by adding $O(\frac{n \poly(\log n)}{k})$ edges to this forest we can ensure every edge in $G$ is retained in the subgraph or possesses many sufficiently disjoint paths between its endpoints. These sufficiently disjoint paths then enable us to bound the effective resistance across every edge in $G$ when measured through the subgraph. Although there could be different ways to add edges and form these disjoint paths, our specific approach will call a near-linear time algorithm for computing path sparsifiers on a graph obtained by contracting low-diameter clusters inside $G$. Hence, throughout this section we make reference to an abstract algorithm for path sparsification of unweighted graphs (we give a specific instantiation in Section~\ref{sec:path_sparsify}).

\begin{definition}[Path Sparsification Algorithm]
\label{def:aps}
We call an algorithm a $(\SPS,\TPS,\alpha)$-path sparsification algorithm if it takes in an unweighted graph $G$ with $n$ nodes and $m$ edges and returns a $(10\alpha,\alpha)$-path sparsifier for it with $\SPS(m,n)$ edges in $\TPS(m,n)$ time. We assume the functions $\SPS, \TPS$ are supermodular\footnote{A function $f(x,y)$ is supermodular if $f(a+ b, c+ d) \geq f(a,c) + f(b,d)$ for any $a,b,c,d$.} and non-decreasing in both arguments. 
\end{definition}

We first give a procedure which returns an ultrasparse subgraph which generates small effective resistance overestimates whenever the input graph's vertices can be partitioned into a small number of low-diameter clusters.

\begin{lemma}
Let $G = (V,E,w)$ be a weighted $n$-node $m$-edge graph with all edge weights at least $\wmin$. Let $T_1, T_2, \dots T_{\nu}$ be a spanning forest of $G$ such that each individual tree $T_i$ has effective resistance diameter\footnote{The effective resistance diameter of a graph $H$ is defined as $\max_{u, v \in H} \Reff_H(u,v)$} at most $\delta$. Let $\APS$ be a $(\SPS, \TPS, \alpha)$-path sparsification algorithm (Definition~\ref{def:aps}). 
Algorithm~\ref{alg:augtree} computes a subgraph $G''$ such that $H = \bigcup_i T_i \cup G''$ has at most with at most $n + \SPS(m,\nu)$ edges and for any edge $(u,v) \in E$, $\Reff_H(u,v) \leq 3 \delta + 1/\wmin$. Further Algorithm~\ref{alg:augtree} runs in time $O(m) + \TPS(m,\nu)$ if $\nu > 1$ and $O(m)$ time otherwise. 
\label{lemma:augtree}
\end{lemma}
\begin{algorithm2e}[h]
\LinesNumbered
\KwIn{Graph $G = (V,E,w)$, $\{T_1, T_2, \dots T_\nu\}$ forest in $G$, $\APS$ path-sparsification algorithm}
\KwOut{Subgraph $G''$}
$G' \gets (V,E,\mathbf{1})$ \tcp*{Unweighted copy of $G$ without edge weights} 
$\{V_1, V_2, \dots V_\nu \} =$ connected components of forest $\{T_1, T_2, \dots T_\nu\}$ \\
$G' \gets G' \backslash \{V_1, V_2, \dots V_\nu\}$, deleting self-loops  \\
\lIf{$G' = \emptyset$}{\Return{$\emptyset$}}
$G'' \gets \APS(G')$ \\
\Return{$G''$ with output edges mapped to the original (uncontracted) vertex set}
\caption{$G'' = \mathsf{AugmentTree} (G,\{T_1, T_2, \dots T_\nu \}, \APS )$} 
\label{alg:augtree}
\end{algorithm2e}

\begin{proof}
We first bound the runtime of the algorithm. If $\nu = 1$, the algorithm clearly runs in $O(m)$ time: we may thus assume $\nu > 1$. We can compute the contracted graph $G'$ directly in $O(m)$ time. Further, as $G'$ has $\nu$ vertices and at most $m$ edges, the cost of the call to $\APS$ is bounded by $\TPS(m,\nu)$.

Next we bound the number of edges in the graph $H =  \bigcup T_i \cup G''$. $H$ consists of a forest $\bigcup_i T_i$ combined with the output of $\APS$, with the edges mapped to the original (uncontracted) graph. As $\APS$ is called on a graph with $\nu$ vertices and at most $m$ edges we add at most $\SPS(m, \nu)$ edges to it yielding the claim.

Finally, we prove the bound on $\Reff_H(u,v)$ for any edge $(u,v) \in E(G)$. We analyze this in cases. First, we consider the case when the edge $(u,v)$ is fully contained inside a tree $T_i$. We observe that $T_i$ has effective resistance diameter at most $\delta$ by assumption. Thus as $H$ contains $T_i$ we have by Claim~\ref{effres_props} that $\Reff_H(u,v) \leq \Reff_{T_i}(u,v) \leq \delta$. If instead $(u,v)$ is not contained inside a tree $T_i$, it must be that $u$ lies in some $T_j$ and $v$ lies in some $T_k$ for $j \neq k$. In this case, we argue that the low resistance diameter of the trees in forest combined with the path sparsifier $G''$ enables us to certify a bound on $\Reff_H(u,v)$. We observe that the graph $G'$ obtained by contracting every tree $T_i$ contains an edge from $T_j$ to $T_k$ corresponding to $(u,v)$. By the guarantee of path sparsification either this edge is retained in $G''$, or $G''$ contains $10\alpha$ vertex-disjoint paths of length at most $\alpha$ connecting $T_j$ to $T_k$ in $G'$. In the former case, the edge $(u,v)$ is retained in $G''$ and hence the output graph $H$ contains $(u,v)$: thus $\Reff_H(u,v) \leq 1$. In the latter case, $G''$ contains $10\alpha$ vertex-disjoint paths of length at most $\alpha$ connecting $T_j$ to $T_k$ in $G'$. Now the edges in $G''$ correspond to weighted edges in the original input graph and therefore each have weight at least $W$. As each vertex in $G'$ corresponds to a tree $T_i$ and since each $T_i$ has resistance diameter at most $\delta$, we observe that these paths correspond to paths in $H$ of effective resistance at most $\alpha (\delta + 1/W)$ connecting vertices in $V_j$ to vertices in $V_k$. By Lemma~\ref{lem:treefact}, bounding the effective resistance in such settings, and Claim~\ref{effres_props}, this implies the effective resistance between $u$ and $v$ is bounded by $2\delta + \frac{\delta+1/W}{10} \leq 3 \delta + 1/W$.
\end{proof}

We now recursively combine this result with the graph decomposition from Section~\ref{sec:decomposition} to prove the main result of this section.

\begin{algorithm2e}[h!]
\LinesNumbered
\KwIn{Graph $G = (V,E,w)$,  $(\SPS,\TPS,\alpha)$-path sparsifier oracle $\APS$, $k,\gamma$ parameters}
\KwOut{Subgraph $H$ with potentially smaller edge weights satisfying Theorem~\ref{thm:akpw}, $\mathbf{\tau}$ leverage score overestimates of the edges in $G$}
\newcommand{\resistance}{r}
Set $\resistance_e \gets \wmax /w_e$ for all $e \in E$ \tcp*{Edge lengths for AKPW}
$F \gets \emptyset$, $S \gets \emptyset$, $\mathbf{\tau} \gets \mathbf{1} \in \R^E$ \\
$\beta = \exp\left(-\sqrt{0.5 \log \gamma \cdot \log\left(48 \log k \sqrt{\log \gamma}\right)} \right)$, $\sigma = \log_{1/\beta} \gamma$, $\delta = 48 \sigma \beta^{-1} \log k$ \\
$E_1, E_2, \cdots E_\ell \gets$ partition of edges where $E_i$ contains all edges with  $\resistance_e \in [\delta^{i-1}, \delta^i)$ \label{line:partitionEi} \\
$t \gets 1$ \\
\While{\text{$F$ does not span $G$ or $t \leq \ell$} \label{line:akpw-while}}{
$E' \gets \bigcup_{j = t - \sigma+1}^t E_j$ \\
$G_t \gets G[E'] \backslash F_{t-1}$ \\ \label{line:inclinGt} 
$G'_t \gets (V(G_t), E(G_t), \mathbf{1})$ \tcp*{Unweighted copy of $G_t$} 
$V_{t,i}, T_{t,i}^j \gets \mathsf{Decompose}(G'_t, \{E_{t-\sigma}, E_{t-\sigma+1} \dots E_t \},\beta/6, \delta/4)$ \label{line:decomposeakpw} \\
\For{each $V_{t,i}$}{
$G_{t,i} \gets \mathsf{AugmentTree}(G_t[V_{t,i}], \{T_{t,i}^1, T_{t,i}^2, \dots \}, \APS)$ \\
$F_t \gets F_{t-1} \cup \bigcup_{i,j} T_{t,i}^j$ \label{line:13} \\
$S \gets S \cup G_{t,i}$ \label{line:14}\\
\For{$e \in G[V_{t,i}]$}{
$\tau_e \gets 4 w_e w_{\max}^{-1} \delta^{t+1}$ \\
$E_j \gets E_j -\{ e \}$ \\
}
}
\For{$j = t-\sigma+1, \dots t$}{
$Y_j \gets $ arbitrary subset of $6m/k^2$ edges from each $E_j$ \\
$S \gets S \cup Y_j$ \label{line:keepextra}\\
$E_j \gets E_j - Y_j$ \\
}
$S \gets S \cup E_{t-\sigma}$ \label{line:giveup} \\
$t\gets t+1$ 
}
$F = F_{t-1}$ \\
\For{$e \in F \cup S$}{
$\tau_e \gets 1$ \DontPrintSemicolon \tcp*{Set stretch overestimate to $1$ if in output subgraph} 
}
\Return{$H = F \cup S, \mathbf{\tau}$} \DontPrintSemicolon  \tcp*{Edges in $H$ are given the same weight they had in $G$}
\caption{\newline $H = \mathsf{SpectralSubgraph}(G=(V,E,w),\APS,k,\gamma)$} \label{alg:akpw}
\end{algorithm2e}
\newcommand{\st}{\mathsf{stretch}}

\begin{theorem}
Let $G = (V,E,w)$ be an $n$-node, $m$-edge graph with edge weights which are polynomially-bounded in $n$. Let $k,\gamma \geq 2$ be parameters. Let $\APS$ be an $(\SPS,\TPS,\alpha)$-path sparsification algorithm (in the sense of Definition~\ref{def:aps}) for any $\alpha$. In $O(m + \TPS \left(O(m), O\left(\frac{m}{k}\right) )\right)$ time Algorithm~\ref{alg:akpw} finds a subgraph $H$ with at most 
\[
n + O\left(\frac{m}{\gamma} + \frac{m \log \gamma}{k^2} \right) + \SPS\left(O(m) , O\left(\frac{m}{k}\right) \right) 
\]
edges which is a $\kappa$-distortion subgraph of $G$ for 
\[
\kappa = O\left(m \exp\left( \sqrt{8 \log \gamma \cdot \log \left(48 \log k \sqrt{\log \gamma}  \right)} \right) \log k \sqrt{\log \gamma} \right) 
\]
It also returns a vector $\mathbf{\tau} \in \R^E$ which satsifies $\tau_{(u,v)} \geq w_{(u,v)} \Reff_H(u,v)$ and $\norm{\tau}_1 \leq \kappa$.
\label{thm:akpw}
\end{theorem}
We remark that the sparsity and runtime guarantees in the above are independent of $\alpha$. Before we prove this theorem, we state and prove some structural invariants about the algorithm:
\begin{lemma}
\label{lem:edge-decrease}
Consider an execution of Algorithm~\ref{alg:akpw}, and consider an iteration $t$ of the while loop on \Cref{line:akpw-while}. Each time a weight bucket $E_j$ is included in $G_t$ on \Cref{line:inclinGt}, we conclude that iteration of the while loop by decreasing the number of edges in $E_j$ by a factor of $\beta$. In addition, each connected component of $F_t$ is a tree of effective resistance diameter at most
$w_{\max}^{-1} \delta^{t+1}$.
\end{lemma}
\begin{proof} We first prove that $|E_i|$ decreases by a factor of $\beta$ each iteration. Let $E_i$ be processed in iteration $t$, and note that the graph $G_t$'s edges are partitioned into at most $\sigma$ buckets. For our value $\delta = 48 \sigma \beta \log k$, if edge set $E_i$ is processed in iteration $t$ Lemma~\ref{lemma:decomposition} implies the call to $\mathsf{Decompose}$ forms a vertex partition $V_{t,i}$ which cuts at most
\[
\beta |E_i| + 6 |E(G_t)| e^{-\frac{(\delta/4)(\beta/6)}{\sigma}} = \beta |E_i| + 6|E(G_t)| e^{-\frac{\delta \beta}{24 \sigma}} \leq \beta |E_j| + 6m k^{-2}
\] 
edges from $E_i$, where we used $|E(G_t)| \leq m$. Further, on \Cref{line:keepextra} we move $6 m/k^2$ edges from $E_i$ to $S$ during every iteration $E_i$ is processed. Thus we see that $E_i$ ends the iteration with at most $\beta |E_i|$ edges as desired.

For the second condition, we induct on $t$. For $t=1$, observe that every edge seen in $G_t$ has resistance between $1$ and $\delta$: this implies it has weight between $w_{\max}/\delta$ and $w_{\max}$. By the guarantee of Algorithm~\ref{alg:decompose} (\Cref{lemma:decomposition}), the trees $T_{1,i}^j$ have unweighted radius at most $\frac{\delta}{4}$ and hence unweighted diameter at most $\frac{\delta}{2}$. As every edge in these trees has effective resistance at most $\delta w_{\max}^{-1}$ the claim follows. Now assume the claim for $t = \nu$: we will show it for $t = \nu+1$. The edges in $G_{\nu+1}$ have effective resistance at most $w_{\max}^{-1} \delta^{\nu+1}$. Now, the edges added to $F_{\nu+1}$ belong to $T_{\nu+1,i}^j$: in the unweighted contracted graph $G'_{\nu+1}$ these are trees of (unweighted) diameter at most $\frac{\delta}{2}$ by Lemma~\ref{lemma:decomposition}.  In $G$, the edges from $T_{\nu+1,i}^j$ connect together subsets of vertices which correspond to the forests in $F_{\nu}$: by the induction hypothesis the trees of $F_{\nu}$ have effective resistance diameter at most $w_{\max}^{-1}  \delta^{\nu+1}$. Combining these two observations, any path through a forest in $F_{\nu+1}$ travels through at most $\frac{\delta}{2}$ edges with resistance at most $w_{\max}^{-1} \delta^{\nu+1}$ (coming from the edges of $T_{\nu+1,i}^j$) and at most $\frac{\delta}{2}$ forests in $F_{\nu}$ with effective resistance diameter $w_{\max}^{-1}  \delta^{\nu+1}$. Adding these together, we obtain an effective resistance overestimate of 
\[
\frac{\delta}{2} \left( w_{\max}^{-1} \delta^{\nu+1} \right) + \frac{\delta}{2} \left( w_{\max}^{-1} \delta^{\nu+1} \right) = w_{\max}^{-1}  \delta^{\nu+2}.
\]
\end{proof}
With Lemma~\ref{lem:edge-decrease} established, we prove that Algorithm~\ref{alg:akpw} outputs a low-distortion subgraph:
\begin{lemma}
Let $G = (V,E,w)$ be a $n$-node $m$-edge graph and let $H,\tau$ be the output of Algorithm~\ref{alg:akpw} in the setting of Theorem~\ref{thm:akpw}. Then $H$ is a $\kappa$-distortion subgraph of $G$ for 
\[
\kappa = O\left(m \exp\left( \sqrt{8 \log \gamma \cdot \log \left(48 \log k \sqrt{\log \gamma}  \right)} \right) \log k \sqrt{\log \gamma}\right)
\] 
Further, $\mathbf{\tau}$ satisfies $\mathbf{\tau}_{(u,v)} \geq w_{(u,v)} \Reff_H(u,v)$ for any $(u,v) \in E$ and $\norm{\tau}_1 \leq \kappa$.
\label{lem:akpw-stretch}
\end{lemma}
\begin{proof}
Since $H$ is clearly a subgraph of $G$ it suffices to show that $\norm{\tau}_1 \leq \kappa$ and $\mathbf{\tau}_{(u,v)} \geq w_{(u,v)} \Reff_H(u,v)$ for any $(u,v) \in E(G)$. We do this by using the effective resistance guarantee of Algorithm~\ref{alg:augtree} to certify resistance bounds on edges contained within a partition piece $V_{t,i}$. 

Since $H$ is a subgraph of $G$, every edge of $G$ that is added $H$ receives $\tau_{(u,v)} = 1 \geq w_{(u,v)} \Reff_H(u,v)$. Now let $e=(u,v) \in E(G)$ be contained in some $V_{t,i}$. We will show that the forest $F_t$ combined with the path sparsifier $G_{t,i}$ computed with $\mathsf{AugmentTree}$ give $e$ an effective resistance overestimate of $4 w_{\max}^{-1} \delta^{t+1}$.

First, observe that each vertex in $G_t$ corresponds to a tree in $F_{t-1}$ and therefore $V_{t,i}$ corresponds to a subset of those trees. Since $T_{t,i}^1, T_{t,i}^2, \dots$ are trees inside $G_t[V_{t,i}]$, we see that every tree in $F_t$ is fully contained in some $V_{t,i}$. Let $F_{t,i} = F_t[V_{t,i}]$, and note that by Lemma~\ref{lem:edge-decrease} each tree in $F_{t,i}$ has effective resistance diameter $w_{\max}^{-1} \delta^{t+1}$ when the edges are given the edge weights they have in $G$. In addition, the edges in $G_t[V_{t,i}]$ which we pass into $\mathsf{AugmentTree}$ have weight at least $w_{\max} \delta^{-t}$ in $G$. Therefore, Lemma~\ref{lemma:augtree} ensures that any for any edge $(u,v)$ in $V_{t,i}$ we have 
\[
\Reff_{H}(u,v) \leq \Reff_{F_t \cup G_{t,i}}(u,v) \leq 3 w_{\max}^{-1} \delta^{t+1} + w_{\max}^{-1} \delta^{t} \leq 4 w_{\max}^{-1} \delta^{t+1}. 
\]
We remark that this corresponds to the value given to $\tau_{(u,v)}$ in the algorithm.

The above shows that the effective resistance through $H$ across the endpoints of any edge contained in a $V_{t,i}$ is bounded. We now fix a weight bucket $E_j$ and bound $E_j$'s contribution to $H$'s spectral distortion. 
Observe that the number of edges in $E_j$ which are not within a $V_{t,i}$ in iteration $t = j+x$ of the algorithm is at most $|E_j| \beta^{x}$ by Lemma~\ref{lem:edge-decrease}. As the weight of any edge in $E_j$ is at most $w_{\max} \delta^{1-j}$, we see that at most $|E_j| \beta^{x}$ end up receiving a $\tau_{(u,v)}$ value larger than
\[
\left(w_{\max} \delta^{1-j}\right)\left( 4 w_{\max}^{-1} \delta^{j+x+1} \right) = 4 \delta^{2+x}.
\]
Therefore, the edges in $E_j$ get effective resistance overestimates summing to at most 
\begin{align*}
4 |E_j| \sum_{x  = 0}^{\sigma-1} \beta^x \delta^{(2+x)}
\end{align*}
since after $\sigma$ iterations we add the remaining edges in $E_j$ to $H$. This is at most
\begin{align*}
4 |E_j| \sum_{x  = 0}^{\sigma-1} \beta^x \delta^{(2+x)} &= 4 |E_j| \sum_{x = 0}^{\sigma-1} \delta^2 \left( 48 \sigma \log k\right)^x \\ 
&\leq 5 |E_j| \delta^2 \left( 48 \sigma \log k\right)^{\sigma - 1} = 5 |E_j| \beta^{-2}  \left( 48 \sigma \log k\right)^{\sigma +1}
\end{align*}
where we used 
\begin{align*}
48 \sigma \log k = \frac{48 \log k \log \gamma}{-\log \beta} &= \frac{48 \log k \log \gamma}{\sqrt{0.5 \log \gamma \cdot \log \left( 48 \log k \sqrt{\log \gamma} \right) }} \\
&\geq \frac{48 \log k \sqrt{\log \gamma}}{  \log \left( 48 \log k \sqrt{\log \gamma} \right)} \geq 5
\end{align*}
for $k, \gamma \geq 2$ and $\sum_{i=0}^n c^i = \frac{c^{n+1} - 1}{c-1} \leq \frac{5}{4} c^n$ for $c \geq 5$. Our choice of $\beta$ yields  $\log 1/\beta = \sqrt{0.5 \log \gamma \cdot \log\left(48 \log k \sqrt{\log \gamma}\right)}$: this implies
\[ 
\sigma = \frac{\log \gamma}{\log 1/\beta} = \frac{\log \gamma}{\sqrt{0.5 \log \gamma \cdot \log \left( 48 \log k \sqrt{\log \gamma} \right)}} \leq \sqrt{\log \gamma}.
\]
Thus we have
\begin{align*}
\beta^{-2}  \left( 48 \sigma \log k\right)^{\sigma +1} &\leq \beta^{-2} \left( 48  \log k  \sqrt{\log \gamma} \right) \left( 48 \log k \sqrt{\log \gamma } \right)^{\frac{\log \gamma}{\log 1/\beta}} \\
&=  \left(48 \log k \sqrt{\log \gamma}  \right) \exp \left(2 \log 1/\beta + \frac{\log \gamma}{\log 1/\beta} \log\left(48 \log k \sqrt{\log \gamma} \right) \right) \\
&= 48 \exp\left( \sqrt{8 \log \gamma \cdot \log \left(48 \log k \sqrt{\log \gamma}  \right)} \right) \sqrt{\log \gamma} \log k.
\end{align*}
The last equality used the definition of $\beta$ and the algebraic fact that $c_1 x + c_2 / x = 2 \sqrt{c_1 c_2}$ for $x = \sqrt{c_2/c_1}$. Substituting this in yields that $E_j$'s contribution to the spectral distortion of $H$ is bounded by
\[
240 |E_j|\exp\left( \sqrt{8 \log \gamma \cdot \log \left(48 \log k \sqrt{\log \gamma}  \right)} \right) \log k \sqrt{\log \gamma}.
\]
implying the claimed bound.
\end{proof}
Finally, we bound the runtime and sparsity guarantees of Algorithm~\ref{alg:akpw}:
\begin{lemma}
Let $G$ be a weighted graph with $n$ nodes and $m$ edges, and let $H$ be the output of Algorithm~\ref{alg:akpw} in the setting of Theorem~\ref{thm:akpw}. Then $H$ has at at most
\[
n + O\left(\frac{m}{\gamma} + \frac{m \log \gamma}{k^2} \right) + \SPS\left(O(m) , O\left( \frac{m}{k} \right) \right) 
\]
edges. Further, Algorithm~\ref{alg:akpw} runs in time $O(m + \TPS \left(O(m), O\left(\frac{m}{k}\right) \right)$. 

\end{lemma}
\begin{proof} Observe that there are four different ways edges can be added to $H$: they can be added to the forest $F_t$ on \Cref{line:13}, to $S$ through the path sparsifiers $G'_{t,i}$ on \Cref{line:14}, to $S$ via the extra edges from each $E_j$ we keep on \Cref{line:keepextra}, or to $S$ by the check on \Cref{line:giveup} (after a weight bucket has been processed in a $G_t$ sufficiently many times). We bound these in order. Clearly, the returned forest $F$ consists of at most $n$ edges. By the guarantee of Algorithm~\ref{alg:augtree}, the calls to $\mathsf{AugmentTree}$ during iteration $t$ are on graphs $G_t[V_{t,i}]$: let $m(t)$ denote the total number of edges contained in the $G_t[V_{t,i}]$. We observe that $G_t$'s edge set is a subset of $\bigcup_{j=t-\sigma}^t E_j$, and moreover each edge in $E_j$ is contained inside at most one $G_t[V_{t,i}]$ (as once this happens we delete it from our edge set): thus $\sum_{t} m(t) \leq m$. Next, each $G_t[V_{t,i}]$ contains a forest $T_{t,i}^l$: by the guarantee of $\mathsf{Decompose}$ we observe that the total number of $T_{t,i}^l$ is the number of $V_{t,i}$ plus at most $4 m(t) e^{- \frac{(\delta/4) \cdot (\beta/6)}{\sigma}} = 4 m(t) e^{-2\log k} \leq 4m(t)/k$. Now as $\APS$ is called inside $\mathsf{AugmentTree}$ only when $T_{t,i}^l$ consists of more than $1$ tree, we may aggregate all the calls to $\APS$ in iteration $t$ into a single call on a graph with $m(t)$ edges and $O(m(t)/k)$ nodes. Since $\SPS$ is supermodular and non-decreasing in both arguments, we see that the number of edges added to $H$ via path sparsifiers is at most
\[
\sum_t \SPS\left(m(t), O\left(\frac{m(t)}{k} \right) \right) \leq \SPS\left(m, O\left(\frac{m}{k} \right) \right).
\]
For the edges added on \Cref{line:keepextra}, we again see that each $E_j$ is included in $E'$ at most $\sigma$ times-- thus this collectively adds $O(\frac{m \sigma}{k^2}) \leq O(\frac{m \log \gamma}{k^2})$ edges. Finally, after an $E_j$ has been processed $\sigma$ times we observe that it decreases in size by a factor of $\gamma$: thus the addition on \Cref{line:giveup} adds $O(m/\gamma)$ edges to $H$. Combining gives our claimed size bound. 

Finally, we prove the running time of our algorithm. We first bound the cost of all steps excluding the calls to $\mathsf{AugmentTree}$. Each time a weight class is fed to $\mathsf{Decompose}$ on \Cref{line:decomposeakpw}, the number of edges in that class falls by a factor of $\beta$. Thus, the total contribution of weight class $E_j$ to the running time of all calls to $\mathsf{Decompose}$ is only $O(m)$: the number of edges left to consider in the recursively generated subproblems falls geometrically. The time it takes to sort the edges into weight buckets $E_j$ is at most $O(m)$ via radix sort with base $\poly(n)$ (here we use our assumption of polynomially-bounded edge weights) and the running time of every other step in the algorithm can be implemented in the trivial fashion in $O(m)$ time. Finally, to bound the runtime of the calls to $\mathsf{AugmentTree}$ we again observe that during iteration $t$ we can aggregate the nontrivial calls it makes to $\APS$ to a single one on a graph with at most $m(t)$ edges and $O(m(t) / k)$ vertices. As $\TPS$ is also superlinear in both arguments, the runtime of these calls to $\APS$ can again be bounded by $\TPS(O(m), O(m/k))$ as desired. 
\end{proof}
Combining the last two lemmas gives Theorem~\ref{thm:akpw}. To conclude this section, we show that combining Theorem~\ref{thm:akpw} with the path sparsification algorithm construction in Section~\ref{sec:path_sparsify} and choosing parameters appropriately yields an efficient construction of $\kappa$-distortion subgraphs.
\fastlpsubgraphs*
\begin{proof}
Let $\PathSparsify(G,O(\log^5 n))$ be the algorithm guaranteed by Theorem~\ref{thm:path-sparse}: note that this is a $(O(n \log^8 n), O(m + n \log^{18} n), O(\log^5 n))$-path sparsification algorithm. We apply $\mathsf{SpectralSubgraph}$ to $G$ with parameters $k = \log^{18} n$ and $\gamma = (\log \log n)^c$. Combining this with the guarantee of Theorem~\ref{thm:akpw}, we therefore see that $H$ contains
\[
n + O \left(\frac{m\log^{10} n}{k}+ \frac{m \log \gamma}{k^2}  + \frac{m}{\gamma} \right) = n + O\left(\frac{m}{(\log \log n)^c}\right)
\]
edges, and it is computed in time 
\[
O\left( m (\log \log  n)^{1/2} + \frac{m \log^{18} n}{ \log^{18} n} \right) =  O \left( m (\log \log  n)^{1/2}\right).
\]
It remains to bound the spectral distortion of the computed subgraph. We use the notation $\log^{(i)}(n)$ to denote the result of applying the log function $i$ times: hence $\log^{(1)}(n) = \log n$ and $\log^{(i+1)} n = \log \left( \log^{(i)} n \right)$. For our values of $\gamma$ and $k$ we see 
\begin{align*}
\sqrt{8 \log \gamma \cdot \log \left(48 \log k \sqrt{\log \gamma}  \right)} &= \sqrt{8 c \log^{(3)} n \cdot \left( \log^{(3)} n + O(\log^{(4)} n) \right)} \\
&= \left(\sqrt{8c} + o(1)\right) \log^{(3)} n.
\end{align*}
Thus by Theorem~\ref{thm:akpw} our computed subgraph is a $\kappa$-distortion subgraph for 
\begin{align*}
\kappa &= O\left( m \exp \left( ( \sqrt{8c} + o(1)) \log^{(3)} n \right) \log^{(2)} n \sqrt{\log^{(3)} n} \right) \\ 
&=  O\left( m \left( \log \log n \right)^{ \sqrt{8c} +1 + o(1)} \right).
\end{align*}
Assembling these pieces yields the claim.
\end{proof}

\section{Efficient Path Sparsification}
\label{sec:path_sparsify}

In this section we prove \Cref{thm:path-sparse} showing that path sparsifiers can be efficiently computed. To prove this result we provide several new algorithmic components of possible independent interest. First, in \Cref{sec:short_expander_paths} we show that dense near-regular expanders have many short vertex disjoint paths. Then, in  \Cref{sec:expander-compute} we leverage this result with a new sampling scheme and previous expander partitioning results to show that we can efficiently path-sparsify large amount of the volume of dense near-regular graphs. In \Cref{sec:deg_reg_decomp} we then show that every dense graph can be efficiently decomposed into dense near-regular subgraphs. Carefully applying these tools yields our desired result in \Cref{sec:path-sparse-proof}.

\newcommand{\outvar}[1]{ {#1}_{\mathrm{out}} }
\newcommand{\invar}[1]{ {#1}_{\mathrm{in}} }
\newcommand{\ain}{a^\mathrm{in}}
\newcommand{\aout}{a^\mathrm{out}}
\newcommand{\bin}{b^\mathrm{in}}
\newcommand{\bout}{b^\mathrm{out}}
\newcommand{\vin}{v^\mathrm{in}}
\newcommand{\vout}{v^\mathrm{out}}
\newcommand{\tin}{t^\mathrm{in}}
\newcommand{\tout}{t^\mathrm{out}}
\newcommand{\sout}{s^\mathrm{out}}

\newcommand{\directed}[1]{\overrightarrow{#1}}
\newcommand{\fwd}[1]{\overrightarrow{#1}}
\newcommand{\bckwd}[1]{\overleftarrow{#1}}

\newcommand{\forward}{\rightarrow}
\newcommand{\backward}{\leftarrow}

\subsection{Short Vertex Disjoint Paths in Expanders}
\label{sec:short_expander_paths}

In this section we show that in every dense expander of balanced degrees, for every pair of vertices $s$ and $t$ there are many vertex disjoint paths between them (where here and throughout we ignore the necessary shared use of $s$ and $t$). The number of paths and the length of these paths depend on the degree ratio and the conductance of the graph. Formally we define conductance, \Cref{def:conductance}, and present the main theorem of this section, \Cref{thm:paths_in_expanders} below.

\begin{definition}[(Edge) Conductance]
\label{def:conductance}
For undirected graph $G = (V, E)$ (possibly with self-loops) and $S \subseteq V$ we define the \emph{(edge) conductance} of $S$ and $G$ by
\[
\condedge(S) \defeq \frac{|\cutset(S)|}{\min\{ \vol(S), \vol(V\setminus S) \}} 
\enspace
\text{and}
\enspace
\condedge(G) \defeq \min_{S \subseteq V , S \notin \{\emptyset, V\}} 
\condedge(S)
\enspace \text{respectively.}
\]
\end{definition}

We call any family of $n$-node graphs $G$ with $\condedge(G) = \Omega(\log^c n)$ for some $c$, \emph{expanders}. Our main result of this section is the following theorem regarding vertex disjoint paths in such expanders.

\begin{theorem}[Short Vertex Disjoint Paths in Approximately Regular Expanders]
\label{thm:paths_in_expanders}
For all pairs of vertices $s$ and $t$ in an $n$-node undirected graph $G = (V,E)$ (possibly with self-loops) there is a set of at least $\condedge(G) \dmin(G) / (8 \dratio(G))$ vertex-disjoint paths from $s$ to $t$ of length at most $ (4 \dratio(G)/ \condedge(G)) \cdot \ln(n/\dmin(G))$.
\end{theorem}

Our main technical tool towards proving \Cref{thm:paths_in_expanders} is that graphs with large vertex conductance have many short vertex-disjoint paths. The formal definition of vertex conductance, \Cref{def:vertex_conductance}, and this tool, \Cref{lem:vertex_expansion_to_paths}, are given below.

\begin{definition}[Vertex Conductance]
\label{def:vertex_conductance}
For undirected graph $G = (V, E)$ and $S \subseteq V$ we define the \emph{vertex conductance} of $S$ and $G$ by
\[
\condvert(S) \defeq \frac{|N(S) \setminus S |}{\min\{ |S| , |V \setminus S| \}} 
\enspace \text{ and } \enspace \condvert(G) \defeq \min_{S \subseteq V , S \notin \{\emptyset, V\}}
\condvert(S)
\enspace
\text{ respectively.} 
\]
\end{definition}

Note that for any set $S$ with $|S| = \lceil |V| / 2\rceil$ we have $\condvert(S) \leq 1$. Consequently $\condvert(G) \in [0,1]$ for all undirected $G$.

\begin{restatable}[Short Vertex-Disjoint Paths in Vertex Expanders]{lemma}{vertExpandPaths}
\label{lem:vertex_expansion_to_paths}
Let $G = (V,E)$ be an $n$-node undirected graph with $\condvert(G) \geq \phi$. Then for all nodes $s, t \in V$ with $s \neq t$ of degree at least $d$ there is a set of at least $\phi d/8$-vertex disjoint paths from $s$ to $t$ of length at most $(4/\phi) \ln(n/d)$.
\end{restatable}

\Cref{lem:vertex_expansion_to_paths} implies  \Cref{thm:paths_in_expanders} by a standard technique of relating edge and vertex expansion.

\begin{proof}[Proof of \Cref{thm:paths_in_expanders}]
By \Cref{lem:vertex_expansion_to_paths} it suffices to show that \newline $\condvert(G) \geq$  $\condedge(G) / \dratio(G)$. To prove this, let $S \subseteq V$ be arbitrary and note that by assumption $\Vol(S) \geq \dmin(G) |S|$, $\Vol(V\setminus S) \geq \dmin(G) |V \setminus S|$, and $|\cutset(S)| \leq  \dmax(G) |N(S) \setminus S|$. Consequently, 
\[
\condvert(S) 
= \frac{|N(S) \setminus S|}{\min\{|S|,|V\setminus S| \} }
\geq 
\frac{(\dmax(G))^{-1} |\cutset(S)|}{(\dmin(G))^{-1} \min\{\vol(S),\vol(V\setminus S)\}}
= \frac{\condedge(S)}{\dratio(G)}
~.
\]
The claim then follows by the definition of $\condvert(G)$ and $\condedge(G)$.
\end{proof}

Consequently, in the rest of this section, we prove \Cref{lem:vertex_expansion_to_paths}. Our inspiration for this lemma is the seminal result of \cite{KleinbergR96} which proved an analogous result edge-disjoint paths in expanders. Their proof considered the minimum cost flow problem of routing flow of minimum total length between $s$ and $t$ and by reasoning about primal and dual solutions to this linear program, they obtained their result. We prove \Cref{lem:vertex_expansion_to_paths} similarly, by considering the minimum-length flow in the natural directed graph which encodes vertex-disjointness defined as follows.

\begin{definition}[Directed Representation of Vertex Capacitated Graphs]
Given undirected graph $G = (V,E)$ we let $\directed{G} \defeq (\directed{V}, \directed{E})$ denote the directed graph where for each $a \in V$ we have vertices $\ain$, and $\aout$ and edge $(\ain, \aout)$ and for each edge $\{a,b\} \in E$ we have edges $(\bout, \ain)$ and $(\aout, \bin)$. 
\end{definition}

Note that any path of vertices $a$, $b$, $c$, $d$, $e \in V$ has an associated path $\outvar{a}$, $\invar{b}$, $\outvar{b}$, $\invar{c}$, $\outvar{c}$, $\invar{d}$, $\outvar{d}$, $\invar{d} \in \directed{V}$ path in $\directed{G}$. Further a set of $a$ to $b$ paths are vertex-disjoint in $G$ if and only if their associated paths are edge-disjoint in $\directed{G}$. Also a simple path in $G$ has length $k$ if and only if its associated path in $\directed{G}$ has length $2k - 1$. Consequently, to reason about short vertex disjoint paths in $G$ it suffices to reason about short edge-disjoint paths in $\directed{G}$. 

To reason about the length of these paths, as in \cite{KleinbergR96} we consider the minimum cost flow problem corresponding to sending a given given amount of flow from $s$ to $t$ while using the fewest number of edges. However, unlike \cite{KleinbergR96} the graph we use is directed, i.e. $\directed{G}$, and thus we need to characterize the minimizers of a slightly different minimum cost problem. This optimality characterization of the minimum cost flow problem is given below and proven in \Cref{sec:primal_dual_mincost}. 

\begin{restatable}[Dual Characterization of Shortest Flow]{lemma}{mincostflow} 
\label{lem:minimum_cost_lemma}
For directed graph $G = (V, E)$ and vertices $s,t \in V$ if there are at most $F$ edge-disjoint paths from $s$ to $t$ then there is an integral $s$-$t$ flow $f \in \{0,1\}^E$ corresponding to $F$ edge-disjoint paths from $s$ to $t$ using a minimum number of edges and $v \in \R^V$ such that for every $(a,b) \in e$ if $f_e = 1$ then $v_a - v_b \geq 1 $ and if $f_e = 0$ then $v_a - v_b \leq 1$.
\end{restatable}

As in \cite{KleinbergR96}, our approach to showing that the paths are short is to sweep over $v$ and show that the associated sets increase rapidly. Here, we tailor the analysis to the structure of our directed (as opposed to undirected) minimum cost flow problem. Our main structural lemma is given below. 

\begin{lemma}[Characterization of Shortests Flow] \label{lem:digraph_sets} Let $G = (V,E)$ be an undirected $n$-node graph for which there are $F$ vertex disjoint paths from $s \in V$ to $t \in V$. Further, let $f \in \{0,1\}^{\directed{E}}$ be an integral flow corresponding to $F$ disjoint paths from  $\outvar{s}$ to $\invar{t}$ in $\directed{G}$ using a minimum number of edges and let $v \in \R^{\directed{V}}$ be as described in \Cref{lem:minimum_cost_lemma}. The following properties hold:
\begin{enumerate}

\item The values of $v$ decrease monotonically along each path in $f$ and decrease by at least $1$ after each edge. Consequently, the length of each path is at most $y_{\sout} - y_{\tin}$.

\item For all $\alpha \in [y_{\sout} , y_{\tin}]$ if $S_{\alpha}^{\directed{V}} \defeq \{a \in \directed{V} | y_a \leq \alpha \}$ and 
$S_{\alpha}^{V} \defeq \{a \in V |
 \ain \in S_{\alpha}^{\directed{V}} ~ \text{ or }  ~ \aout \in S_{\alpha}^{\directed{V}} \}$ then either $|S_{\alpha -2}^V| \geq n/2$  or
 $| S_{\alpha - 2}^{V} |  \geq (1 + \condvert) \cdot (| S_{\alpha}^V| - F)$.
 
 \item We have $|S_{y_{\sout} -1}^V| \geq \deg(s) + 1 - F$
 
\end{enumerate}

\end{lemma}

\begin{proof}
\emph{Claim 1}: By \Cref{lem:minimum_cost_lemma}, whenever $f_e = 1$ for $e = (a,b)$  then $v_a - v_b \geq 1$. Combined with the fact that $f$ corresponds to disjoint $s$ to $t$ paths immediately yields the claim.

\emph{Claim 2}: By claim 1, the set of edges leaving $S_\alpha^{\directed{V}}$ is of size exactly $F$. Consequently, leveraging that if $f_e = 0$ for $e = \{a , b \}$ we have $v_{b} \geq v_a - 1$ by \Cref{lem:minimum_cost_lemma} we have that if $a \in S_{\alpha}^{V}$, then either (1) $\ain \in S_{\alpha}^{\directed{V}}$, $f_{(\ain, \aout)} = 1$, and $\aout \notin S_{\alpha}^{\directed{V}}$ or (2) $\aout \in S_{\alpha - 1}^{\directed{V}}$. Further, if $\aout  \in S_{\alpha - 1}^V$ and $b \in N(a)$ and $f_{(\aout, \bin)} = 0$ then we have $\bin \in S_{\alpha - 2}^V$, again by \Cref{lem:minimum_cost_lemma}. Consequently, if there are $\ell$ vertices for which case (1) holds then the neighbors of all the other $|S_{\alpha}^V|$ vertices are in $|S_{\alpha - 2}^{V}|$ except for at most $F - \ell$ vertices. The claim then follows as $\ell \in [0,F]$ and
 \[
| S_{\alpha - 2}^{V} |  \geq (1 + \condvert) \cdot (| S_{\alpha}^V| - \ell) - (F - \ell) \geq (1 + \condvert) (|S_{\alpha}^V| - F)
\]
\emph{Claim 3}: Again by \Cref{lem:minimum_cost_lemma} we have $\bin \in S_{y_{\sout} - 1}^V$ for all $\bin$ where $\{a, b\} \in E$ and $f_{(\aout,\bin)} = 0$, which happens for all but $F$ edges.

\end{proof}

We now have everything we need to prove \Cref{lem:vertex_expansion_to_paths}

\begin{proof}[Proof of \Cref{lem:vertex_expansion_to_paths}]
Let $T \defeq ( N(s) \cap N(t) ) \setminus \{s,t\}$. Note that there are $|T|$ vertex disjoint paths of length at most $2$ from $s$ to $t$ (ignoring the shared use of $s$ and $t$). Further, we see that $N(s) \setminus T$ and $N(t) \setminus T$ each have size at least $d - |T|$ and by assumption of vertex conductance this implies that every set $S$ with $(N(s) \setminus T) \subset S$ and $S \cap (N(t) \setminus T) = \emptyset$ has $|N(S) \setminus S| \geq \phi (d - |T|)$ and therefore if we further constrain that $S \cap (T \setminus \{s,t\}) = \emptyset$ this implies that $|(N(S) \setminus (S \cup T))| \geq \phi(d - |T|) - |T|$. By maxflow minimum cut theorem on $\directed{G}$ this implies that there are at least $\phi d - (1 + \phi) |T|$ disjoint paths from $s$ to $t$, not using $T$. Consequently, the number of vertex disjoint paths in total from $s$ to $t$ is
\begin{align*}
\max\{\phi d - (1 + \phi) |T| , 0\} + |T| &=
\max\{\phi (d - |T|), |T|\}
\geq \min_{\alpha \geq 0} \max\{\phi (d - \alpha), \alpha\} \\
&= \left(\frac{\phi}{1 + \phi} \right) d \geq \frac{\phi d}{2} ~.
\end{align*}
It simply remains to bound the length of a smaller set of paths.

To bound the length of these paths, let $F = \phi d / 8$. Further, let $f \in \{0,1\}^{\directed{E}}$ be an integral flow corresponding to $F$ disjoint paths from  $\sout$ to $\tin$ in $\directed{G}$ and let $v \in \R^{\directed{V}}$ be described as in \Cref{lem:minimum_cost_lemma}. We now prove by induction that for all $t \geq 0$ it is the case that either
\begin{equation}
|S_{y_s - 1 - 2 t}^V| \geq n/2
\enspace \text{ or } \enspace
|S_{y_s - 2t}^V| \geq (1 + (\phi / 2) )^t d / 2
\label{inductive_hypothesis}
\end{equation}
Note that for the base case we have that 
\[
|S_{y_s - 2}^V| \geq d + 1 - F \geq d/2 ~.
\]
For the inductive case note that if the claim holds for $t$ and it is not the case that $|S_{y_s - 1 - 2 (t + 1)}^V| \geq n/2$ then $F \leq |S_{y_s - 1 - 2t}^V| \cdot \phi /4$ and consequently, by \Cref{lem:digraph_sets} and the inductive hypothesis we have
\[
|S_{y_s - 1 - 2(t+1)}^V| \geq (1 + \phi) \cdot (|S_{y_s - 1 - 2t}| - F) \geq (1 + \phi)(1 - \phi/4) |S_{y_s - 1 - 2t}| \geq (1 - (\phi/ 2)) |S_{y_s - 1 - 2t}| 
\]
where we used $\phi \in (0,1)$. Consequently by induction \eqref{inductive_hypothesis} holds for all $t > 0$ and for some $t \leq  \frac{2}{\phi} \ln(n/d)$ we have $|S^V_{y_s- 1 - 2t}| \geq n/2$. 
By symmetry this also implies that $y_t \leq y_s + 2 +  \frac{4}{\phi} \ln(n/d)$. Since a path of length $k$ in $G$ is length $2k - 1$ in $\directed{G}$ the result then follows again by  \Cref{lem:digraph_sets}.
\end{proof}

\newcommand{\ExpanderDecompose}{\mathsf{ExpanderDecomp}}
\newcommand{\PartialPathSparsify}{\mathsf{PartialPathSparsify}}
\newcommand{\Ecut}{E_{\mathrm{cut}}}

\subsection{Path Sparsification on Dense Near-Regular Expanders}
\label{sec:expander-compute}

Here we provide an efficient procedure to compute path sparsifiers of a constant fraction of the edges in a dense degree regular graph. This procedure leverages \Cref{thm:paths_in_expanders}, which shows that dense near-regular expanders have many short vertex disjoint paths. Coupled with a procedure for partitioning a graph into approximately regular subgraphs (\Cref{sec:deg_reg_decomp}) this yields our main theorem of this section, an efficient path sparsification procedure. Consequently, in the remainder of this subsection our goal is to prove the following theorem. 

\begin{theorem}[Partial Path Sparsification of Nearly Regular Graphs]
	\label{thm:partial_path_sparsifiers} Given any $n$-node $m$-edge undirected unweighted graph $G = (V, E)$ and $k \geq 1$, the procedure $\PartialPathSparsify(G, k)$ (\Cref{alg:partialpathsparsify}) in time  $O(m + n k \cdot \dratio(G) \log^8(n))$, outputs $F,\Ecut \subseteq E$ such that  w.h.p. in $n$
	\begin{itemize}
		\item \emph{(Size Bound)}: $|F| = O(n k \cdot \dratio(G) \log(n))$, $|\Ecut| \leq |E|/2$ and
		\item \emph{(Path Sparsification)}: $G[F]$ is a $(\Omega(k / (\dratio(G) \log^2(n))) , O(\dratio(G) \log^4(n)) )$-path sparsifier of $(V,E\setminus\Ecut)$.
	\end{itemize}
\end{theorem}

Our partial path sparsification procedure, $\PartialPathSparsify(G, k)$ (\Cref{alg:pathsparsifier}), works simply by randomly sampling the edges, partition the resulting graph into expanders, and output the edges of those expanders as a path sparsifier of the edges on those node induced subgraphs in the original graph. In the following lemma we give basic properties of this random sampling procedure, \Cref{lem:unif}, which is reminiscent of the sublinear sparsification result of \cite{Lee14arxiv}. 
After that we give the expander partitioning procedure we use, 	\Cref{thm:expanderdecomp} from  \cite{SW19} and our algorithm, \Cref{alg:pathsparsifier}. We then prove \Cref{thm:partial_path_sparsifiers} by showing that the expanders found have sufficiently many vertex disjoint paths by \Cref{thm:paths_in_expanders} and the right size properties by \Cref{thm:expanderdecomp} 

\begin{lemma}
\label{lem:unif}
Let $G = (V,E)$ be an unweighted $n$-node graph, $d \leq \dmin(G)$, and let $H$ be a graph constructed by sampling every edge from $G$ with probability $p = \min\{1 , \Theta(d^{-1} \log n) \}$
for some parameter $d \leq \dmin(G)$. With high probability in $n$
\[
\frac{1}{2} \mlap_H - \frac{pd}{n} \mlap_{K_n} \preceq p \mlap_G \preceq \frac{3}{2} \mlap_H + \frac{pd}{n} \mlap_{K_n}.
\]
and
\begin{equation}
\label{eq:uniform_sample_degree_bound}
\deg_H(a) \in \left[\frac{p}{2} \cdot \deg_G(a) ~ , ~ 2p \cdot \deg_G(a) \right]
\text{ for all } a \in V ~.
\end{equation}
\end{lemma}

\begin{proof}
Define auxiliary graph $G'$ with $\mlap_{G'} = \mlap_{G} + \frac{2d}{n} \mlap_{K_n}$. We observe that for any nodes $u,v$,
\[
\effres_{G'}(u,v) \leq \frac{n}{2d} \effres_{K_n}(u,v) = \frac{1}{d}.
\]
Thus, we interpret our sampling procedure to construct $H$ as sampling edges from $G'$ with the leverage score overestimates
\begin{equation*}
    \widetilde{\effres}(e)=
    \begin{cases}
      1/d & \text{if}\ e \in G \\
      1 & \text{if}\ e \in \frac{2d}{n} K_n.
    \end{cases}
\end{equation*}
We observe that these values are valid leverage score overestimates in $G'$: thus sampling and reweighting the edges in $G$ with probability $p$ and preserving the edges in $\frac{2d}{n} K_n$ produces a graph $H'$ with $\mlap_{H'} = \frac{1}{p} \mlap_{H} + \frac{2d}{n} \mlap_{K_n}$ 
such that $\frac{1}{2} \mlap_{H'} \preceq \mlap_{G'}  \preceq \frac{3}{2} \mlap_{H'}$ (see, e.g. Lemma 4 \cite{CohenLMMPS14arxiv} ).  Rearranging yields the desired 

\[
\frac{1}{2} \mlap_H - \frac{pd}{n} \mlap_{K_n} \preceq  p \mlap_G \preceq \frac{3}{2} \mlap_H + \frac{pd}{n} \mlap_{K_n} 
\]

Finally, \eqref{eq:uniform_sample_degree_bound} follows by an application of the Chernoff bound to the number of edges picked incident to each node in $G$. Since $d \leq \dmin(G)$ this number concentrates around its expected value with high probability in $n$ for appropriate choice of constant in the assumption of $p$.
\end{proof}

\begin{theorem}[Expander Decomposition \cite{SW19}]
	\label{thm:expanderdecomp}
	There is a procedure \newline $\ExpanderDecompose(G)$ which given any $m$-edge graph $G = (V,E)$, in time $O(m \log^7 m )$ with high probability outputs a partition of $V$ into $V_1, V_2, \cdots V_k$ such that 
	\begin{itemize}
		\item $\condedge(G\{V_i\}) \geq \Omega(1/\log^3(m))$ for all $i \in [k]$,
		\item $\sum_i |\partial_G(V_i)| \leq m/8$,
	\end{itemize}
	where  $G{\{V_i\}}$ denotes the induced subgraph of $G$ with self-loops added to vertices such that their degrees match their original degrees in $G$.  
\end{theorem}

\begin{proof}
This is a specialization of Theorem 1.2 from \cite{SW19} where $\phi$ in that theorem is chosen to be $\Theta(1/\log^3(m))$.
\end{proof}

\begin{algorithm2e}[h]
	\LinesNumbered
	\caption{$F = \PartialPathSparsify(G, k \geq 1)$} \label{alg:partialpathsparsify}

	\KwIn{$G = (V,E)$ simple input graph with degrees between $\dmin$ and $\dmax$}
	\KwOut{$(F, \Ecut)$ such that $F$ is a path sparsifier for $(V, E \setminus \Ecut)$ and $|\Ecut| \leq |E|/2$}
	
	\vspace{0.05in}
	\tcp{Sample edges uniformly at random to apply \Cref{lem:unif}}
	$d \defeq \frac{\dmin(G)}{10k}$ and $p = \min\{1 , \Theta(d^{-1} \log n) \}$ (where the constant in $\Theta$ is as in \Cref{lem:unif})\;
	\lIf{$p = 1$}{\Return $(E,\emptyset)$} \label{line:quick_return}
	$E' = \emptyset$ \;
	\lFor{$e \in E$}{
		Add $e$ to $E'$ with probability $p$.    
	}
	$G' = (V, E')$\;
	
	\vspace{0.05in}
	\tcp{Find expanders and use their edges as a path sparsifier}
	$(V_1, V_2, \cdots V_r) \gets \ExpanderDecompose(G')$\label{line:decompose} \tcp*{Computed via \Cref{thm:expanderdecomp}}
	For all $i \in [r]$ let $G_i = (V_i, E_i) \defeq G'[V_i]$\;
	Let $E_{\mathrm{cut}} \defeq \cup_{i\in[r]} \partial_G(V_i)$\;	
	Let $F \defeq \cup_{i \in [r]} E_i$\;
	\Return $(F, E_{\mathrm{cut}})$\;
\end{algorithm2e}

\begin{proof}[Proof of \Cref{thm:partial_path_sparsifiers}]
	
First, suppose that $p = 1$. In this case, by \Cref{line:quick_return} the algorithm outputs $F = E$ and $\Ecut = \emptyset$ in linear time. Consequently, $F$ is a path sparsifier of desired quality with  $m$ edges and $|\Ecut| \leq |E|/2$. Since, in this case  $d^{-1} \log n = \Omega(1)$ we have $\dmin(G) = O(k \log n)$ and the theorem follows as
\[
2 p m \leq O(n \dmax(g)) = O(n \dratio(G) \dmin(G)) = O(n k \cdot \dratio(G) \log(n)) ~,
\]
Consequently, in the remainder of the proof we assume $p < 1$.
		
Next, note that by design, $G'$ was constructed so \Cref{lem:unif} applies. Consequently, with high probability in $n$ the following hold:
\begin{itemize}
    \item $G'$ is an edge-subgraph of $G$ satisfying $p \mlap_G \preceq \frac{3}{2} \mlap_{G'} + \frac{pd}{n} \mlap_{K_n}$.
    \item $\deg_H(a) \in \left[\frac{p}{2} \cdot \deg_G(a) ~ , ~ 2p \cdot \deg_G(a) \right]$ for all $a \in V$.
    \item $G'$ contains at most $2pm$ edges. 
\end{itemize}

Leveraging the bounds we prove the path sparsification property and size bound for $F$. By \Cref{thm:expanderdecomp} we have that for all $i \in [k]$, the $V_i$ output by $\ExpanderDecompose$ satisfy that $\Phi_{G'\{V_i\}} = \Omega(1/\log^3(m))$ for all $i \in [k]$. Further, by the above properties of $G'$ we have that $\dmin(G'\{V_i\}) \geq \frac{p}{2} \cdot \dmin(G)$ and $\dmax(G'\{V_i\}) \leq 2p \cdot \dmax(G)$. Consequently, by \Cref{thm:paths_in_expanders} and the fact that $p < 1$ we have that $G'\{V_i\}$ has at least 
\[
\Omega \left( \frac{ \dmin(G) p }{ \dratio(G) \log^3(n) } \right)
= \Omega \left( 
\frac{ k }{ \dratio(G) \log^2(n) }
\right)
\]
vertex disjoint paths from $s$ to $t$ of length at most $O(\dratio(G) \log^4(n))$ for any $s,t \in V_i$. Further, since every edge in $(F, E \setminus E_{\mathrm{cut}})$ has both endpoints in some $V_i$ we see that $F$ is a path sparsifier as desired. Further, $F$ has at most 
\[
2pm = O(n \dmax(G) \log(n) k / \dmin) = O(n k \log(n) \dratio(G))
\]
edges by the properties above.

Next, we bound the size of $\Ecut$. Note that 
\begin{equation}
\label{eq:partial_sparse_size_1}
|\Ecut| 
= \frac{1}{2} \sum_{i \in [r]} |\partial_G(V_i)| 
= \frac{1}{2} \sum_{i \in [r]} v_i^\top \mlap_{G}(V_i) v_i
\end{equation}
where $v_i$ is the indicator vector for $V_i$, i.e. $v_i \in \R^V$ with $[v_i]_a = 1$ if $a \in V_i$ and $[v_i]_a = 0$ if $a \notin V_i$. Since, $p \mlap_G \preceq \frac{3}{2} \mlap_{G'} + \frac{p d}{n} \mlap_{K_n}$ we have that for all $i \in [r]$ that
\begin{equation}
\label{eq:partial_sparse_size_2}
p |\partial_G(V_i)|
= p v_i^\top  \mlap_G v_i
\leq \frac{3}{2} v_i^\top \mlap_{G'} v_i + \frac{pd}{n} 
= \frac{3}{2} |\partial_{G'}(V_i)| + \frac{pd}{n} |V_i| |V\setminus V_i| ~.
\end{equation}
Combining \eqref{eq:partial_sparse_size_1} and \eqref{eq:partial_sparse_size_2} yields that
\begin{align*}
|\Ecut| &\leq \frac{1}{2} \sum_{i \in [r]} 
\left(
 \frac{3}{2p} |\partial_{G'}(V_i)| + d |V_i|
\right)
\leq 
\frac{3}{32p}  |E'| + dn
\leq \frac{3m}{16} + \frac{\dmin(G) n}{10 k}
\leq \frac{|E|}{2}
\end{align*}
where in the third to last step we used that $\sum_{i \in [r]} |\partial_{G'}(V_i)| \leq |E'|/8$ by \Cref{thm:expanderdecomp}, in the second to last step we used that $|E'| \leq 2pm$ and that $d = \dmin(G) / (10k)$, and in the last step we used that $\dmin(G) n \leq |E|$.

Finally, we bound the running time of the algorithm. Every line except for \Cref{line:decompose} in our algorithm is a standard graph operation that can be implemented in linear total time. The call to $\ExpanderDecompose$ on \Cref{line:decompose} is on a graph with $O(mp)$ edges. Consequently, by \Cref{thm:expanderdecomp} the call takes time
\[
O(mp \log^7(m)) 
= O(m (\log(n) / d) \log^7(n))
= O(n k \cdot \dratio(G) \log^8(n))
\] 
where the first equality used the simplicity of $G$ and that $p = \Theta(d^{-1} \log n) \leq 1$ and the second equality used that $m \leq n \dmax(G)$ and the definition of $d$. 
\end{proof}

\newcommand{\DegreeLowerBound}{\mathsf{DegreeLowerbound}}
\newcommand{\BipartiteSplit}{\mathsf{BipartiteSplit}}
\newcommand{\RegularDecomposition}{\mathsf{RegularDecomp}}
\newcommand{\DecomposeBipartite}{\mathsf{BipartiteDecomp}}

\subsection{Approximately Degree Regular Graph Decompositon}
\label{sec:deg_reg_decomp}


Here we provide a linear time procedure to decompose a constant fraction of a dense graph into a  nearly-regular dense pieces supported on a bounded number of vertices. We use degree regularity to turn edge expansion bounds on a graph into vertex expansion bounds and finding vertex disjoint paths in \Cref{sec:short_expander_paths} which in turn we use to efficiently compute path sparsifiers. 

The main result of this section is the following theorem on computing such a decomposition.

\begin{theorem}[Regular Decomposition]
	\label{thm:reg_decomp}
	Given $n$-vertex $m$-edge simple undirected graph $G = (V,E)$ with $\davg(G) \geq 2000 (\ln (2n))^2$,  $\RegularDecomposition(G)$ (\Cref{alg:reg_decomp}) in expected $O(m)$ time
	outputs graphs $H_1 =(V_1, E_1), ... , H_\ell = (V_\ell, E_\ell)$ which are edge disjoint subsets of $G$ such that
	\begin{enumerate}
		\item (Vertex Size Bound): $\sum_{i \in \ell} |V_i| \leq 4n \ln n $.
		\item (Volume Lower Bound): $\sum_{i \in \ell} \vol(E_i) \geq \vol(G) / 100$
		\item (Degree Regularity Bound): $\dratio(H_i) \leq 1000 (\ln(2n))$ for all $i \in [\ell]$
		\item (Minimum Degree Bound): $\dmin(H_i) \geq \davg(G) / (250 \ln n)$ for all $i \in [\ell]$.\footnote{This condition is not use for our path sparsification construction, but is included due to its possible utility.}
	\end{enumerate} 
\end{theorem}

We build $\RegularDecomposition$ and prove \Cref{thm:reg_decomp} in several steps. First we provide $\DegreeLowerBound$ (\Cref{alg:degreelower}) which simply removes vertices of degree less than a multiple of the average. It is easy to show (\Cref{lem:deg_lower}) that this procedure runs in linear time, doesn't remove too many edges, and ensures that the minimum degree is a multiple of the average. 

Leveraging $\DegreeLowerBound$  we provide two procedures, $\BipartiteSplit$ (\Cref{alg:bipartitesplit}) and $\DecomposeBipartite$ (\Cref{alg:bipartitedecompse}) which together, show how to prove a variant of \Cref{thm:reg_decomp} on bipartite graphs where the max degree is not too much larger than the average degree for each side of the bipartition. 
We show that provided these average degrees are sufficiently large,  $\BipartiteSplit$ (\Cref{alg:bipartitesplit}) splits the graph into pieces of roughly the same size where the degrees on the larger side are preserved up to a multiplicative factor (See \Cref{lem:bipartite_split}). This procedure simply randomly partitions one of the sides and the result follows by Chernoff bound. The procedure $\DecomposeBipartite$ (\Cref{alg:bipartitedecompse}) then carefully applies $\BipartiteSplit$ and $\DegreeLowerBound$.

Our main algorithm, $\RegularDecomposition$ (\Cref{alg:reg_decomp}) operates by simply bucketing the vertices into to groups with similar degree and considering the subgraphs of edges that only go between pairs of these buckets. 
The algorithm then applies either $\DecomposeBipartite$  or  $\DegreeLowerBound$ to subgraphs of sufficiently high volume and analyzing this procedure proves \Cref{thm:reg_decomp}.

\begin{algorithm2e}[h]
	\caption{$\{G_i\}_{i\in[k]} = \DegreeLowerBound(G,c)$} \label{alg:degreelower}
	\LinesNumbered
	\KwIn{Graph $G = (V,E)$, parameter $c \in (0, 1)$}
	Let $\davg := 2m/n$, $S = V$, $R = \emptyset$\;
	\While{$\dmin(G[S]) < c \cdot \davg$}{
		Pick $a \in S$ with $\deg_{G[S]}(a) < c  \cdot \davg$\;
		$S := S \setminus \{a\}$\; 
	}
	\Return $G[S]$ \;
\end{algorithm2e}

\begin{lemma}[Degree Lower Bounding]
\label{lem:deg_lower}
For any $n$-node $m$-edge graph $G = (V,E)$ and $c \in (0, 1)$, $\DegreeLowerBound(G,c)$ (\Cref{alg:degreelower}) outputs $G[S]$ for $S \subseteq V$ such that 
\[
\Vol(G[S]) \geq (1- c) \Vol(G)
\text{ and }
\dmin(G[S]) \geq c \cdot \davg(G)
\]
\end{lemma}

\begin{proof}
This procedure can be implemented in linear time by storing $\deg_{G[S]}(a)$ for all $a \in V$ and updating it in $O(1)$ per edge when a vertex is removed. Further, $\dmin(G[S]) \geq c \davg(G)$ by design of the algorithm. Finally, every time we remove a vertex $a$ from $S$ we remove at most $c \davg$ edges from $G[S]$. Consequently, the final $S$ satisfies $|V\setminus S| \cdot c \davg \leq n \cdot c \cdot \davg = 2 c m = c \cdot \vol(G)$ and $\vol(G[S]) \geq (1- c) \vol(G)$.

\end{proof}

\SetKwRepeat{Do}{do}{while}

\begin{algorithm2e}[h]
\LinesNumbered
\KwIn{Bipartite graph $G = (V,E)$, bipartition $(L,R)$ of $V$, $k \in [1,|L|]$}
\tcp{Assume for all $e = (a,b) \in E$, $a \in L$, $b \in R$}
\tcp{Assume $\deg_G(b) / k \geq 20 \ln(2n)$ for all $b \in R$ and $|L| / k \geq 20 \ln(2n)$}
\Do{$\deg_{G_i}(b) \notin [\frac{\deg_G(b)}{2k}, \frac{3 \deg_G(b)}{2k}]$ or $|L_i| \notin [\frac{|L|}{2k}, \frac{3|L|}{2k}]$  for some $b \in R$, $i \in [k]$ }{
Let $L_i = \emptyset \subseteq L$ for all $i \in [k]$ \;
For each $a \in L$ add $a$ to $L_i$ for $i \in [k]$ uniformly, independently at random \;
Let $G_i = G[L_i \cup R]$ for all $i \in [k]$\;
}
\Return $\{G_i\}_{i \in [k]}$ \;
\caption{$\{G_i\}_{i\in[k]} = \BipartiteSplit(G,(L,R),k)$} \label{alg:bipartitesplit}
\end{algorithm2e}

\begin{lemma}[Bipartite Graph Splitting] 
\label{lem:bipartite_split}	
Given $n$-node $m$-edge simple bipartite graph $G = (V,E)$ with bipartition into $(L,R) \subseteq V$ and $k \in [1, |L|]$ where
\[
\frac{\deg_G(b)}{k} \geq 20 \ln(4n) \text{ for all } b \in R \text{ and }
\frac{|L|}{ k} \geq 20 \ln(4n)
\]
 $\mathsf{BipartiteSplit}(G, (L,R), k)$ (\Cref{alg:bipartitesplit}) in expected $O(m)$ time
 outputs $\{G_i\}_{i\in[k]}$ that partition the edges such  that for all $i \in [k]$ and $b \in R$
\begin{equation}
\label{eq:success_condition}
\deg_{G_i}(b) \in \left[ \frac{\deg_G(b)}{2k} , \frac{3\deg_G(b)}{2k} \right]
\text{ and } 
|L_i| \in \left[ \frac{|L||}{2k} , \frac{3|L|}{2k} \right]
\end{equation}
\end{lemma}

\begin{proof}
Note that each loop of the algorithm clearly takes linear time and if the algorithm terminates its output is as desired. Consequently, it suffices to show that the probability the loop repeats is bounded by some fixed constant probability.
 
Let $\mu = \min\{\min_{b \in R} \deg_G(b) / k , |L|/k\}$. Further, for each $a \in L$ and $i \in [k]$ let $x_{i,a}$ be a random variable set to $1$ if $a \in L_i$ and $0$ otherwise. Now note that for all $b \in R$ we have by the fact that $\E[x_{i,a}] = 1/k$ we have 
\[
\deg_{G_i}(b) = \sum_{a \in N_G(b)} x_{i,a}
\text{ and }
\E \left[ \deg_{G_i}(b) \right] = \frac{\deg_G(b)}{k} \geq \mu \cdot  ~.
\]
Consequently, by Chernoff bound we have that for all $b \in R$ and $i \in [k]$
\[
\Pr\left[  \deg_{G_i}(b) \leq \frac{\deg_G(b)}{2k}  \right] \leq \exp\left(-\frac{\mu}{8}\right)
\text{ and }
\Pr\left[  \deg_{G_i}(b) \geq \frac{3\deg_G(b)}{2k}  \right] \leq \exp\left(-\frac{\mu}{10}\right)
\]
Further, for all $i \in [k]$, by the same reasoning (e.g. suppose there was a vertex in $R$ of degree $|L|$),
\[
\Pr\left[ |L_i| \leq \frac{|L|}{2k}  \right] \leq \exp\left(-\frac{\mu}{8}\right)
\text{ and }
\Pr\left[ |L_i| \geq \frac{3 |L|}{2k}  \right] \leq \exp\left(-\frac{\mu}{10}\right)
\]
Now since by assumption $\mu \geq 20 \ln(4n)$, by applying union bound to the $2(|L| + 1) \cdot k$ different reasons the loop in $\BipartiteSplit$ might repeat, the loop repeats with probability at most
\[
2 (|L| + 1) \cdot k \cdot \exp(- \mu / 10) \leq 2(|L| +1) \cdot k \cdot \frac{1}{8n^2} \leq \frac{1}{4}
\]
where we used that since the graph is non-empty $|L| + 1 \leq n$.
\end{proof}

\begin{lemma}[$\DecomposeBipartite$ (\Cref{alg:bipartitedecompse})]
	\label{lem:decomp_bipartite}
	Let $G = (V,E)$ be an arbitrary simple bipartite graph with bipartition $(L,R)$ such that  $\davg^L \defeq \vol_G(L) / |L| $ and $\davg^R \defeq \vol_G(R) / |R|$ satisfy $\davg^L \geq 40 \ln(2n)$ and $\davg^R \geq 40 \ln(2n)$. In expected $O(m)$ time, $ \DecomposeBipartite(G,(L,R))$ outputs graphs $H_1 =(V_1, E_1), ... , H_k = (V_k, E_k)$ which are edge disjoint subsets of $G$ with
	\begin{enumerate}
		\item (Vertex Size Bound): $\sum_{i \in \ell} |V_i| \leq 4 n$.
		\item (Volume Lower Bound): $\sum_{i \in \ell} \vol(E_i) \geq \vol(G) / 8$
		\item (Degree Regularity Bound): $\dratio(H_i) \leq 16 c$ for all $i \in [k]$ where \[
		c \defeq \max\{ \dmax^L / \davg^L, \dmax^R / \davg^R  \}
		\]
		for  $\dmax^L \defeq \max_{a \in L} \deg_G(a)$ and $\dmax^R \defeq \max_{a \in R} \deg_G(a)$.
		\item (Minimum Degree Bound): $\dmin(H_i) \geq \min\{\davg^L , \davg^R \} / 16$.
	\end{enumerate} 
\end{lemma}

\begin{algorithm2e}[h]
\LinesNumbered
\KwIn{Bipartite graph $G = (V,E)$ and bipartition $(L,R)$ of $V$}
\tcp{Assume $\vol(L) / |L| \geq 40 \ln(2n)$ and $\vol(R) / |R| \geq 40 \ln(2n)$ }
Let $\davg^L := \vol(L) / |L| $ and $\davg^R := \vol(R) / |R|$ \;
Swap $L$ and $R$ if needed so that $|R| \leq |L|$ \;
Let $R' = \{a \in R | \deg_G(a) \geq (1/2) \davg^R \}$ and $G' = G(L\cup R')$\;
\lIf{$|R'| \geq  |L| /2$}{
	\Return $\DegreeLowerBound(G, 1/2)$ \label{line:decompbipartite:early_return}
}
$\{G_i\}_{i \in [k]} = \BipartiteSplit(G',(L,R'),k)$ for $k = \lfloor |L| / |R'| \rfloor$ \;
\Return $\{H_i\}_{i \in [k]}$ where $H_i =  \DegreeLowerBound(G_i, 1/2)$\;
\caption{$\{H_i\}_{i\in[k]} = \DecomposeBipartite(G,(L,R))$} \label{alg:bipartitedecompse}
\end{algorithm2e}

\begin{proof}
First, note that $G'$ is a vertex induced subgraph of $G$ where only vertices in $R$ with degree less than half the average in $R$ are removed. Since $\sum_{a \in R \setminus R'} \deg_G(A) \leq |R| \davg^R / 2 \leq \vol_G(R)$ it follows that $\vol_{G'}(R') \geq \frac{1}{2} \vol_G(R)$. Further, since $G$ is bipartite this implies that $\vol(G') \geq \vol(G) / 2$.

Next, suppose the algorithm returns on \Cref{line:decompbipartite:early_return} (i.e. $|R| < |L|/2$ ). In this case, \Cref{lem:deg_lower} (which analyzes $\DegreeLowerBound$) implies that the returned graph, which we denote $H$, is a vertex induced subgraph of $G'$ with $\vol(H) \geq \vol(G') / 2 \geq \vol(G) / 4$ and $\dmin(H) \geq \davg(G') / 2 \geq \davg(G)/ 4$. Since $\davg(G) \geq \min\{\davg^L, \davg^R\}$, this immediately yields the desired vertex size bound, vertex lower bound, and the minimum degree bound. Further, since the graph is bipartite we have $\davg^L / \davg^R = |R| / |L|$. Therefore, since $|L|/2 \leq |R| \leq L$ we have $\davg^R / 2 \leq \davg^L \leq \davg^R$ and $\davg(G) \geq \davg^L \geq \davg^R / 2$. Consequently
\[
\dratio(H) \leq \frac{\dmax(G)}{\davg(G) / 4}
\leq 4  \cdot \max\left\{ \frac{\dmax^L}{\davg(G)} , \frac{\dmax^R}{\davg(G)}  \right\}
\leq 4 
\max\left\{ \frac{\dmax^L}{\davg^L} , \frac{\dmax^R}{(1/2 )\davg^R}  \right\}
\leq 8 c
\]
and the result holds in this case. 
 
Therefore, in the remainder of the proof we assume instead that $|R'| \leq |L| / 2$. Further, since $G$ is bipartite we know that $\davg^L \geq 2 \davg^R $. Note that this implies that $|L|/|R| \geq 2$ and therefore $k  \in [|L|/(2|R'|), |L|/|R'|]$. Note that the average degree of a vertex in $L$ in $G'$ is at least $\davg^L / 2$, by our reasoning regarding $G'$ and consequently, since the graph is simple $|R'| \geq \davg^L / 2$. This implies  
\[
\frac{|L|}{k} \geq |R'| \geq \frac{|R|}{2} \geq \frac{\davg^L}{2} \geq 20 \ln(2n)
\]
where the last inequality follows from the assumption on the input.  Further, by design we have that for all $b \in R'$
\[
\deg_{G'}(b)  \geq \davg^R / 2 \geq 20 \ln(2n) ~.
\]
Consequently, \Cref{lem:bipartite_split} applies to $\BipartiteSplit(G',(L,R'),k)$ and for all $i \in [k]$ and $b \in R'$ 
\begin{equation}
\deg_{G_i}(b) \in \left[ \frac{\deg_{G'}(b)}{2k} , \frac{3\deg_{G'}(b)}{2k} \right]
\text{ and } 
|L_i| \in \left[ \frac{|L|}{2k} , \frac{3|L|}{2k} \right]
~.
\end{equation}

Now let $\davg^{L'} \defeq \vol_{G'(L)} / |L| $ and $\davg^{R'} \defeq \vol_{G'}(R') / |R'|$. Note that $|L_i| \leq 3 |L| / (2k) \leq 3 |R'|$ and $k \leq |L|/|R'| = \davg^{R'} / \davg^{L'}$. This implies
\[
\davg(G_i) = \frac{2}{|R'| + |L_i|} \sum_{b \in R'} \deg_{G_i}(b)
\geq \frac{1}{2 |R'|} \sum_{b \in R'} \frac{\deg_{G'}(b)}{2k} 
\geq 
\frac{\davg^{R'}}{4k} ~.
\]
Since the average degree of a vertex in $R'$ in $G'$ is at least the average degree of a vertex in $R$ in $G$, we have $\davg^{R'} \geq \davg^{R}$. Further, we have $\davg^L \leq \davg^{L'} \leq \davg^{R'} / k$ by the construction of $G'$. This implies  
\[
\dmax(G_i)
\leq \max\{ (3/(2k)) \dmax^R , \dmax^L  \}
\leq c \cdot \max\{ (3/(2k)) \davg^{R'} , \davg^{R'} / k  \}
\leq 2c \cdot \davg^{R'} / k 
\] 
and $\dratio(G_i) \leq 8c$ . Further, since $\davg^{R'} \geq \davg^{L'} k$ and $\davg^{L'} \geq \davg^{L} / 2$ by the construction of $G$, we have that $\davg(G_i) \geq \davg^L / 8$. Consequently, the volume lower bound, degree regularity bound, and minimum degree bound follow from the fact that invoking $\DegreeLowerBound(G_i, 1/2)$ only removes edges and decreases the volume by at most a factor of $2$ and decreases the min degree to at most $(1/2)$ the average by \Cref{lem:deg_lower}. Finally, the vertex size bound followed from the bound on $k$ and that the only vertices repeated are $R'$ which are repeated at most $k$ times.
\end{proof}

We now have everything we need to present $\RegularDecomposition$ (\Cref{alg:reg_decomp}), our graph decomposition algorithm, and analyze it to prove  \Cref{thm:reg_decomp}, the main result of this section.

\begin{algorithm2e}[h]
	\LinesNumbered
	\KwIn{Undirected, unweighted, connected graph $G = (V,E)$ with $n$-vertices and $m$-edges}
	$G' = (V', E')$ for $G' := \DegreeLowerBound(G, 1/2)$\;
	Let\footnote{We use $e$ rather than $2$ to define $k$ in terms of $\ln$ and help compatibility with \Cref{lem:bipartite_split}.} $S_i := \{a \in V' ~ | ~ \deg_G'(a) \in [e^{i - 1}, e^i)  \}$ for all $i \in [1, k]$ where $k \defeq \lfloor \ln(n) \rfloor$\;
	$\mathcal{H}_{out} := \emptyset$ \;
	\For{  $i,j \in [k]$ with $i \leq j$ (including $i = j$) } {
		Let $V_{i,j} := S_i \cup S_j$ and $E_{i,j} := \{\{a,b\} \in E' ~ | ~ a \in S_i, b\in S_j\}$ \;
		Let $G_{i,j} := (S_i \cup S_j, E_{i,j})$ \;
		\If{$\vol(G_{i,j}) \geq \vol_{G'}(S_i) / (2 \ln n)$ and $\vol(G_{i,j}) \geq \vol_{G'}(S_j) / (2 \ln(n))$ \label{line:vol_lower} }{
			\lIf{$ i = j$}{
				$\mathcal{H}_{out} := \mathcal{H}_{out} \cup \{\DegreeLowerBound(G_{i,j}), 1/2\}$
			}
			\lElse{
				$\mathcal{H}_{out} := \mathcal{H}_{out} \cup \DecomposeBipartite(G_{i,j},(S_i,S_j))$ 
				\label{line:decomp_bipartite_apply}
			}
		}
	}
	\Return all output graphs computed
	\caption{$\{G_i\}_{i \in[k] } = \mathsf{RegularDecomposition}(G)$} \label{alg:reg_decomp}
\end{algorithm2e}


\begin{proof}[Proof of \Cref{thm:reg_decomp}]
First we show that whenever $\DecomposeBipartite(G)$ is invoked by the algorithm in \Cref{line:decomp_bipartite_apply} on $G_{i,j}$ then  \Cref{lem:decomp_bipartite} applies with $c \leq  4e \ln n$, $\davg^L \geq 40 \ln(2n)$, $\davg^R \geq 40 \ln(2n)$. Fix an invocation of $\DecomposeBipartite(G)$ on \Cref{line:decomp_bipartite_apply} for $i \neq j$ and let $L = S_i$, $R = S_j$, and $\davg^L \defeq \vol_{G_{i,j}}(S_i) / |S_i|$ and $\davg^R \defeq \vol_{G_{i,j}}(S_j) / |S_j|$. By design, for all $a \in S_i$ and $b \in S_j$ 
\[
\deg_{G'}(a) \in [e^{i -1}, e^{i}]
\text{ and }
\deg_{G'}(b) \in [e^{j - 1}, e^{j}] ~.
\]
Therefore, by the guarantees of \Cref{lem:deg_lower} for $G' := \DegreeLowerBound(G, 1/2)$ 
\[
e^{i - 1} \geq e^{-1} \dmin(G') \geq (1/(2e) )\davg(G) \geq (1000/e) (\ln(2n))^2 ~.
\]
Further, since $\vol(G_{i,j}) \geq \vol_{G'}(S_i) / (2 \ln n)$ and $\Vol_{G'}(S_i) \geq |S_i| e^{i - 1} / 2$ by the degree bounds of $a \in S_i$ we see that 
\[
\davg^L  \defeq \frac{\vol_{G_{i,j}}(S_i) }{ |S_i| } 
\geq \frac{\vol_{G'}(S_i)}{|S_i| \cdot 2 \log n }
\geq \frac{e^{i - 1}}{4 \ln n} \geq \frac{1000 \ln(2n)^2}{4e \ln n} 
\geq 40 \ln(2n) ~.
\]
By the same reasoning $\davg^R \geq 40 \ln(2n)$. To bound $c$, note that deleting edges can only decrease degree and therefore
\[
\max_{a \in L} \frac{\deg_{G_{i,j}}(a)}{\davg^L}
\leq \max_{a \in L} \frac{\deg_{G'}(a)}{(e^{i - 1} / (4 \log(n)))}
\leq \frac{ e^{i} }{(e^{i - 1} / (4 \cdot \log n))}
= 4 e \log n ~.
\]
Since by symmetry the same bound holds for $R$, the desired bound for $c$ holds. 

\emph{(Vertex Bound)}: Note that a vertex can appear in at most $k$ different $G_{i,j}$. Consequently, the result follows by \Cref{lem:deg_lower} and  \Cref{lem:decomp_bipartite}.

First note that by \Cref{lem:deg_lower} we have $\vol(G') \geq (1/2) \vol(G)$ and $\dmin(G') \geq (1/2) \davg(G)$.

\emph{(Volume Bound)}: Note that by \Cref{lem:deg_lower} we have $\vol(G') \geq (1/2) \vol(G)$. Further, define the set $P \defeq \{(i,j) \in [k] \times [k] ~ | ~ i \leq j \}$ and let 
\[
P_{\geq} \defeq \{ (i,j) \in P ~ | ~ \vol(G_{i,j}) \geq \vol_{G'}(S_i) / (2 \ln n)\text{ and }\vol(G_{i,j}) \geq \vol_{G'}(S_j) / (2 \ln n ) \} ~.
\]  
Note that $P$ contains the indices for every $G_{i,j}$ considered and that $P_{\geq}$ denotes the subset of them for which the volume of $G_{i,j}$ is large enough that \Cref{line:vol_lower} is true. By design,
\begin{align*}
\sum_{(i,j) \in P \setminus P_{\geq}} \vol(G_{i,j}) 
&< \sum_{(i,j) \in P} \left( 
\max \left\{
\frac{ \vol(G'(S_j)) }{2 \ln n} , \frac{ \vol(G'(S_i)) }{2 \ln n}
\right\}
\right) \\
&\leq 
\frac{k}{2 \ln n}  \sum_{i \in [k]} \vol(G'(S_j)) \leq \frac{1}{2} \vol(G') ~.
\end{align*}
Since every edge $e \in E'$ is in some $G_{i,j}$ this then implies that
\[
\sum_{(i,j) \in P_{\geq}} \vol(G_{i,j}) = \sum_{(i,j) \in P} \vol(G_{i,j}) - \sum_{(i,j) \in P \setminus P_{\geq}} \vol(G_{i,j})  
\geq \vol(G') -\frac{1}{2} \vol(G') \geq \frac{1}{4} \vol(G) ~.
\]
Since our output is simply the result of invoking $\DegreeLowerBound$ and \newline $\DecomposeBipartite$ on these graphs and by \Cref{lem:deg_lower} and \Cref{lem:decomp_bipartite} and these procedures decrease the volume by at most a factor of 8, the result follows. 

\emph{(Degree Regularity Bound)}: For graph $G_{i,j}$ with $i \neq j$ this follows from the bound on $c \leq  4e \ln n$ given in the first paragraph of this proof, \Cref{lem:decomp_bipartite}, and that $16 \cdot 4e \ln n \leq 1000 \ln(2n)$. For graph $G_{i,j}$ with $i = j$, the same reasoning implies the ratio of the maximum degree to the average degree is at most $c$ and the result follows by  \Cref{lem:deg_lower}, which shows that $\DegreeLowerBound(G_{i,j}, 1/2)$ only decreases the maximum degree and makes the minimum degree is at least half the average degree. 

\emph{(Minimum Degree Bound)}: By the reasoning of the first paragraph of this section we know that whenever $\DecomposeBipartite(G)$ is invoked by the algorithm in \Cref{line:decomp_bipartite_apply} on $G_{i,j}$  then 
\[
\davg^L \geq 
\frac{\vol_{G'}(S_i)}{|S_i| \cdot 2\log n}
\geq \frac{\dmin(G')}{4 \log n}
\geq \frac{\davg(G)}{8 \log n}
\]
where in the last step we used that the average degree in $G'$ is at least the average degree in $G$ by  \Cref{lem:deg_lower}. Consequently, the result follows again by  \Cref{lem:deg_lower}, \Cref{lem:decomp_bipartite}, and the fact that $16 \cdot 8  \leq 250$.
\end{proof}

\newcommand{\Eremain}{E_{\mathrm{remain}}}
\newcommand{\kpartial}{k_{\mathrm{partial}}}

\subsection{Putting it All Together}
\label{sec:path-sparse-proof}

Here we show how to put together all the results of the previous subsection to prove  \Cref{thm:path-sparse}. Our algorithm, $\PathSparsify$ (\Cref{alg:pathsparsifier}) simply performs a regular decomposition of the input graph by \Cref{alg:reg_decomp} (\Cref{thm:reg_decomp}) of \Cref{sec:deg_reg_decomp} and then performs of partial path sparsification of each of these graphs by \Cref{alg:partialpathsparsify} (\Cref{thm:partial_path_sparsifiers}) of \Cref{sec:expander-compute}. In the remainder of this section we provide and analyze $\PathSparsify$ (\Cref{alg:pathsparsifier}) to prove \Cref{thm:path-sparse} (restated below for convenience).

\pathsparse*

\begin{algorithm2e}[h]
	\LinesNumbered
	\caption{$F = \PathSparsify(G, k \geq 1)$} \label{alg:pathsparsifier}
	\KwIn{$G = (V,E)$ simple input graph with degrees between $\dmin$ and $\dmax$}
	\KwOut{$F \subseteq E$ with $|F| = O(n k \log^3(n))$ such that $G[F]$ is a $(k, O(\log^5 n))$-path-sparsifier of $G$}
	
	$\kpartial = \Theta(\log^3 n)$ \tcp*{For constants in $\kpartial$ see \Cref{thm:path-sparse} proof}
	$\Eremain \gets E$, $F \gets \emptyset$\;
	
	\While{$\davg(G(\Eremain)) \geq 2000(\log(2n))^2$}{ \label{line:path_sparsify:while_start}
		$\{H_i =(V_i, E_i)\}_{i \in [\ell]} \gets \mathsf{RegularDecomposition}(G)$ 
		\tcp*{\Cref{alg:reg_decomp} (\Cref{thm:reg_decomp})}
		\For{$i \in [\ell]$} {
			$(F^{(i)}, \Ecut^{(i)}) \gets \PartialPathSparsify(H_i, \kpartial)$ \tcp*{\Cref{alg:partialpathsparsify} (\Cref{thm:partial_path_sparsifiers})} \label{line:path_sparsify:sparsify}
		} 
		$F \gets \cup_{i \in [\ell]} F^{(i)}$ and $\Eremain \gets \cup_{i \in [\ell]} \Ecut^{(i)}$ \; 	}\label{line:path_sparsify:while_end}
	 $F \gets F \cup \Eremain$\;
	\Return $F$\;
\end{algorithm2e}

\begin{proof}[Proof of \Cref{thm:path-sparse}]

We first consider the execution of a single loop of \Cref{alg:pathsparsifier}, i.e. \Cref{line:path_sparsify:while_start} to \Cref{line:path_sparsify:while_end}. Since $\davg(G(\Eremain)) \geq 2000(\log(2n))^2$ 
we can apply \Cref{lem:deg_lower} to analyze the execution of $\RegularDecomposition$.  \Cref{lem:deg_lower} implies that this line takes expected $O(|\Eremain|)$ time and outputs $\{H_i =(V_i, E_i)\}_{i \in [\ell]}$ such that at least constant fraction of the edges of $\Eremain$ are in the $E_i$, $\sum_{i \in \ell} |V(H_i)| = O(n \ln^2 n)$, $\dratio(H_i)) = \Omega(\log(2n))$, and $\dmin(H_i) = \Omega(\davg(G) /  \log n)$ for all $i\in [\ell]$ (where we used that the number of vertices in $|\Eremain|$ is at most $n$). Consequently, in a single execution of the while loop, \Cref{thm:partial_path_sparsifiers} shows that \Cref{line:path_sparsify:sparsify} takes time
\begin{equation}
\label{eq:path_sparse_time}
O\left(
\sum_{i \in [\ell]} |E_i| + |V_i| \kpartial \dratio(H_i) \log^8(n)
\right)
=
O\left(|\Eremain| + n \kpartial \log^9(n)\right) 
\end{equation}
and w.h.p. in $n$ outputs $\{(F^{(i)}, \Ecut^{(i)})\}_{i \in [\ell]}$ such that 
\[
\sum_{i \in [\ell]} |F^{(i)}| = O\left(\sum_{i \in [\ell]} |V_i| \kpartial \cdot \dratio(H_i) \log(n) \right)
= O\left(n \kpartial \log^2(n) \right)
\]
each $H_i(F)$ is a $(\Omega(\kpartial/ \log^3(n)) , O(\log^5 n) )$-path sparsifier of $(V_i, E_i \setminus \Ecut^{(i)} )$ and
 $\sum_{i \in [\ell]} | \Ecut^{(i)} | \leq c |\Eremain|$ for some constant $c \in (0, 1)$. 
 
The preceding paragraph ultimately shows that w.h.p. each iteration of the loop, i.e. \Cref{line:path_sparsify:while_start} to \Cref{line:path_sparsify:while_end}, takes expected time $O(|\Eremain| + n \kpartial \log^9(n))$ to output a $(\Omega(\kpartial/ \log^3(n)) , O(\log^5 n) )$-path sparsifier on a constant fraction of the edges.
Consequently, the loop terminates in $O(\log |E|) = O(\log n)$ iterations. Note that when the loop terminates $\davg(G(\Eremain)) < 2000(\log(2n))^2$ and therefore  $|\Eremain| = O(n \log^2(n)) = O(n \kpartial \log^2(n))$. Consequently, $F$ returned by the algorithm has size at most $O(n \kpartial \log^3 n)$. Further, note that if for any edge-disjoint graphs $\bar{G}_i = (V, \bar{E}_i)$ for $i \in [k]$ and $\bar{F}_i \subseteq \bar{E}_i$ it is the case that each $\bar{G}[\bar{F}_i]$ is an $(\alpha,\beta)$-path sparsifier of $\bar{G}_i$ then $\bar{G}[\cup_{i \in [k]} \bar{F}_i]$ is an $(\alpha, \beta)$-path sparsifier of $\bar{G} = (V, \cup_{i \in [k]} \bar{E}_i)$. Consequently, $G[F]$ is an $(\Omega(\kpartial/ \log^3(n)) , O(\log^5 n) )$-path-sparsifier of $G$. By choice of constants in the setting of $\kpartial = \Theta(k \log^3 n)$ we have that the output of the algorithm is as desired. Further, the run time follows from the reasoning in the preceding paragraph about the runtime of a single loop (i.e. \eqref{eq:path_sparse_time}) and that the number of edges in $\Eremain$ decrease by a constant in each iteration. 

\end{proof}

\newcommand{\Sample}{\mathsf{Sample}}
\newcommand{\solveeps}{\varepsilon_{solve}}
\newcommand{\solve}{\mathsf{Solve}}
\newcommand{\runtime}{\mathcal{T}}
\newcommand{\solvertime}{\widetilde{\runtime}}
\newcommand{\ElimSolve}{\mathsf{EliminateAndSolve}}
\newcommand{\LowStretch}{\mathsf{LowStretch}}
\newcommand{\AGD}{\mathsf{PreconNoisyAGD}}
\newcommand{\PreconRichardson}{\mathsf{PreconRichardson}}
\newcommand{\Recursive}{\mathsf{RecursiveSolver}}
\section{Laplacian Solvers with Low-Stretch Subgraphs}
\newcommand{\mz}{\mathbf{Z}}
\label{sec:solver}
In this section, we prove our main theorem regarding algorithms with improved running times for solving Laplacian linear systems. We show how to use the low distortion spectral subgraphs developed in \Cref{sec:spectral_stretch} to prove the following theorem.

\solver*

Our proof is based on an analogous claim in \cite{CohenKMPPRX14} regarding different spectral subgraph guarantees. Several proofs in this section are adaptations of lemmas from  \cite{CohenKMPPRX14} to our setting.  We provide the proof in full here both for completeness and because the specific guarantees of \cite{CohenKMPPRX14} do not tolerate the extra edges in our preconditioners coming from our path sparsifiers. In addition our analysis based on noisy accelerated gradient descent is slightly tighter than that of \cite{CohenKMPPRX14}, enabling us to obtain an improved guarantee.

A key matrix fact we apply in this section is a slight extension of a claim from \cite{CohenKMPPRX14} regarding $\Sample$ (\Cref{alg:sample}), a matrix sampling procedure from \cite{CohenKMPPRX14}.

\begin{algorithm2e}[t]
	\LinesNumbered
	\KwIn{$\mathbf{Y}_i = v_i v_i^\top$ are rank one matrices, $\tau_i$ are upper bounds of leverage scores, i.e. $\tau_i \geq \Tr[\mathbf{Y}_i \mathbf{X}^\dagger]$ for all $i$, and $\delta < 1$ is an arbitrary parameter.}
	\KwOut{Matrix $\mathbf{X}$ satisfying conditions of Lemma~\ref{lemma:sample}.}
	$\mathbf{Z} \gets \mathbf{X}$, 
	$s = \sum_{i \in [m]} \tau_i, t = \delta^{-1} s$ \\
	$r \gets $ randomly chosen integer in $[t, 2t-1]$ \\
	\For{$j = 1, 2 \ldots r$}{
		Pick index $i$ with probability proportional to $\tau_{i}$ \\
		$\mathbf{Z} \gets \mathbf{Z} + \frac{\delta}{\tau_{i}} \mathbf{Y}_{i}$ \\
	}
	\Return{$\mathbf{Z}$}
	\caption{$\mathbf{Z} = \Sample(\{ \mathbf{Y}_1, ..., \mathbf{Y}_m\}, \mathbf{X}, \mathbf{\tau}, \delta)$ (from \cite{CohenKMPPRX14})} \label{alg:sample}
\end{algorithm2e}

\begin{lemma}[Adaptation of Lemma 2.3 from \cite{precon-exp}]
\label{lemma:sample}
Suppose $\mathbf{X}$ and $\mathbf{Y} = \sum_{i \in [m]} \textbf{Y}_i$ are symmetric matrices with the same null space such that $\mathbf{X} \preceq \mathbf{Y}$, $b = \mathbf{Y} \bar{x}$, and $x$ is an arbitrary vector. Let the $\mathbf{Y}_i$ matrices be rank-one, and let $\tau \in \R^E_{\geq 0}$ be leverage score overestimates in that they satisfy $\tau_i \geq \Tr[\mathbf{Y}_i \mathbf{X}^\dagger]$. Let $\mathbf{Z} = \Sample(\{\mathbf{Y}_1, ... \mathbf{Y}_m \}, \mathbf{X}, \tau, \frac{1}{10})$, and define $x'$ as
\[
x' = x - \frac{1}{10} \mathbf{Z}^\dagger (\mathbf{Y} x - b).
\]
Then
\begin{align}
\label{eqn:exdec}
\mathbb{E}_{r, i_1, i_2, \ldots i_r} \left[ \norm{x' - \bar{x} }^2_\mathbf{Y} \right] \leq \left(1 - \frac{1}{40} \right) \norm{x- \bar{x} }^2_\mathbf{Y}.
\end{align}
Further, $\mathbf{Z}$ can be computed in $O(m + \norm{\mathbf{\tau}}_1)$ time, and each matrix $\mathbf{Y}_i$ is added at least once to $\mathbf{Z}$ with probability at most $\min\{1, 20 \tau_i\}$. Finally, for some fixed constant $c_s$ with high probability in $n$ we have 
\[
\frac{1}{c_s \log n} \mathbf{Y} \preceq \mathbf{Z} \preceq c_s \log n \mathbf{Y}.
\]
\end{lemma}
\begin{proof}
Equation~\ref{eqn:exdec} and the bound on the algorithm's runtime are directly copied from Lemma 2.3 from \cite{precon-exp}. For the bound on the probability $\mathbf{Y}_i$ is added to $\mathbf{Z}$, we look at each execution of line $5$ of the algorithm. For iteration $j$, we pick $\mathbf{Y}_i$ with probability $\frac{\tau_i}{s}$. Thus, we pick $\mathbf{Y}_i$ at most $r \frac{\tau_i}{s} \leq \frac{2s}{\delta} \frac{\tau_i}{s} =  20 \tau_i$ times in expectation. The conclusion follows by Markov's inequality. The final claim follows immediately from Lemma C.2 from \cite{precon-exp}.
\end{proof}

We additionally employ a partial Cholesky factorization lemma from \cite{CohenKMPPRX14} which enables us to reduce solving ultrasparse graph Laplacians to solving Laplacians with a much smaller number of edges. 
\begin{lemma}
Let $G$ be a weighted graph on $n$ vertices and $n + m'$ edges. There is a routine $\ElimSolve(G, \solve,b)$ which computes a vector $x$ satisfying 
\[
\norm{x - \mlap_G^\dagger b}^2_{\mlap_G} \leq \epsilon \norm{\mlap_G^\dagger b}^2_{\mlap_G}
\]
using $O(n + m')$ time plus one call to $\solve$, which is an $\epsilon$-approximate solve for graphs with at most $O(m')$ nodes and edges.
\end{lemma}

Variants of this result are used in many prior Laplacian system solvers, e.g. \cite{SpielmanT04,KMP11,KMP14}. For a detailed proof of this lemma with floating-point error analysis see Appendix C of \cite{peng13}. 

We now assemble these pieces to give an algorithm for solving Laplacians in graphs which contain ultrasparse low-stretch subgraphs:

\begin{algorithm2e}[t]
	\LinesNumbered
	\KwIn{$G$ graph, $G'$ subgraph, $\tau_i$ are upper bounds of leverage scores through $G'$, $b$ vector, $\epsilon$ error tolerance.}
	\KwOut{$x$ approximately satisfying $\mlap_G x = b$.}
	$x_1 \gets 0$ \\ 
	\For{$i \gets 1$ \KwTo $200 \log 1/\epsilon$}{
		$H_i \gets \Sample(G,G',\mathbf{\tau},\frac{1}{10})$ \\
		\If{$|E(H_i)| \leq 1600 \norm{\mathbf{\tau}}_1 + |E(G')|$ \label{line:toobig}}{
			$r_i \gets \mlap_G x_i -b$ \\
			$y_i \gets \ElimSolve(H_i, \solve, r_i)$ \\
			$x_{i+1} \gets x_i - \frac{1}{10} y_i$ \\
		}
		\Else{ 
			$x_{i+1} \gets x_i$ \\
		}
	}
	\Return{$x_i$}
	\caption{$\mathbf{Z} = \PreconRichardson(G, G', \mathbf{\tau}, b,\epsilon,\solve)$} \label{alg:richardson}
\end{algorithm2e}

\begin{lemma}
Let $G$ be a weighted $n$-node $m$-edge graph and let $G'$ be a subgraph of $G$ with $n + m'$ edges. Let $\tau_1, \ldots \tau_m\R^E$ be values satisfying $\tau_e \geq w_e \Reff_{G'}(u,v)$ for any $e = (u,v)$. Let $b$ be a vector, and let $\bar{x} = \mlap_G^\dagger b$. Then if $\solve$ is a $(1600 c^2_s \log^2 n)^{-1}$-Laplacian solver (Definition~\ref{def:solver})', Algorithm $\PreconRichardson(G, G',\mathbf{\tau}, b, \epsilon)$ computes a vector $x$ satisfying
\[
\mathbb{E} \left[ \norm{x - \bar{x}}^2_{\mlap_G} \right] \leq \epsilon \norm{\bar{x}}^2_{\mlap_G}
\]
using $O(\log 1/\epsilon)$ iterations. Each iteration consists of $O(m + \norm{\mathbf{\tau}}_1)$ work plus one call to $\solve$ on a graph with $O(\norm{\mathbf{\tau}}_1 + m')$ edges. 
\label{lem:preconrichardson}
\end{lemma}

\begin{proof}
We bound the expected decrease of $\norm{x_i - \bar{x}}_{\mlap_G}^2$ in an iteration. We first show the condition on \Cref{line:toobig} holds with large probability. By Lemma~\ref{lemma:sample} the expected number of edges added to $G'$ when forming $H$ is at most $20 \norm{\mathbf{\tau}}_1$. Thus by Markov's inequality $H$ contains fewer than $|E(G')| + 1600 \norm{\mathbf{\tau}}_1$ edges with probability at least $1 - \frac{1}{80}$. If we call this event $\rho$,  by Markov's inequality we have
\[
\mathbb{E} \left[ \norm{\bar{x} - \left( x_i - \frac{1}{10} \mlap_{H_i}^\dagger (\mlap_G x_i-b) \right)}^2_{\mlap_G} ~ \Big| ~ \rho \right] \leq \left(1 - \frac{1}{80} \right) \norm{\bar{x} - x_i}^2_{\mlap_G}
\]
where the expectation is over the randomness within a single iteration of the while loop. Now by the guarantees of $\solve$ and Lemma~\ref{lemma:sample} we know that in each iteration with high probability
\begin{align*}
     \mathbb{E} \left[\norm{y_i - \mlap_{H_i}^\dagger r_i}^2_{\mlap_G}\right] &\leq  (c_s \log n) \mathbb{E} \left[\norm{y_i - \mlap_{H_i}^\dagger r_i}^2_{\mlap_{H_i}}\right] 
     \leq \frac{c_s \log n}{1600 c_s^2 \log^2 n} \norm{\mlap_{H_i}^\dagger r_i }^2_{\mlap_{H_i}} \\
     &\leq \frac{1}{1600 c_s \log n} \norm{\mlap_G(\bar{x} - x_i)}^2_{\mlap_{H_i}^\dagger} 
     \leq \frac{c_s \log n}{1600 c_s \log n}  \norm{\mlap_G (\bar{x} - x_i)}^2_{\mlap_G^\dagger} \\
     &= \frac{1}{1600} \norm{\bar{x} - x_i}^2_{\mlap_G}.
\end{align*}
 
Now, note that for any two vectors $u,v$ and any Euclidean norm we have
\[
\norm{u + v}^2 = \norm{u}^2 + \norm{v}^2 + 2 u^\top v \leq (1 +\alpha) \norm{u}^2 + (1 + \alpha^{-1}) \norm{v}^2
\]
for any $\alpha > 0$ by the Cauchy-Schwarz inequality and the AM-GM inequality. With this, we obtain for any $\alpha > 0$
\begin{align*}
\norm{\bar{x} - x_{i+1}}^2_{\mlap_G} &= \norm{\bar{x} - \left( x_i - \frac{1}{10} y_i  + \frac{1}{10} \mlap_{H_i}^\dagger r_i - \frac{1}{10} \mlap_{H_i}^\dagger r_i  \right)}^2_{\mlap_G} \\
&\leq \left(1+ \alpha\right) \norm{\bar{x} - \left(x_i -  \frac{1}{10} \mlap_{H_i}^\dagger r_i \right)}^2_{\mlap_G} + \frac{1}{100} \left(1 + \alpha^{-1} \right) \norm{y_i - \mlap_{H_i}^\dagger r_i}^2_{\mlap_G}.
\end{align*}
Choosing $\alpha = \frac{1}{400}$, we thus obtain 
\begin{align*}
\mathbb{E}\left[ \norm{\bar{x} - x_{i+1}}^2_{\mlap_G} | \rho \right] &\leq \left(1 + \frac{1}{400} \right) \left(1 - \frac{1}{80} \right) \norm{\bar{x} - x_i}^2_{\mlap_G} + \frac{1}{100} \left( \frac{401}{1600} \right) \norm{\bar{x} - x_i}^2_{\mlap_G} \\
&\leq  \left( 1 - \frac{1}{160} \right) \norm{\bar{x} - x_i}^2_{\mlap_G}.
\end{align*}
Therefore we obtain
\begin{align*}
    \mathbb{E}\left[ \norm{\bar{x} - x_{i+1}}^2_{\mlap_G}  \right] &\leq  \Pr(\rho) \mathbb{E}\left[ \norm{\bar{x} - x_{i+1}}^2_{\mlap_G}  | \rho \right] + Pr(\neg \rho) \mathbb{E}\left[ \norm{\bar{x} - x_{i+1}}^2_{\mlap_G}  | \neg \rho \right] \\
    &\leq \left(1 - \frac{1}{80}\right) \left(1 - \frac{1}{160} \right) \norm{\bar{x} - x_{i}}^2_{\mlap_G} + \frac{1}{80} \norm{\bar{x} - x_{i}}^2_{\mlap_G} \\
    &\leq \left(1 - \frac{1}{200} \right) \norm{\bar{x} - x_{i}}^2_{\mlap_G}.
\end{align*}
As we perform $200 \log 1/\epsilon$ iterations, this error decrease guarantee implies the output $x$ satisfies the desired bound of
\[
\mathbb{E}\left[ \norm{\bar{x} - x}^2_{\mlap_G}  \right] \leq \left(1 - \frac{1}{200} \right)^{200 \log 1/\epsilon} \norm{\bar{x} - x_1}_{\mlap_G}^2 \leq \epsilon \norm{\bar{x}}_{\mlap_{G}}^2 = \epsilon \norm{\mlap_G^\dagger b}_{\mlap_G}^2\,.
\]

We now bound the runtime per iteration. In each iteration we make $1$ call to $\Sample$, which by Lemma~\ref{lemma:sample} requires $O(m + \norm{\tau}_1)$ time. Now if the sampled subgraph $H_i$ does not satisfy the condition on line $4$ of the algorithm, we conclude the iteration. If we instead have  $|E(H_i)| \leq 3200 \norm{\tau}_1 + |E(G')|$, the call to $\ElimSolve$ on line $6$ requires $O(m)$ work plus a single call to $\solve$ on a graph with $O(\norm{\tau}_1 + m)$ edges. The claim follows.
\end{proof}
Finally, we apply this primitive recursively to precondition an accelerated gradient descent algorithm with guarantees given by the following theorem. 

\begin{restatable}[Randomized Preconditioned AGD]{theorem}{agdmaintheorem}
	\label{thm:agd}
	Let $\ma, \mb \in \R^{n \times n}$ be symmetric PSD matrices with $\ma \preceq \mb \preceq \kappa \ma$ for $\kappa \geq 1$, let $b \in \im(\ma)$, let $\epsilon \in (0, 1)$ and let $\solve_{\mb}$ be a  $\frac{1}{10\kappa}$-approximate solver for $\mb$. Then 
	$\ApproxAGD(\ma, b, \epsilon, \kappa, \solve_{\mb})$ (\Cref{alg:approx_agd}) is an
	$\epsilon$-approximate solver for $\ma$ and its runtime is the runtime of $O(\sqrt{\kappa} \log(1/\epsilon))$ iterations each of which consist of applying $\ma$ to a vector, invoking $\solve_{\mb}$, and additional $O(n)$ time operations.
\end{restatable}

While related theorems are standard to the literature and a deterministic variant analyzing Chebyshev iteration appears in \cite{CohenKMPPRX14}, we provide this theorem both for completeness and to simplify and improve our analysis. The theorem is discused in greater detail and proved in \Cref{sec:chebyshev}. With this, we have the pieces to give our final algorithm for solving Laplacian linear systems:

\begin{algorithm2e}[t]
\LinesNumbered
\KwIn{$G$ graph,$b$ vector, $\LowStretch$ oracle that returns low-stretch subgraphs, $\epsilon \in (0,1/2]$ error tolerance, $\delta \in (0,1)$}
\KwOut{$x$ approximately satisfying $\mlap_G x = b$.}
Let $z$ denote the solution of $f(z) +2 + \delta =z$, where $f$ is defined in \Cref{thm:preconcheby}\\
$\gamma = C (\log \log n)^{z}$, for sufficiently large constant $C$ \\
$\{H,\tau \} \gets \LowStretch(G, z)$ \\ 
$\kappa \gets$ an overestimate of $\norm{\tau}_1$ output by $ \LowStretch, \eta \gets \frac{\gamma \kappa}{m}$ \\
$G' = G + (\eta -1) H$, $\tau' \gets \frac{\tau}{\eta}$ \\
$\mathsf{RecSolver} \gets \Recursive(\cdot,\cdot, \frac{1}{1600 c_s^2 \log^2 n},\LowStretch) $ \\
$\mathsf{RichardsonSolver}_{G'} = \PreconRichardson(G',\eta H,\tau',\cdot, \frac{1}{10 \eta}, \mathsf{RecSolver})$ \label{line:richardson} \\
$x \gets \AGD(G,b,\epsilon, \eta, \mathsf{RichardsonSolver}_{G'})$ \label{line:AGD}\\
\Return{$x$}
\caption{$\mathbf{Z} = \Recursive(G,b,\epsilon,\delta,\LowStretch)$} \label{alg:cheby}
\end{algorithm2e}
\begin{theorem}
\label{thm:preconcheby}
Let $G=(V,E,w)$ be a $n$-node $m$-edge graph and let $\LowStretch$ by an algorithm which takes as input an $n'$-node $m'$-edge graph $G'$ and parameters $C,c>0$ and returns a $\kappa$-distortion subgraph of $G$ with at most $n'+ \frac{m'}{(C \log \log n')^c}$ edges and a corresponding vector of leverage score overestimates in $O(m \log \log n)$ time, where 
\[
\kappa = O\left(m \left( \log \log n'\right)^{f(c)} \right)
\]
for some concave, monotone increasing function $f$ with $f(0) > 0$. Let $\delta > 0$ be a parameter, and let $z$ denote the (unique) solution to the equation $f(z) + 2 +\delta = z$. Then for all sufficiently large $n$ $\Recursive$ is an $\epsilon$-approximate Laplacian solver for $G$ with running time
\[
O\left( m (\log \log n)^{z-1} \log(1/\epsilon) \right).
\]
\end{theorem}
\begin{proof}

We first prove the algorithm is an $\epsilon$-approximate Laplacian solver. We proceed by strong induction on $m$, the number of edges in the input graph $G$. Assume that $\Recursive$'s output is correct for all graphs with fewer than $m$ edges. In one level of recursion, we construct a subgraph $H$ with associated stretch overestimates $\tau$ with $n + \gamma^{-1} m$ edges that achieves $\kappa$-spectral distortion for
\[
\kappa \leq \zeta m (\log \log n)^{f(z)},
\]
where $\zeta$ is an absolute constant. We use this to form a graph $G'$ with stretch overestimates $\tau' = \frac{\tau}{\eta}$. We observe that 
\[
\eta H \preceq G' \quad \text{and} \quad G \preceq G' = (\eta-1) H + G \preceq \eta G
\]
since $H$ is a subgraph of $G$: thus we conclude that $\tau'$ are valid stretch overestimates of the edges in $G'$. Further, the choice of parameters $\eta, \kappa, \gamma$ implies that $\PreconRichardson_{G'}$ on line~\ref{line:richardson} applies $\Recursive$ to graphs with $O(\gamma^{-1} m + \eta^{-1} \kappa) = O(\gamma^{-1} m)$ edges: for sufficiently large constant $C$ we see that this is less than $m$. Thus the calls to $\Recursive$ are correct by induction, and by Lemma~\ref{lem:preconrichardson} $\mathsf{RichardsonSolver}_{G'}$ is a $\frac{1}{10 \eta}$-solver for $G'$. Finally since $G' \approx_{\eta} G$ we conclude by Theorem~\ref{thm:agd} that $\AGD$ is an $\epsilon$-appoximate Laplacian solver: this completes the induction.

We now bound the running time of the algorithm. Fix a constant value $m_0$, and note that $\Recursive$ runs in $O(1)$ time for all graphs with fewer than $m_0$ edges since the proof of correctness implies that the algorithm runs in finite time. Let $\runtime(m,\epsilon)$ denote the running time of our algorithm on a graph with $m \geq m_0$ edges with error parameter $\epsilon$. In one level of recursion, we first perform one call to $\LowStretch$-- this requires $O\left(m \log \log n \right)$ time. All remaining steps in our recursive algorithm are trivially $O(m)$ time except for the call to $\AGD$ on line \ref{line:AGD}. We recall $\mlap_{G} \preceq \mlap_{G'} \preceq \eta \mlap_G$: by Theorem~\ref{thm:agd} the call to $\AGD$ performs $O(\sqrt{\eta} \log(1/\epsilon) )$ iterations, each of which performs $O(m)$ work plus one call to $\mathsf{RichardsonSolver}_{G'}$ with error parameter $\frac{1}{10 \eta}$. Observing that $\norm{\tau'}_1 \leq \eta^{-1} \kappa = \gamma^{-1} m  \leq m$, by Lemma~\ref{lem:preconrichardson} each of these calls to $\PreconRichardson$ runs in time $O(m \log (\eta))$ plus the time needed for $O(\log (\eta))$ calls to $\Recursive$ on graphs with $O(\gamma^{-1} m)$ edges and error parameter $\solveeps = \frac{1}{1600 c_s^2 \log^2 n}$. Thus, one recursive loop of the algorithm performs
\[
O(m \log \log n) + O\left(m \sqrt{\eta} \log(1/\epsilon) \log \eta \right)
\]
work plus the cost of $O\left(\sqrt{\eta} \log(1/\epsilon) \log \eta \right)$ calls to $\Recursive$ on graphs with at most $\beta \gamma^{-1} m $ edges (where $\beta$ is the constant hidden in the big-$O$ notation) and error parameter $\solveeps$. This implies the recurrence
\[
\runtime(m,\epsilon) \leq \psi \left(\log \log n+ \sqrt{\eta} \log(1/\epsilon) \log \eta \right)  \left(m  +  \runtime\left(\beta \gamma^{-1} m, \solveeps \right) \right) 
\]
for some explicit constant $\psi$. Define $\solvertime(m)= \runtime(m, \solveeps)$. We first establish and solve a recurrence for  $\solvertime(m)$: we will use this to prove the full runtime claim. Since $\beta  \gamma^{-1} m  < m$ we may apply our induction hypothesis: further as $\eta \geq 1$ and $\log(1/\solveeps) \leq 2 \log \log n + O(1) \leq 3 \log \log n$ for $n$ sufficiently large we have
\[
\solvertime(m) \leq 4 \psi \left( \sqrt{\eta} \log \eta \log \log n \right) \left( m + \solvertime(\beta \gamma^{-1} m) \right).
\]
We will show that 
\[
4 \psi \left( \sqrt{\eta} \log \eta \log \log n \right) \beta \gamma^{-1} \leq \frac{3}{4}
\]
for the appropriate choice of constant $C$: this will allow us to apply the master theorem for this recurrence. Observe
\begin{equation}
\label{eqn:eta}
\eta = \frac{\gamma \kappa}{m} \leq \zeta \left( \log \log n\right)^{f(z)} \gamma = \zeta \left( \log \log n\right)^{z-2-\delta} \gamma =  \zeta C^{-1} \gamma^2 \left( \log \log n\right)^{-2-\delta}.
\end{equation}
Thus
\[
4 \psi \left( \sqrt{\eta} \log \eta \log \log n \right) \beta \gamma^{-1} \leq 4 \psi \zeta^{1/2} C^{-1/2} \beta \log \eta \left( \log \log n\right)^{-\delta/2}.
\]
As $\log \eta \leq 2 \log \gamma + O(1) = O(\log \log \log n)$ and $m \geq m_0$, we choose $m_0$ sufficiently large and conclude \[
4 \psi \left( \sqrt{\eta} \log \eta \log \log n \right) \beta \gamma^{-1} \leq 4 \psi \zeta^{1/2} C^{-1/2} \leq \frac{3}{4}
\] 
for a sufficiently large choice of $C$. Thus for any $m \geq m_0$ we have
\[
\solvertime(m) \leq \frac{3}{4} \alpha^{-1} \left( m + \solvertime(\alpha m) \right)
\]
for $\alpha = \beta \gamma^{-1}$. The master theorem thus implies
\[
\solvertime(m) \leq O(\alpha m) = O\left(m \left(\log \log n\right)^z \right)
\]
for all $m$. Applying this to the recurrence for $\runtime(m,\epsilon)$, we obtain
\begin{align*}
\runtime(m,\epsilon) &\leq \psi \left(\log \log n+ \sqrt{\eta} \log(1/\epsilon) \log \eta \right)  \left(m  +  \solvertime\left(\beta \gamma^{-1} m \right) \right) \\
&\leq O \left(m \log \log n+ m \sqrt{\eta} \log(1/\epsilon) \log \eta \right)
\end{align*}
since $\solvertime(\beta \gamma^{-1} m) = O(m)$. We note by \eqref{eqn:eta} that
\[
\sqrt{\eta}  \leq O\left( \gamma \left( \log \log n\right)^{-1-\delta/2} \right) = O\left( \left( \log \log n \right)^{z-1-\delta/2} \right).
\]
We now observe $\log \eta \leq \left(\log \log n\right)^{\delta/2}$ for large enough $n$ and $z = 2 + \delta + f(z) > 2$: these together imply
\[
\runtime(m,\epsilon) \leq O\left(m \log \log n + m \left(\log \log n\right)^{z-1} \log(1/\epsilon) \right) = O\left( m \left(\log \log n\right)^{z-1} \log(1/\epsilon) \right)
\]
as desired.

\end{proof}
Finally, we apply the $\kappa$-distortion subgraphs computed in Theorem~\ref{thm:fast_lp_subgraphs} to obtain our main result: 
\solver*
\begin{proof}
We observe that the procedure given in Theorem~\ref{thm:fast_lp_subgraphs} yields $\kappa$-distortion subgraphs with $n + \frac{m}{C (\log \log n)^c}$ edges, for 
\[
\kappa = O\left(m \left(\log \log n\right)^{1+\sqrt{8c} +o(1)} \right).
\]
For sufficiently large $n$ and constant in the big-$O$ notation, this satisfies the conditions of Theorem~\ref{thm:preconcheby} for function $f(x) = 1 + \sqrt{8x} + \delta$ for any $\delta > 0$. Theorem~\ref{thm:fast_lp_subgraphs}  constructs such subgraphs in $O(m)$ time which is sufficient for our guarantee. Substituting these into the guarantee of Theorem~\ref{alg:cheby} gives an $\epsilon$-approximate solver running in time
\[
O\left(m (\log \log n)^{z-1} \log(1/\epsilon)\right),
\]
where $z$ is the solution to $\sqrt{8z} + 3 + 2 \delta = z$ for any $\delta \geq 0$. Solving this equation reveals $z = 7 + \sqrt{40 + 8 \delta} + 2 \delta$: by choosing $\delta$ sufficiently small this gives an exponent of $6 + \sqrt{40} + \chi$ for any $\chi \geq 0$. We remark that $6 + \sqrt{40} \approx 12.324$. 
\end{proof}

We made only limited attempts to optimize the $\loglog$ dependence of this algorithm. We believe this dependence may be improved and here describe possible improvements. First, our method of analyzing the recursion is in some sense, weaker than that of \cite{CohenKMPPRX14}. In their setting the low stretch subgraph is a tree: in that case the recursively generated subproblems natively contain a low-distortion subgraph with $\kappa = O(m)$. This observation enables them to make different choices for the recursion parameter $\eta$. Although this observation is key to ensure their algorithm runs in $\otilde(m \sqrt{\log n})$ time and not $\otilde(m \log n)$, in our case the only effect is to reduce the polynomial dependence on $\log \log n$. By carefully applying this technique to our setting we believe our runtime can be improved.

A larger obstruction in obtaining a better runtime is the use of a ``bottom-up" recursion based on \cite{AKPW} in our construction of $\kappa$-distortion subgraphs. While this suffices to obtain our claim, our running time would be much improved by using a ``top-down" graph decomposition more closely resembing \cite{CohenKMPPRX14}. We leave it as an interesting problem for future work. 

As an additional remark we observe that we can convert the expected-decrease guarantee of Theorem~\ref{thm:lap} to a high-probability bound, provided that we allow our runtime bound to hold in expectation and with an extra $O(\log \log n)$ factor. To obtain this, we use Lemmas~4.5 and 4.9 from \cite{CohenKPPR14preconditionArxiv}, which allow us to construct a linear operator $\mz$ satisfying $\mlap^\dagger \preceq \mz \preceq \log^{4} n \mlap^\dagger$ which can be computed and applied in $O(m \log \log n)$ time.\footnote{Lemma~4.5 from \cite{CohenKPPR14preconditionArxiv} implies an $O(m \log \log n)$-time algorithm to construct a `graph-tree tuple’ with $m+n$ edges and $\norm{\tau}_p^p \leq O(m \log^p n)$ in $O(m)$ time for any $p \in (1/2 , 1)$.  Lemma~4.9 can be modified to use the solver of \cite{CohenKMPPRX14} instead of \cite{KMP11} to apply $\mz$: doing so enables us to apply it in $O(m + \log^{-2p} \norm{\tau}_p^p \sqrt{\log n} (\log \log n)^4 )$ expected time. Plugging in Lemma~4.5 and choosing $p = 2/3$ yields an expected running time bound of $O(m)$.} With this, we simply apply Theorem~\ref{thm:lap} with $\epsilon \gets \frac{\epsilon}{2 \log^4 n}$: Markov's inequality implies the output $x$ has $\norm{\mlap x - b }_{\mlap^\dagger}^2 = \norm{x - \mlap^\dagger b}_\mlap^2 \leq \frac{\epsilon}{\log^4 n} \norm{b}_\mlap^2$ with probability $1/2$. If this holds, we may use $\mz$ to verify in $O(m \log \log m)$ time whether $\norm{x - \mlap^\dagger b}_\mlap^2 \leq \epsilon \norm{b}_{\mlap}^2$. The claim follows by repeating this procedure until a desired solution is found.

\section*{Acknowledgments}

We thank Yair Carmon, Yang P. Liu, Michael Kapralov, Jonathan Kelner, Navid Nouri, Richard Peng, and Jakab Tardos for helpful discussions. We thank anonymous reviewers for feedback on earlier versions of this paper. Aaron Sidford was support in part by a Microsoft Research Faculty Fellowship, NSF CAREER Award CCF-1844855, NSF Grant CCF-1955039, a PayPal research gift award, and a Sloan Research Fellowship. 

\bibliographystyle{alpha}
\bibliography{bib}

\appendix

\section{Ultrasparsifiers by Spectral Graph Theory}
\label{sec:ultrasparsifiers}

In this section, we prove our two main results regarding ultrasparsifiers. We begin by proving the existence of ultrasparsifiers for sums of arbitrary rank-$1$ matrices.

\bss*

To obtain this result we first give the following \Cref{thm:badtree} regarding spectral properties of 
subsets of sums of rank one matrices and then we use it to prove \Cref{thm:ultrasparsifer_existence}.

\begin{theorem}
\label{thm:badtree}
Let $v_1 ... v_m \in \R^{n}$ and let $\ma = \sum_{i \in [m]} v_i v_i^\top$ be full rank. For any $k \geq 1$ there exists $S \subset [m]$ with $|S| \leq n + \frac{n}{k}$ such that $\mb \defeq \sum_{i \in S} v_i v_i^\top$ has $\Tr[\mb^{-1} \ma] \leq mk$.
\end{theorem}

\begin{proof}
We start with $S = [m]$ and greedily remove elements from $S$ to minimize the increase in $\Tr[\mb^{-1} \ma]$. Observe that the initial value of this trace is $n$ since $\mb = \ma$. By the Sherman-Morrison formula for rank 1 matrix updates, for any invertible matrix $\mm$ and vector $v$ with $v^\top \mm^{-1} v \neq 1$
\[
\left(\mm - v v^\top \right)^{-1} = \mm^{-1} + \frac{\mm^{-1} v v^\top \mm^{-1}}{1 - v^\top \mm^{-1} v}
\text{ and is invertible} ~.
\]
With this, we analyze the change to $\Tr[\mb^{-1} \ma]$ after one element is removed from $S$. For any $i \in S$,
\[
\Tr \left[\left( \mb - v_i v_i^\top \right)^{-1} \ma \right] = \Tr[ \mb^{-1} \ma] + \frac{v_i^\top \mb^{-1} \ma \mb^{-1} v_i}{1 - v_i^\top \mb^{-1} v_i}
\]

We consider randomly sampling $i$ to remove from $S$ with probability $p_i \propto 1 - v_i^\top \mb^{-1} v_i$. We will show the increase to the trace is bounded in expectation. Since for $i \in S$, $v_i v_i^\top \preceq \mb$, all the $p_i$ are all nonnegative and non-zero only when $1 - v_i^\top \mb^{-1} v_i > 0 $. 
Selecting $i$ in this way yields
\begin{align*}
\mathbb{E} \left[ \Tr \left[\left( \mb - v_i v_i^\top \right)^{-1} \ma \right] \right] 
&= \Tr[ \mb^{-1} \ma] + \sum_{i \in S} p_i \frac{v_i^\top \mb^{-1} \ma \mb^{-1} v_i}{1 - v_i^\top \mb^{-1} v_i} \\
&= \Tr[ \mb^{-1} \ma] +  \frac{\sum_{i \in S} v_i^\top \mb^{-1} \ma \mb^{-1} v_i}{\sum_{i \in S} 1 - v_i^\top \mb^{-1} v_i}.
\end{align*}
By the cyclic property of trace, we observe $\sum_{i \in S} v_i^\top \mb^{-1} \ma \mb^{-1} v_i = \Tr[\mb \mb^{-1} \ma \mb^{-1}] = \Tr[\mb^{-1} \ma]$ and $\sum_{i \in S} 1 - v_i^\top \mb^{-1} v_i = |S| - n$. Applying these equations yields
\[
\mathbb{E} \left[ \Tr \left[\left( \mb - v_i v_i^\top \right)^{-1} \ma \right] \right] =  \Tr[ \mb^{-1} \ma] \left(1 + \frac{1}{|S| - n} \right)
\]
and therefore whenever $|S| > n$, there exists some $i \in S$ satisfying
\[
\Tr \left[\left( \mb - v_i v_i^\top \right)^{-1} \ma \right] \leq  \Tr[ \mb^{-1} \ma] \left(1 + \frac{1}{|S| - n} \right) = \Tr[ \mb^{-1} \ma] \left(\frac{|S| - n + 1}{|S| - n} \right).
\]
By repeatedly applying this bound from $|S| = m$ to $|S| = n + \lceil\frac{n}{k}\rceil$, we remove all but $n + \frac{n}{k}$ elements from $S$ and end up with $\mb$ satisfying the desired bound of
\[
\Tr \left[ \mb^{-1} \ma \right] \leq n \prod_{s = n + \lceil\frac{n}{k}\rceil}^m \left(\frac{s - n +1}{s - n} \right) = n \left( \frac{m - n + 1}{n/k} \right) \leq mk ~.
\]
\end{proof}

We now show a modification of \cite{BSS} which allows us to convert the output of \Cref{thm:badtree} into an ultrasparsifier. 

\begin{theorem} Let $v_1, v_2, ..., v_m \in \R^n$ be vectors such that $\sum_i v_i v_i^\top = \eye$. Let $\ma \in \R^{n \times n}$ be a matrix satisfying $\ma \preceq \eye$ and $\Tr [ \ma^{-1} ] = \kappa$. Then for any $q \geq 0$ with $\lceil (\kappa + 2n) q \rceil \leq n$,  there exists $S \subseteq [m]$ with $|S| = \lceil (\kappa + 2n) q  \rceil$  with corresponding weights $w_i > 0$ such that
\[
q \eye \preceq \ma + \sum_{i \in S} w_i v_i v_i^\top \preceq 3 \eye. 
\]
\label{thm:bss-ultra}
\end{theorem}

Our proof of this result is as a consequence of technical lemmas from \cite{BSS} restated below.

\begin{lemma}[Combination of Lemmas 3.3, 3.4, 3.5 from \cite{BSS}]
Let $v_1, v_2, ... v_m \in \R^n$ be vectors such that $\sum_i v_i v_i^\top = \eye$. Define the functions $\Phi^u(\mm) \defeq \Tr [(u \eye - \mm )^{-1}]$ and $\Phi_l(\mm) \defeq \Tr [(\mm - l \eye)^{-1}]$.  Let $\ma$ be a matrix satisfying $l \eye \preceq \ma \preceq u \eye$ as well as
\[
\Phi^u\left(\ma\right) \leq \gamma_U \quad \text{and} \quad \Phi_l\left(\ma\right) \leq \gamma_L.
\]
Then for $\delta_U, \delta_L$ satisfying $1/\delta_U + \gamma_U \leq 1/\delta_L - \gamma_L$ there exists $i \in [m]$ and $t>0$ such that $(l + \delta_L) \eye \preceq \ma + t v_i v_i^\top \preceq (u + \delta_U) \eye$ as well as
\[
\Phi^{u+\delta_U}\left(\ma + t v_i v_i^\top \right) \leq \gamma_U \quad \text{and} \quad \Phi_{l+\delta_L}\left(\ma + t v_i v_i^\top \right) \leq \gamma_L.
\]
\label{lemma:bss_technical}
\end{lemma}
Our use of \Cref{lemma:bss_technical} mirrors its use in \cite{KMST}: we iteratively add vectors to the output of \Cref{thm:badtree} to increase the spectral upper and lower bounds on $\ma$ appropriately. After a small number of iterations, we certify that $\ma$ is appropriately spectrally bounded and terminate.
\begin{proof}[Proof of \Cref{thm:bss-ultra}]
We let $\ma^{(0)} \defeq \ma$ and for $j \geq 0$ iteratively define $\ma^{(j+1)} \defeq \ma^{(j)} + t v_i v_i^\top$ for some $t \geq 0$ and $i  \in [m]$ (each depending on $j$). Further, we set $\gamma_U = n$, $\gamma_L = \kappa$, $u_0 = 2$, $l_0 = 0$, $\delta_U = \frac{1}{n}$, $\delta_L = \frac{1}{2n + k}$. Observe that
\[
\Phi^{u_0}\left(\ma^{(0)}\right) = \Tr \left( 2 \eye - \ma \right)^{-1} \leq \Tr \left( \eye \right) \leq n =  \gamma_U
\]
and 
\[
\Phi_{l_0}\left(\ma^{(0)}\right) = \Tr \left( \ma^{-1} \right) = \kappa = \gamma_L.
\]
Further, the choice of parameters ensures $1/\delta_U + \gamma_U \leq 1/\delta_L - \gamma_L$. Thus, inductively applying \Cref{lemma:bss_technical} yields that for each $j > 1$ there exists $t > 0$, and $i \in [S]$ where $\ma^{(j)} = \ma^{(j-1)} + t v_i v_i^\top$ satisfies
\[
\Phi^{u_0 + j\delta_U}\left(\ma^{(j)} \right) \leq \gamma_U \quad \text{and} \quad \Phi_{l_0 + j \delta_L} \left(\ma^{(j)} \right) \leq \gamma_L.
\]
Therefore $\ma^{(s)}$ for  $s = \lceil (\kappa + 2n)q \rceil$ satisfies the desired bound of
\[
q \eye  \preceq \left( l_0 + s \delta_L \right) \eye \preceq \ma^{(s)} \preceq \left( u_0 + s \delta_U \right) \eye \preceq 3 \eye  
\]
where in the last inequality we used the upper bound on $s$.
\end{proof}
Finally, we combine \Cref{thm:badtree} and \Cref{thm:bss-ultra} and give the proof of \Cref{thm:ultrasparsifer_existence}.
\begin{proof}[Proof of \Cref{thm:ultrasparsifer_existence}]
We note that we may assume $\ma = \sum_{i \in [m]} v_i v_i^\top$ is full-rank: otherwise we may add $u_j \in \ker(\ma)$ to make the result full rank, run the rest of the proof, and remove the added $u_j$ before returning the output. Since the $u_j$ are orthogonal to the $v_i$, removing them cannot affect the space spanned by $\ma$'s eigenvectors. 

By applying \cite{BSS}, we can find a collection of $34 n$ vectors $v'_1, v'_2, \dots, v'_{34 n}$ which are reweighted copies of the $v_i$ and satisfy $\ma \preceq \sum_{i \in [16 n]} v'_i \left(v'_i\right)^\top \preceq 2 \ma$. We assume $k>1000$ in the rest of the argument: the $v'_i$ immediately satisfy our requirements otherwise as $k$ was assumed to be at least $2$. 
\newcommand{\mc}{\mathbf{C}}

Given the vectors $v'_1, v'_2, \dots, v'_{34 n}$, define $\mc = \sum_{i \in [34 n]} v'_i \left(v'_i\right)^\top$: note that $\frac \ma \preceq \mc \preceq 2 \ma$. As $\mc$ is therefore full-rank, we apply \Cref{thm:badtree} and thus obtain a set $S$ of $n + \frac{n}{k}$ vectors such that 
\[
\Tr \left( \mc \left(\sum_{i \in S} v_i v_i^\top \right)^{-1} \right) \leq 34 n k.
\]
Thus, $\mb = \mc^{-1/2} \left(\sum_{i \in S} v_i v_i^\top \right) \mc^{-1/2}$ has $\mb \preceq \eye$ (as it is formed from an unweighted subset of the vectors that form $\mc$) and $\Tr (\mb^{-1}) \leq 34 nk$. Applying \Cref{thm:bss-ultra} with $\kappa = 34 nk$ and $q = \frac{1}{36 k^2}$, we observe that there exists a set $T$ of  $(\kappa + 2n) q \leq \frac{n}{k}$ vectors with corresponding weights $w_i \geq 0$ such that 
\[
\frac{1}{36 k^2} \mc \preceq \mc^{1/2} \mb \mc^{1/2} + \sum_{i \in T} w_i v_i v_i^\top \preceq 3 \mc. 
\]
Using the definition of $\mb$ and rearranging, we obtain
\[
\frac{1}{36k^2} \ma \preceq \frac{1}{36 k^2} \mc \preceq \sum_{i \in S \cup T} w_i v_i v_i^\top \preceq 3 \mc \preceq 6 \ma 
\]
Finally, $|S \cup T| \leq |S| + |T| \leq n + \frac{n}{k} + \frac{n}{k} = n + \frac{2n}{k}$: the output is a sum of outer products of at most $n + \frac{2n}{k}$ vectors. The claim follows by choosing $k \leftarrow 216 k^2$  and scaling the output. 
\end{proof}
We made no attempt to optimize the constants in the above proof. We additionally remark that Theorem~\ref{thm:ultrasparsifer_existence} immediately implies the existence of $k$-ultrasparsifiers with $n + O\left(\frac{n}{\sqrt{k-1}}\right)$ edges: we simply apply it to a graph Laplacian written in the form $\sum_{e \in E(G)} \left( \sqrt{w_e} b_e\right)\left( \sqrt{w_e} b_e\right)^\top$. 
We now construct ultrasparsifiers for graphs with improved guarantees by combining our $\kappa$-distortion subgraph construction with this BSS-derived framework. We begin with the general claim of our construction: we specialize it to several interesting parameter regimes as a corollary. 

\begin{theorem}[Polynomial-Time Ultrasparsifier Construction in Graphs]
\label{thm:ultra}
Let $G$ be any $n$-vertex graph with polynomially-bounded edge weights. There exists a polynomial time algorithm which constructs a reweighted subgraph $H$ with $n + O\left( \frac{n}{\gamma} \right)$ edges such that $\mlap_H \preceq \mlap_G \preceq \alpha \mlap_H$ for 
\begin{align*}
\alpha = O\left(\gamma \exp\left( \sqrt{8 \log \gamma \cdot \log \left(48 \log \left(\gamma  \log^{10} n\right) \sqrt{\log \gamma}  \right)} \right) \log  \left(\gamma  \log^{10} n\right) \sqrt{\log \gamma} \right) 
\end{align*}
\end{theorem}
\begin{proof}
We note that by preprocessing the graph with \cite{BSS}, we may assume that the graph has $O(n)$ edges with at most a constant factor loss in the final approximation error. We employ Theorem~\ref{thm:akpw} with $k = \gamma \log^{10} n$ and path sparsification algorithm given by Theorem~\ref{thm:path-sparse}. By the theorem's guarantee, this returns a subgraph $G'$ in polynomial time with 
\[
n + O\left( \frac{n \log^{10} n}{k} + \frac{n \log \gamma}{k^2} + \frac{n}{\gamma} \right) = n + O\left( \frac{n \log^{10} n}{\gamma \log^{10} n} + \frac{n}{\gamma} \right) = n + O\left(\frac{n}{\gamma}\right)
\]
edges. Further, the output subgraph is a $\kappa$-distortion subgraph with 
\[
\kappa =  O\left(n \gamma \exp\left( \sqrt{8 \log \gamma \cdot \log \left(48 \log \left(\gamma  \log^{10} n\right) \sqrt{\log \gamma}  \right)} \right) \log  \left(\gamma  \log^{10} n\right) \sqrt{\log \gamma} \right).
\]
We note that $\Tr\left[ \mlap_{G}^{1/2} \mlap_{G'}^{\dagger} \mlap_{G}^{1/2} \right] \leq \kappa$ by definition of $\kappa$-distortion, and that the collection of rank-1 matrices $w_e \mlap_{G}^{-1/2} b_e b_e^\top \mlap_{G}^{-1/2}, e \in G$ sums to $\eye$.
If $b_e b_e^\top$ are the edge Laplacian matrices that form $\mlap_G$, applying Theorem~\ref{thm:bss-ultra} with $\alpha \defeq 3 \frac{n}{\kappa \gamma}$ yields a set $S$ of $\left\lceil(\kappa + 2 n) \alpha \right\rceil = O\left(\frac{n}{\gamma}\right)$ edges with corresponding weights $w'_i > 0$ such that 
\[
3 \alpha \eye \preceq \mlap_G^{-1/2} \mlap_{G'} \mlap_{G}^{-1/2} + \sum_{e \in S} w'_e \mlap_{G}^{-1/2} b_e b_e^\top \mlap_{G}^{-1/2} \preceq 3 \eye.
\]
Scaling down the resulting matrix and multiplying both sides of the matrices by $\mlap_G^{1/2}$ gives a reweighted subgraph $H$ with $n + O\left(\frac{n}{\gamma}\right)$ such that $\alpha \mlap_G \preceq \mlap_H \preceq \mlap_G$ as desired. 
\end{proof}

As a corollary of this result, we prove \Cref{thm:ultra_export}, which we now recall.

\ultraexport*

\begin{proof}
We employ \Cref{thm:ultra} with different values of $\gamma$. For the first claim, we choose $\gamma = \left(\log \log n\right)^c$, and note that \Cref{thm:ultra} yields $\alpha$-ultrasparsifiers, with
\begin{align*} 
\alpha &= O\left( \gamma \exp\left( \sqrt{ 8 \log \gamma \cdot \log \left(48 \log \left(\gamma \log^{10} n\right) \sqrt{\log \gamma}\right)} \right) \log\left(\gamma \log^{10} n\right) \sqrt{\log \gamma} \right)\\
&= O\left( \left(\log \log n \right)^{c + \sqrt{8c} + 1 + o(1)} \right)
\end{align*}
by applying the definition of $\gamma$. For the second claim, we choose $\gamma = \alpha$ in \Cref{thm:ultra}: we obtain an ultrasparsifier of quality
\begin{align*} 
O\left( \alpha \exp\left( \sqrt{ 8 \log \alpha \cdot \log \left(48 \log \left(\alpha \log^{10} n\right) \sqrt{\log \alpha}\right)} \right) \log\left(\alpha \log^{10} n\right) \sqrt{\log \alpha} \right).
\end{align*}
Since $\alpha = \omega(\log^\delta n)$ for some fixed constant $\delta > 0$, we have $\log\left( \alpha \log^{10} n\right) \leq \log \left(\alpha^{1+ 10/\delta} \right) +O(1) = \left(1 + \frac{10}{\delta} \right) \log \alpha + O(1) \leq \left(2 + \frac{20}{\delta} \right) \log \alpha$ for sufficiently large $n$. Substituting this in, we obtain
\begin{align*} 
&O\left( \alpha \exp\left( \sqrt{ 8 \log \alpha \cdot \log \left(48 \log \left(\alpha \log^{10} n\right) \sqrt{\log \alpha}\right)} \right) \log\left(\alpha \log^{10} n\right) \sqrt{\log \alpha} \right) \\
&\leq O\left( \alpha \exp\left( \sqrt{ 8 \log \alpha \cdot \log \left(\left(96 + \frac{960}{\delta} \right) \log^{3/2} \alpha \right)} \right)  \log^{3/2} \alpha \right) \\
&= O\left( \alpha \exp\left( \sqrt{ 8 \log \alpha \cdot \left( \frac{3}{2} \log \log \alpha + O(1) \right)} \right)  \log^{3/2} \alpha \right) = \alpha^{1+o(1)}
\end{align*}
as claimed.
\end{proof}

\section{Primal Dual Characterization of Shortest Flow}
\label{sec:primal_dual_mincost}

Here we prove our primal-dual characterization of minimum cost flow that we use to reason about vertex disjoint paths in expanders. Though this is fairly standard and straightforward, 
we include a brief derivation here for completeness. 

\mincostflow*

\begin{proof}[Proof of \Cref{lem:minimum_cost_lemma}]
For all $a,b \in V$ let $\indicvec_{a,b} \defeq \indicvec_a - \indicvec_b$ where for all $c \in V$ we et $\indicvec_c \in \R^V$ denote the indicator vector for $c$, i.e. the vector that is a zero in all coordinates except for $c$ where it has value $1$. Further, let  $\mb \in \R^{E \times V}$ denote the edge-vertex incidence matrix of graph where for each edge $e = (a, b) \in E$ row $e$ of $\mb$ is $\indicdiff_{a,b}$. 

Leveraging this notation, we consider the following minimum cost flow problem of computing the flow of minimum total length that sends $F$ units of flow from $s$ to $t$ and puts at most one unit of non-negative flow is put on each edge:
\begin{equation}
\label{eq:mincost_problem}
\min_{f \in \R^E : f_e \in [0, 1] \text{ for all } e \in E \text{ and } \mvar{B}^\top f = F \cdot \delta_{s, t}} f^\top \vones
\end{equation}
To see that \eqref{eq:mincost_problem} corresponds to the desired flow problem, note that for all $a \in V$ and $f \in \R^E$, $[\mb^\top f]_a = \indicvec_a^\top \sum_{e = (a,b) \in E} \indicdiff_{a,b} f_e = \sum_{e = (a,b) \in E} f_e - \sum_{e = (b,a)} f_e$, i.e. the net flow leaving leaving vertex $a$ through $f$ in the graph.

There is always an integral minimizer for this problem and it corresponds to $F$ disjoint paths from $s$ to $t$ using a minimum number of edges.\footnote{This is a standard result regarding minimum cost flow. One way to see this is to note that given any solution $f$ to \eqref{eq:mincost_problem} if the edges $e$ with $f_e \notin \{0, 1\}$ form a cycle (viewing each directed edge $(a,b)$ as an undirected edges $\{a,b\}$) then flow can be sent in one direction of the cycle without increasing $f^\top \vones$ while preserving feasibility until at least one less edge has $f_e \notin \{0, 1\}$. Consequently, there is an optimal solution to \eqref{eq:mincost_problem} where the edges with $f_e \notin\{0, 1\}$ are acyclic (when viewed as undirected edges). However, by the constraint $\mb^\top f = F\cdot \indicdiff_{s,t}$ this implies that all edges have $f_e \in \{0, 1\}$ in this case. Consequently, there is an optimal integral flow $f$. (This holds more generally, see e.g. \cite{DaitchS08, schrijver2003combinatorial}.) Further, if there is a directed cycle in $G$ with a positive value of $f_e$ on each edge, a feasible $f$ with decreased $f^\top \vones$ can be found by sending flow in the reverse of each cycle. Consequently, there is an optimal integral acyclic flow and again by the the constraints this implies that $f$ corresponds to $F$ disjoint paths from $s$ to $t$. 
 }
Letting, $\mvar{0}_E, \mvar{I}_E \in \R^{E \times E}$ denote the all zero matrix and identity matrix respectively, letting $\vzero_E, \vones_E \in \R^E$ denote the all zero vector and all ones vector respectively, and letting
\[
\ma
 = 
\left(
\begin{matrix}
\mb  & \mvar{I}_E \\
\mvar{0}_E & \mvar{I}_E \\
\end{matrix}
\right)
\text{ , }
b 
=
\left(
\begin{matrix}
\vones_m \\ 
\vzero_m 
\end{matrix}
\right)
\text{ , and }
c 
=
\left(
\begin{matrix}
F \cdot \delta_{s,t} \\ 
\vones_m 
\end{matrix}
\right)
\]
we can write this problem equivalently as 
\begin{equation}
\label{eq:mincost_primal_dual}
(P) =  \min_{x \in \R^{E + E}_{\geq 0} ~ : ~  \mvar{A}^\top x = b} b^\top x
~ \text{ and } ~
(D) = \max_{y \in \R^{V + E}, s \in \R^{E + E}_{\geq 0} ~ : ~ \mvar{A} x - s = b} c^\top y 
\end{equation}
where we use $\R^{E + E}$ and $\R^{V + E}$ denote concatenations of two $\R^E$ vectors and concatenation of a $\R^V$ vector with a $\R^E$ vector, respectively. That $(P)$ is equivalent to the original minimum cost flow problem 
follows from the fact that $f \leq \vones$ entrywise if and only if $f + x = \vones$ for some $x \in \R^E_{\geq 0}$ and that $(D) = (P)$ follows from standard strong duality of linear programs. 

Now, let $x \in \R^{E + E}_{\geq 0}$ and $(y,s) \in \R^{V + E} \times \R_{\geq 0}^{E \times E}$ be optimal solutions ot $(P)$ and $(D)$ respectively in \eqref{eq:mincost_primal_dual}. Further, let without loss of generality $f$ be the concatenation of $f \in \R^E_{\geq 0}$ and $\vones - f \in \R^E_{\geq 0}$, let $y$ be the concatenation of $v \in \R^V$ and $z \in \R^E$, and let  $s$ be the concatenation of $s^1 \in \R^E_{\geq 0}$ and $s^2 \in \R^E_{\geq 0}$. Further, let $f$ be an integral minimizer and note that it corresponds to $F$ disjoint paths, and either $f_e = 0$ or $1-f_e = 0$ for all $e \in E$. 

Now, by optimality of $x$ and $s$ we know that $x^\top s = 0$ and therefore $f_e \cdot s_e^1 = 0$ for all $e \in E$ and $(1 - f_e) \cdot s_e^2 = 0$ for all $e \in E$. Further, for all edges $e \in (a,b)$ feasibility of $(y,s)$ for $(D)$ implies 
\[
v_a - v_b + z_e - s^1_e= 1 \text{ and } z_e - s^2_e = 0 ~.
\]
Consequently, if $f_e = 1$ then $s_e^1 = 0$ and  we have $s_e^1 = 0$, $z_e = s_e^2$ and $v_a - v_b = 1 - s_e^2$, i.e. $v_a - v_b \leq 1$. Further, if $f_e = 0$ then $s_e^2 = 0$, $z_e = 0$, and  $v_a - v_b = 1  +s_e^1$, i.e. $v_a - v_b \geq 1$. 
\end{proof}

\section{Randomized Preconditioned  Accelerated Gradient Descent}
\label{sec:chebyshev}

In this section we prove the following \Cref{thm:agd} regarding preconditioned accelerated gradient descent (AGD) for solving linear systems with random error in the preconditioner. That this accelerated preconditioned linear system solver handles randomized error aids our analysis in \Cref{sec:solver}.

\agdmaintheorem*

The robustness of accelerated methods to error has been studied in a variety of contexts (see e.g. \cite{LS13,MonteiroS13a,DevolderGN14,BubeckJLLS19}).We provide the proof in this section for completeness and to obtain a statement tailored to our particular setting. Limited attempts were made to optimize for the parameters and error tolerance for the preconditioner in the method.

We remark that this theorem is similar to one in \cite{CohenKMPPRX14} which analyzed preconditioned Chebyshev iteration with bounded solving error. Interestingly, a similar result as to \Cref{thm:agd} can be achieved by analyzing the method of that paper with randomized error in the solver. Straightforward modification of their analysis yields a variant of \Cref{thm:agd} where \Cref{eq:solve_requirement}  in the definition of a solver is replaced with $\E \norm{x - \ma^\pseudo b}_{\ma} \leq \sqrt{\epsilon} \norm{b}$ and the accuracy required for the solver for $\mb$ scales with $\epsilon$. We chose to provide the analysis of AGD as it provides an interesting alternative to preconditioned Chebyshev and naturally supported analysis of expected squared errors and preconditoners with accuracy that does not scale with $\epsilon$.

In the remainder of this section we provide $\ApproxAGD$ (\Cref{alg:approx_agd}) and prove
\Cref{thm:agd}. Our notation and analysis are similar to \cite{HinderSS20} and \cite{CohenKMPPRX14} in places and specialized to our setting in others.

\begin{algorithm2e}[h!]
	\LinesNumbered
	\caption{$F = \ApproxAGD(\ma, b, \epsilon,\kappa, \solve_{\mb})$
		\label{alg:approx_agd}
	} 
	\KwIn{Symmetric PSD $\ma \in \R^{n \times n}$, vector $b \in \R^n$, and accuracy $\epsilon \in (0, 1)$ }
	\KwIn{Condition number bound $\kappa$, $\frac{1}{10 \kappa}$-solver, $\solve_{\mb}$, for symmetric PSD $\mb \in \R^{n \times n}$ with $\ma \preceq \mb \preceq \kappa \ma$}
	\KwOut{A vector such that algorithm is an $\epsilon$-approximate solver for $\ma$}
	$x_0 := 0 \in \R^n$, $v_0 := 0 \in \R^n$, and $T := \lceil 4 \sqrt{\kappa}\log(2/\epsilon) \rceil$ \;
	\For{$t = 0$ to $T - 1$}{
		$y_t := \alpha x_t + (1 - \alpha) v_t$ where  $\alpha \defeq \frac{2\sqrt{\kappa}}{1 + 2\sqrt{\kappa}}$ \;
		$ x_{t+1} := y_t - g_t$ where $g_t :=  \solve_{\mb}(\ma y_t - b)$ \;
		$v_{t + 1} := \beta v_t + (1 - \beta)\left[ y_t - \eta g_t \right]$ where $\eta \defeq 2 \kappa$ and $\beta \defeq 1 - \frac{1}{2 \sqrt{\kappa}}$ \;
	}
	\Return $x_t$\;
\end{algorithm2e}

To analyze $\ApproxAGD$ (\Cref{alg:approx_agd}) we first provide the following lemma for bounding the error from an approximate solve.

\begin{lemma}
	\label{lem:precondition_helper}
	Let $\ma, \mb \in \R^{n \times n}$ be symmetric PSD matrices with $\ma \preceq \mb \preceq \kappa \ma$ and let $b \in \mathrm{im}(\ma)$. If $g = \solve_{\mb}(\ma x - b)$ where $\solve_{\mb}$ is an  $\epsilon$-approximate solver for $\mb$ then, $x_* \defeq \ma^\pseudo b$ and
	\begin{equation}
	\label{eq:g_delta_formula}
	\Delta \defeq g - \mb^\pseudo (\ma x - b) = g - \mb^\pseudo \ma(x - x_*)
	\end{equation}
	satisfies
	\begin{equation}
	\label{eq:delta_upper_bounds}
	\E \norm{\Delta}_\mb^2  \leq \epsilon \norm{x - x_*}_{\ma \mb^\pseudo \ma}^2
	\leq \epsilon \norm{x - x_*}_{\ma}^2
	\text{ and }
	\E \norm{g}_{\mb}^2 \leq (1 + \sqrt{\epsilon})^2 \norm{x - x_*}_{\ma \mb^\pseudo \ma}^2\,.
	\end{equation}
\end{lemma}

\begin{proof}
	Note that $b = \ma x_*$ by the assumption that $b \in \im(\ma)$ thereby proving $\eqref{eq:g_delta_formula}$. Further, by definition of an $\epsilon$-approximate solver and $\Delta$ we have
	\begin{equation}
	\label{eq:square_bound}
	\E \norm{\Delta}_{\mb}^2 
	\leq \epsilon \norm{\mb^{\pseudo} ( \ma x - b )}_{\mb}^2
	= \epsilon \norm{x - x_*}_{\ma \mb^\dagger \ma}^2 
	\leq \epsilon \norm{x - x_*}_{\ma}^2 
	\end{equation}
	where in the last step we used that since $\ma \preceq \mb$ we have $\mb^\pseudo \preceq \ma^\pseudo$ and 
	$\ma \mb^\dagger \ma \preceq \ma$. 
	Further, since $\E \norm{\Delta}_\mb \leq \sqrt{ \E \norm{\Delta}_\mb^2 }$ by concavity of $\sqrt{\cdot}$ we have
	\begin{align*}
	\E \norm{g}_{\mb}^2
	&=
	\E \norm{\Delta + \mb^\pseudo \ma ( x - x_* )}^2_\mb
	= \E \left[ \norm{\Delta}_\mb^2 + 2 \left[ \Delta^\top \mb \mb^\pseudo \ma(x - x_*) \right]  
	+ \norm{\mb^\pseudo \ma (x - x_*)}_\mb^2 \right] \\
	&\leq 
	\E \norm{\Delta}_\mb^2 + 2 \E \norm{\Delta}_\mb \norm{x-x_*}_{\ma \mb^\pseudo \ma}
	+ \norm{x - x_*}_{\ma \mb^\pseudo \ma}^2 \\
	&\leq \left[\epsilon + 2\sqrt{\epsilon} + 1 \right] \norm{x - x_*}_{\ma \mb^\pseudo \ma}^2 
	\end{align*}
	where we used Cauchy Schwarz for Euclidean semi-norms, i.e. $a^\top \mb b \leq \norm{a}_{\mb} \norm{b}_{\mb}$ and \eqref{eq:square_bound}. 
\end{proof}

Next we analyze the residual error decrease, i.e. change in $\norm{x - x_*}_\ma$, from taking a single gradient step using $\solve$. We will use this  to analyze the error in computing $x_{t+1}$ from $y_t$ using $\ApproxAGD$ (\Cref{alg:approx_agd}). 

\begin{lemma}
\label{lem:richardson}
Let $\ma, \mb \in \R^{n \times n}$ be symmetric PSD matrices with $\ma \preceq \mb \preceq \kappa \ma$ and for arbitrary $x$ and $b \in \mathrm{im}(\ma)$  let  $y = x - g$ where $g = \solve_{\mb}(\ma x - b)$  and  $\solve_{\mb}$ is an $\epsilon$-approximate linear system solver for $\mb$. $x_* \defeq \ma^\pseudo b$ satisfies
\[
\E \norm{y - x_*}_{\ma}^2
 \leq 
 \norm{x - x_*}_{\ma}^2
 - (1 - \epsilon) \norm{x - x_*}^2_{\ma \mb^\pseudo \ma}
\leq 
 \left(1 - \frac{1 - \epsilon}{\kappa} \right)
 \norm{x - x_*}_{\ma}^2
 ~.
\]
\end{lemma}

\begin{proof}
Let $x_* \defeq \ma^\pseudo b$ and $\Delta \defeq g - \mb^{\pseudo}(\ma x - b)$. By \Cref{lem:precondition_helper} we have
\[
y = x - \left( \mb^{\pseudo} \ma (x - x_*) + \Delta \right)
\enspace \text{ and } \enspace
\E \norm{\Delta}_{\mb}^2 
\leq \epsilon \norm{x - x_*}_{\ma \mb^\dagger \ma}^2 ~.
\]
Consequently, 
\begin{align*}
\norm{y - x_*}_{\ma}^2
&= \norm{x - x_*}_{\ma}^2 - 2 (x - x_*)^\top \ma \left[ \mb^\dagger \ma (x- x_*) + \Delta \right]
 + \norm{\mb^\dagger \ma (x - x_*) + \Delta}_{\ma}^2 ~.
\end{align*}
Since $\ma \preceq \mb$,
\[
\norm{\mb^\dagger \ma (x - x_*) + \Delta}_{\ma}^2
\leq \norm{\mb^\dagger \ma (x - x_*) + \Delta}_{\mb}^2
= \norm{x - x_*}_{\ma \mb^\dagger \ma}^2 + 2(x - x_*)^\top \ma \Delta
+ \norm{\Delta}_\mb^2 ~. 
\]
Now  $\mb \preceq \kappa \ma$ and therefore $\mb^\pseudo \succeq \kappa^{-1} \ma^\pseudo$ and $\ma \mb^\pseudo \ma \succeq \kappa^{-1} \ma$. Combining the above inequalities yields
\begin{align*}
\E \norm{y - x_*}_\ma^2
&\leq \norm{x - x_*}_\ma^2 - \norm{x - x_*}_{\ma \mb^\dagger \ma}^2
+ \E \norm{\Delta}_\mb^2 \\
&\leq 
 \norm{x - x_*}_\ma^2 - (1 - \epsilon) \norm{x - x_*}_{\ma \mb^\dagger \ma}^2
 \leq \left(1 - \frac{1 - \epsilon}{\kappa} \right) \norm{x - x_*}_\ma^2 ~.
 \end{align*}
\end{proof}

\begin{lemma}[Single Step Analysis]
\label{lem:agd_single_step}
 In the setting of \Cref{thm:agd} let $\epsilon_t \defeq \norm{x_t - x_*}_\ma^2$ and $r_t \defeq  \norm{v_t - x_*}_{\mb}^2$ for $x_* \defeq \ma^\pseudo b$. Conditioned on the value of $x_t$ and $v_t$ and considering the randomness in $g_t$ we have
\[
\E\left[\epsilon_{t + 1} + \frac{1}{2\kappa}r_{t + 1} \right]
\leq  
\left(1 - \frac{1}{4 \sqrt{\kappa}} \right)
\left[
\epsilon_t + \frac{1}{2 \kappa} r_t 
\right] ~.
\]
\end{lemma}

\begin{proof}
Throughout we let $z_t \defeq \beta v_t + (1 - \beta) y_t $, $\epsilon_t^y \defeq \norm{y_t - x_*}_{\ma}^2$ and $r_t^y \defeq \norm{y_t - x_*}_{\mb}^2$. Note that $\norm{z_t - x_*}_{\mb}^2 \leq \beta r_t + (1- \beta) r_t^y$ by the convexity of $\norm{\cdot}^2_{\mb}$ and the definition of  $z_t$ . Consequently, expanding the definition of $v_{t + 1}$ and applying that $\eta (1 - \beta) = \sqrt{\kappa}$ yields
\begin{align*}
r_{t + 1}
&= \norm{z_t - x_* - (1 - \beta)\eta g_t}_{\mb}^2 \\
&= \norm{z_t - x_*}_{\mb}^2 + (1 - \beta) \eta \left[- 2  g_t^\top \mb (z_t - x_*) + \eta
(1 - \beta) \norm{g_t}_{\mb}^2\right] \\
&\leq
\beta r_t + 
(1 - \beta) 
r_t^y
+ \sqrt{\kappa}
\left[ 
- 2 g_t^\top \mb (z_t - x_*)
+ \sqrt{\kappa}  \norm{g_t}_{\mb}^2
\right]
\end{align*}
Now \Cref{lem:precondition_helper} implies that
\[
g_t =  \mb^{\pseudo} \ma (y_t - x_*) + \Delta_t 
\enspace \text{  and  } \enspace 
\E \norm{\Delta_t}_{\mb}^2 
\leq \epsilon \norm{y_t - x_*}_{\ma}^2 
= \epsilon \cdot \epsilon_t^y
~,
\]
and the formulas for $y_t$ and $z_t$ imply that
\[
z_t =  \frac{\beta}{1 - \alpha} \left(y_t - \alpha x_t\right) + (1 - \beta) y_t
= y_t + \frac{\alpha \beta}{1 - \alpha} (y_t - x_t) ~.
\]
Combining yields that
\begin{align*}
g_t^\top \mb (z_t - x_*)
&= \left(\mb^\pseudo \ma (y_t - x_*)  \right)^\top \mb \left(y_t - x_* + \frac{\alpha \beta}{1 - \alpha} (y_t - x_t) \right)
+ \Delta_t^\top \mb (z_t - x_*)
 \\
&= \norm{y_t - x_*}_{\ma}^2 + \Delta_t^\top \mb (z_t - x_*) + \frac{\alpha \beta}{1 - \alpha }
\left[ (y_t - x_*)^\top \ma (y_t - x_t)
 \right] ~.
\end{align*}
Further, since
\[
\norm{x_t - x_*}_\ma^2 = \norm{y_t - x_*}_\ma^2 + 2(y_t - x_*)^\top \ma (x_t - y_t) + \norm{y_t - x_t}_\ma^2 ~.
\]
we have that
\begin{align*}
- 2 g_t^\top \mb (z_t - x_*)
&\leq - 2 \epsilon_t^y - 
2 \Delta_t^\top \mb (z_t - x_*) 
+ \frac{\beta \alpha}{1 - \alpha} \left[\epsilon_t  - \epsilon_t^y
\right]
\end{align*}
Now since $\E \norm{\Delta_t}_{\mb}^2 \leq \epsilon \cdot \epsilon_t^y$, the concavity of $\sqrt{\cdot}$ yields 
\begin{align*}
\E \left[- \Delta_t^\top  \mb (z_t - x_*)\right]
&\leq  
\E \norm{\Delta_t}_{\mb} \norm{z_t - x_*}_{\mb}
\leq 
\sqrt{\E \norm{\Delta_t}_\mb^2} \norm{z_t - x_*}_{\mb}
\\
&\leq \sqrt{\epsilon \cdot \epsilon_t^y} \left(\beta \norm{v_t - x_*}_\mb + (1 - \beta) \norm{y_t - x_*}_\mb \right) \\
&\leq  \sqrt{\epsilon \cdot \epsilon_t^y r_t} + \sqrt{\epsilon \cdot \kappa} (1 - \beta) \epsilon_t^y \\
&\leq \frac{r_t}{2} \sqrt{\frac{\epsilon}{\kappa}}  
+ \frac{\sqrt{\epsilon}}{2} \left( \sqrt{\kappa}  + 1 \right) \epsilon_t^y  
\leq \frac{r_t}{2} \sqrt{\frac{\epsilon}{\kappa}}   + \frac{1}{4} \epsilon_t^y
~.
\end{align*}
where in the second to last line we used that $r_t^y \leq \kappa \epsilon_t^y$ and $\beta \leq 1$ and in the last line we used that $\sqrt{a b} \leq \frac{a}{2p} + \frac{bp}{2}$ for all $a,b \in \R$ and $p > 0$, that $\sqrt{\kappa} (1 - \beta) = \frac{1}{2}$, and
$\sqrt{\epsilon} ( \sqrt{\kappa} + 1 ) \leq \frac{1}{2}$.  Combining, and again using that $r_t^y \leq \kappa \epsilon_t^y$ yields
\begin{align*}
\E \left[ \frac{(1 - \beta)}{\sqrt{\kappa}} \cdot r_t^y - 2 g_t^\top \mb (z_t - x_*) \right]
&\leq
\frac{\kappa}{2 \kappa} \cdot
 \epsilon_t^y 
- 2 \epsilon_t^y + 
 \left[ 
r_t \sqrt{\frac{\epsilon}{\kappa}}   + \frac{1}{2} \epsilon_t^y
\right]
+\frac{\beta \alpha }{1 - \alpha} [\epsilon_t - \epsilon_t^y] \\
&\leq
-  \epsilon_t^y + \frac{\beta \alpha }{1 - \alpha} [\epsilon_t - \epsilon_t^y] 
+ r_t \sqrt{\frac{\epsilon}{\kappa}}   ~.
\end{align*}
Further, by \Cref{lem:precondition_helper} and \Cref{lem:richardson} and that $(1 + \sqrt{\epsilon})^2 / (1 - \epsilon) \leq 2$ for $\epsilon \leq 1/10$
\[
\E \norm{g_t}_{\mb}^2 \leq (1 + \sqrt{\epsilon})^2 \norm{y_t - x_*}_{\ma \mb^\pseudo \ma}^2 
\leq \frac{ (1 + \sqrt{\epsilon})^2}{1 - \epsilon} \left[ \epsilon_t^y - \E \epsilon_{t + 1} \right]
\leq 2 [ \epsilon_t^y - \E \epsilon_{t + 1} ]
\]
Combining then yields that
\begin{align*}
\E r_{t + 1}
&\leq
\beta r_t + \sqrt{\kappa}
\left[
 r_t \sqrt{\frac{\epsilon}{\kappa}}
- \epsilon_t^y  
+  \frac{\beta \alpha}{1 - \alpha} [\epsilon_t  - \epsilon_t^y]
+ 2 \sqrt{\kappa} [\epsilon_t^y - \epsilon_{t + 1}]
\right]
\\
&\leq
\left(
\beta + \sqrt{\epsilon}
\right)
r_t + \sqrt{\kappa} 
\left[
- \epsilon_t^y  
+ 2  \beta \sqrt{\kappa} [\epsilon_t  - \epsilon_t^y]
+ 2 \sqrt{\kappa} [\epsilon_t^y - \E\epsilon_{t + 1}]
\right] \\
&=
\left(
\beta +\sqrt{ \epsilon}
\right) r_t
+ 2 \kappa \left[ 
\beta \epsilon_t - \E\epsilon_{t + 1}
\right] ~:
\end{align*}
in the second line we used that $\alpha$ was chosen so that $\frac{\alpha}{1 - \alpha} = 2 \sqrt{\kappa}$, and in the third line we used that $\beta = 1 -  \frac{1}{2\sqrt{\kappa}}$ implying $-1 - 2\beta \sqrt{\kappa} + 2\sqrt{\kappa} = 0$. Since $\beta \leq \beta + \sqrt{ \epsilon} \leq 1 - \frac{1}{4\sqrt{\kappa}}$ rearranging yields the desired bound.
\end{proof}

Leveraging the preceding analysis we can now prove the theorem.

\begin{proof}[Proof of \Cref{thm:agd}]
Applying \Cref{lem:agd_single_step} repeatedly we have that for $x_* = \ma^\pseudo b$
\[
\E \norm{x_T - x_*}^2_\ma \leq \left(1 - \frac{1}{4 \sqrt{\kappa}} \right)^T \left[
\norm{x_0 - x_*}_\ma^2 + \frac{1}{2\kappa }\norm{x_0 - x_*}_\mb^2
\right]
\]
However, since $x_0 = 0$ and $\mb \preceq {\kappa} \ma$ we have 
\[
\norm{x_0 - x_*}_\ma^2 + \frac{1}{2\kappa }\norm{x_0 - x_*}_\mb^2 \leq \frac{3}{2} \norm{\ma^\pseudo b}_\ma^2
= \frac{3}{2} \norm{b}_{\ma^\pseudo}^2 ~.
\]
Further, by choice of $T$ we have
\[
\left(1 - \frac{1}{4 \sqrt{\kappa}} \right)^T
\leq \exp\left(\frac{-T}{4 \sqrt{\kappa}}\right)
\leq \frac{\epsilon}{2} ~.
\]
The  result follows by combining these inequalities and applying the definition of an $\epsilon$-solver and noticing that the iterations consist only of standard arithmetic operations of vector and applying $\solve_{\mb}$ and applying $\ma$ to a vector.
\end{proof}

\section{Effective Resistance Facts}
\label{sec:effres_facts}

Here we give a variety of facts about effective resistance that we use throughout the paper. First, in the following claim we collect a variety of well known facts about effective resistance that we use throughout the paper and then we give additional technical lemmas we use throughout the paper.

\begin{claim}[Effective Resistance Properties]
\label{effres_props}
For any connected graph $G = (V, E)$ with positive edge weights $w \in \R^E$ and all $a,b,c \in V$ it is the case that
\begin{itemize}
	\item \textbf{Flow Characterization}: $\effres_G(a,b) = \min_{\text{unit } a,b \text{ flow } f \in \R^E} \sum_{e \in E} f_e^2 / w_e$.
	\item \textbf{Triangle Inequality}: $\effres_G(a,c) \leq \effres_G(a,b) + \effres_G(b,c)$.
	\item \textbf{Monotonicity}: If $H$ is an connected edge subgraph of $G$ then
 	 $\effres_{G}(a,b) \leq \effres_{H}(a,b)$.
\end{itemize}
\end{claim}

It is a well known fact that the effective resistance between two vertices $s$ and $t$ in a graph consisting $k$ edge-disjoint parallel paths between $s$ and $t$ of length at most $\ell$ is $\ell / k$. Here we give a slight generalization of this fact to bound the effective resistance of two vertices in low-depth trees connected by many short edge-disjoint paths.

\begin{lemma}[Effective Resistance in Well-connected Trees]
\label{lem:treefact}
Let $G = (V, E, w)$ be a weighted unweighted graph which contains two edge disjoint trees $T_1, T_2 \subseteq E$ each of which has effective resistance diameter at most $d$, i.e. the effective resistance between any pair of vertices in a tree is at most $d$. Further suppose that there are at least $k$ edge-disjoint paths between $T_1$ and $T_2$ each of which have effective resistance length at most $\ell$, i.e. for path $P \subseteq E$ we have $\sum_{e \in P} (1/w_e) \leq \ell$. Then for all $a \in T_1$ and $b\in T_2$ we have $\effres_G(a,b) \leq 2d + \ell/k$.
\end{lemma}

\begin{proof}
For each of the $k$ edge-disjoint path $P_i \subseteq \R^E$ let $f_i$ denote a flow that send unit from the path's start in $T_1$, denoted $a_i$, to the path's end in $T_2$, denoted $b_i$. Further, for all $i \in [k]$ let $g_i$ denote the unique unit flow from $a$ to $a_i$ using only edges of $T_1$ and let $h_i$ denote the unique unit flow from $b_i$ to $b$ using only edges of $T_2$. Note, that for all $i \in [k]$ we have that $r_i \defeq f_i + g_i + h_i$ is a unit $a$ to $b$ flow in $G$ and consequently, $f_* \defeq \frac{1}{k} \sum_{i \in [k]} r_i$, is a unit $a$ to $b$ flow in $G$. 

Now, $f_*$ restricted to $T_1$ is the unique flow $f_1$ on $T_1$ that sends one unit from $a$ to the uniform distribution over the $a_i$. Since $T_1$ has effective resistance diameter at most $d$, by the flow characterization of effective resistance (\Cref{lem:treefact}) we can decompose this flow into a distribution over paths of effective resistance length at most $d$, i.e. $f_1 = \sum_{i} \alpha_i t_i$ where each $\alpha_i \geq 0$, $\sum_{i} \alpha_i = 1$, and each $t_i$ is a unit flow corresponding to a path $P$ of effective resistance length at most $d$. Therefore, by convexity of $x^2$ we have
\begin{align*}
\sum_{e \in T_1}  \frac{1}{w_e} [f_1]_e^2 
&= \sum_{e \in T_1} \frac{1}{w_e} \left[\sum_{i} \alpha_i [t_i]_e  \right]^2
\leq 
\sum_{e \in T_1} \frac{1}{w_e} \sum_{i} \alpha_i [t_i]_e^2 \\
&= \sum_{i}  \alpha_i \sum_{e \in T_1} \frac{1}{w_e}  [t_i]_e^2 
 \leq \sum_{i} \alpha d = d\,.
\end{align*}
By symmetric reasoning, $f_*$ restricted to $T_2$, denoted $f_2$ has $\sum_{e \in T_2} [f_2]_e^2 / w_e \leq d$. 

Since $T_1$, $T_2$, and the $P_i$ are edge disjoint, are the only edges with non-zero flow, and have effective resistance length at most $\ell$ we have
\begin{align*}
\sum_{e \in E}  \frac{1}{w_e}  [f_*]_e^2
&=  \sum_{e \in T_1} \frac{1}{w_e}  [f_1]_e^2
+  \sum_{e \in T_2} \frac{1}{w_e}  [f_2]_e^2
+ \sum_{i \in [k]} \sum_{e \in P_i} \frac{1}{w_e}  [f_*]_e^2 \\
&\leq 2d + \sum_{i \in [k]} \sum_{e \in P_i} \frac{1}{w_e} \cdot \frac{1}{k^2} 
\leq 2d + \ell/k ~.
\end{align*}
The result follows as $\effres_G(a,b) \leq \sum_{e \in E}  \frac{1}{w_e}  [f_*]_e^2$  by the flow characterization of effective resistances, \Cref{effres_props}.
\end{proof}

\end{document}